\documentclass[11pt,a4paper]{amsart}
\usepackage[foot]{amsaddr}
\usepackage{ifxetex}
\ifxetex
  \usepackage[no-math]{fontspec}
\else
\fi
\usepackage{tikz}
\usetikzlibrary{arrows.meta, positioning, shapes, calc}
\usepackage{amsmath}
\usepackage{amsfonts}
\usepackage{amssymb}
\usepackage{amsthm}

\usepackage{fullpage}
\usepackage{microtype}

\ifxetex 
  \usepackage[libertine]{newtxmath}
\else
  \usepackage{newtxmath}
\fi
\usepackage[tt=false]{libertine} 
\usepackage{caption}
\usepackage{tcolorbox}
\usepackage{xspace}
\usepackage{bbm}
\usepackage{hyperref, color}
\hypersetup{colorlinks=true,citecolor=blue, linkcolor=blue, urlcolor=blue}
\usepackage[linesnumbered,boxed,ruled,vlined]{algorithm2e}
\usepackage{bm}
\usepackage{bbm}
\usepackage[numbers]{natbib}
\usepackage{xcolor}
\usepackage{enumerate} 
\usepackage{enumitem}
\usepackage{ragged2e}
\usepackage{tabularx}
\usepackage{array}
\usepackage{longtable}
\usepackage{hhline}
\usepackage{multirow}
\usepackage{cleveref}

\newcolumntype{L}[1]{>{\raggedright\arraybackslash}p{#1}}
\newcolumntype{C}[1]{>{\centering\arraybackslash}m{#1}}
\newcolumntype{R}[1]{>{\raggedleft\arraybackslash}p{#1}}

\usepackage{makecell}

\usepackage{footnote}
\usepackage{float}
\makesavenoteenv{tabular}

\renewcommand{\epsilon}{\varepsilon}

\newtheorem{theorem}{Theorem}[section]

\newtheorem*{claim*}{Claim}
\newtheorem{condition}[theorem]{Condition}

\newtheorem{lemma}[theorem]{Lemma}

\newtheorem{corollary}[theorem]{Corollary}
\theoremstyle{definition}

\newtheorem{definition}[theorem]{Definition}
\newtheorem{remark}[theorem]{Remark}
\newtheorem*{remark*}{Remark}

\renewcommand{\Pr}[2][]{ \ifthenelse{\isempty{#1}}
  {\mathop{\mathbf{Pr}}\left[#2\right]} {\mathop{\mathbf{Pr}}_{#1}\left[#2\right]} }

\newcommand{\abs}[1]{\left\vert#1\right\vert}
\newcommand{\set}[1]{\left\{#1\right\}}
\newcommand{\tuple}[1]{\left(#1\right)} 
\newcommand{\inbr}[1]{\left[#1\right]}

\newcommand{\eps}{\varepsilon}
\newcommand{\tp}{\tuple}

\newcommand{\defeq}{\triangleq}
\newcommand{\supp}{\mathrm{supp}}

\def\*#1{\bm{#1}} 
\def\+#1{\mathcal{#1}} 
\def\-#1{\mathrm{#1}} 
\def\=#1{\mathbb{#1}} 
\def\`#1{\mathfrak{#1}}

\usepackage[textsize=tiny]{todonotes}

\newcommand{\yxtodo}[1]{\todo[color=gray!30]{yyx: #1}}
\usepackage{xifthen}

\makeatletter
\def\prob#1#2#3{\goodbreak\begin{list}{}{\labelwidth\z@ \itemindent-\leftmargin
                        \itemsep\z@  \topsep6\p@\@plus6\p@
                        \let\makelabel\descriptionlabel}
                \item[\it Name]#1
               \item[\it Instance]                #2
                \item[\it Output]#3
                \end{list}}
\makeatother

\newcommand{\upd}{\textnormal{\textsf{pred}}}

\newcommand{\ts}{\textnormal{\textsf{TS}}}

\newcommand{\vbl}{\mathsf{vbl}}

\newcommand{\cbad}{\+C^{\-{bad}}}
\newcommand{\tbad}{\+T^{\-{bad}}}
  
\usepackage{enumitem}

\newcommand{\dist}{\mathrm{dist}}

\newcommand{\Lovasz}{Lov\'asz\xspace}

\newcommand{\BadTS}[1]{\textnormal{\textsf{BadTS}}\tuple{#1}} 

\newbool{doubleblind}
\setbool{doubleblind}{true}

\title{Phase Transitions via Complex Extensions of Markov Chains}

\author{Jingcheng Liu, Chunyang Wang, Yitong Yin, Yixiao Yu}
\address{State Key Laboratory for Novel Software Technology, New Cornerstone Science Laboratory, Nanjing University, 163 Xianlin Avenue, Nanjing, Jiangsu Province, 210023, China. \textnormal{E-mail: \texttt{liu@nju.edu.cn}, \texttt{wcysai@smail.nju.edu.cn}, \texttt{yinyt@nju.edu.cn}, \texttt{yixiaoyu@smail.nju.edu.cn}}. JL is supported by the National Science Foundation of China under Grant No. 62472212.}

\begin{document}

\begin{abstract}
We study algebraic properties of partition functions, particularly the location of zeros, through the lens of rapidly mixing Markov chains. The classical Lee-Yang program initiated the study of phase transitions via locating complex zeros of partition functions. Markov chains, besides serving as algorithms, have also been used to model physical processes tending to equilibrium. In many scenarios, rapid mixing of Markov chains coincides with the absence of phase transitions (complex zeros). Prior works have shown that the absence of phase transitions implies rapid mixing of Markov chains. We reveal a converse connection by lifting probabilistic tools for the analysis of Markov chains to study complex zeros of partition functions.

Our motivating example is the independence polynomial on $k$-uniform hypergraphs, where the best-known zero-free regime has been significantly lagging behind the regime where we have rapidly mixing Markov chains for the underlying hypergraph independent sets. Specifically, the Glauber dynamics is known to mix rapidly on independent sets in a $k$-uniform hypergraph of maximum degree $\Delta$ provided that $\Delta \lesssim 2^{k/2}$. On the other hand, the best-known zero-freeness around the point $1$ of the independence polynomial on $k$-uniform hypergraphs requires $\Delta \le 5$, the same bound as on a graph. 

By introducing a complex extension of Markov chains, we lift an existing percolation argument to the complex plane, and show that if $\Delta \lesssim 2^{k/2}$, the Markov chain converges in a complex neighborhood, and the independence polynomial itself does not vanish in the same neighborhood. In the same regime, our result also implies central limit theorems for the size of a uniformly random independent set, and deterministic approximation algorithms for the number of hypergraph independent sets of size $k \le \alpha n$ for some constant $\alpha$.
\end{abstract}	

\maketitle

\setcounter{tocdepth}{1}
\tableofcontents
\setcounter{page}{0}
\clearpage

\section{Introduction}

More than a few important recent advances in theoretical computer science, in combinatorics and probability theory, have been made possible through locating the zeros of suitably chosen multivariate polynomials.
These include improved approximation algorithms for the traveling salesman problem~\cite{gharan2011randomized,karlin2021slightly}, construction of Ramanujan graphs of every degree~\cite{marcus2015interlacing1,marcus2015interlacing2}, deterministic approximate counting algorithms for spin systems~\cite{barvinok2016combinatorics,patel2017deterministic,Liu2017TheIP,liu2019correlation}, an algebraic proof of a generalization of the van der Waerden Conjecture~\cite{gurvits2006van}, a resolution of the long-standing Kadison-Singer conjecture~\cite{marcus2018interlacing}, and notably the theory of negatively dependent random variables~\cite{Borcea2007NegativeDA}. 
Furthermore, there has been a fruitful line of work that exploits a more general form of geometry, notably the development of log-concave polynomials and Lorentzian polynomials, which have led to novel analyses of Markov chains and the resolution of Mason's conjecture~\cite{anari2018log,branden2020lorentzian}. %Arguably, a crucial pillar in these developments is the generalization of stability-preserving linear operators.

The development of multivariate stability theory dates back to the famous Lee-Yang program~\cite{Lee1952StatisticalTO} in statistical physics. In their seminal work, Lee and Yang initiated the study of phase transitions through the location of complex zeros of the partition function while also establishing identities relating key physical quantities to the density function of zeros.
A key insight is that to understand the macroscopic properties of a system at the thermodynamic limit (that is, as the size of the system tends to infinity), one studies the complex zeros in a neighborhood for any finite systems so as to determine whether the quantities of interest remain analytic or can have a discontinuity. One key quantity of particular interest is the so-called free-energy density.
There have also been various generalizations and extensions of Lee-Yang type theorem in statistical physics and combinatorics~\cite{lieb1981general,Heilmann1972TheoryOM,wagner2009weighted},  Chernoff bounds~\cite{kyng2018matrix}, asymptotic normality~\cite{kahn2000normal} and central limit theorems~\cite{lebowitz2016central,michelen2019central,vishesh2022approximate}. Stability theory for a univariate polynomial is also extensively studied in control theory and can be traced back to the famous Routh-Hurwitz criterion~\cite{routh1877treatise,hurwitz1895conditions}. %in electrical control theory~\cite{fettweis1985discrete} where complex numbers naturally have physical interpretations,

Roughly speaking, a phase transition occurs when the macroscopic property of a system is not fully determined by local interactions in the thermodynamic limit (that is, there could be multiple phases). To formalize such a notion, three types of mathematical definitions have been studied: %in the context of spin systems:
\begin{enumerate}
    \item \emph{Probabilistic:} Conditions under which a Gibbs distribution exhibits decay of long-range correlations with respect to distance.
    \item \emph{Algebraic:} Conditions under which a partition function vanishes in the thermodynamic limit. This is also Lee-Yang's view of phase transition.
    \item \emph{Algorithmic:} Conditions under which a spin system out-of-equilibrium quickly returns to thermal equilibrium; in particular, when does a Glauber dynamics mix rapidly to the Gibbs distribution.
\end{enumerate}

Notably, Glauber dynamics can be seen as both a model of physical processes tending to equilibrium, and also an algorithm that can be efficiently simulated.
To this date, each of these distinct-looking definitions has seen fruitful algorithmic applications,  giving rise to algorithms based on the decay of correlations~\cite{weitz06counting}, the absence of zeros~\cite{barvinok2016combinatorics}, and  the direct simulation of Glauber dynamics. 
Numerous efforts have been made to understand the relationship between these three definitions and their relative strengths.
For amenable graphs such as lattices, 
%the classical work of 
Dobrushin and Shlosman~\cite{dobrushin1985completely,dobrushin1987completely} studied the first two types in the form of \emph{complete analyticity} and showed that they are equivalent.
 Stroock and Zegarlinski~\cite{stroock1992logarithmic} showed the equivalence of all three types via log-Sobolev inequalities.
These analyses crucially rely on the amenability of the lattices.

In more general settings, Barvinok~\cite{barvinok2019personal} posed an open question concerning establishing the absence of zeros from the analysis of \emph{any} rapidly mixing Markov chain. 
A key challenge,  as pointed out by Barvinok, is that while an inverse polynomial spectral gap is sufficient to prove the rapid mixing of Markov chains,
a constant radius of zero-free region is often desired for practical applications. 
Until now, little progress has been made in this specific direction.
In contrast, the other direction has seen more success.
Assuming decay of correlation in the form of \emph{contraction}, the absence of zeros follows from the contraction method~\cite{peter2019conjecture,liu2019correlation,Shao2019ContractionAU}.
Furthermore, \cite{anari2020spectral,chen2020rapid} showed that contraction also implies rapid mixing of Glauber dynamics. 
Additionally, the absence of zeros has been shown to imply the decay of correlations (in the form of strong spatial mixing) for self-reducible problems~\cite{gamarnik2023correlation,regts2023absence} and to imply rapid mixing of Glauber dynamics~\cite{Alimohammadi2021FractionallyLA,Chen2021SpectralIV} through a different form of correlation decay known as \emph{bounded total influence}.  Moreover, rapid mixing is known to imply spectral independence~\cite{anari2024universality}, which can be seen as a form of bounded correlations.

We give a rough summary of the state-of-the-art in~\Cref{fig:three-phase-transitions}. 

\begin{figure}[H]
    \centering
    \hspace{.5cm}
    \begin{tikzpicture}[node distance=3cm, thick][shift={(10,10)}]
        \tikzset{>={Latex[width=3mm,length=3mm]}}
          \node (zero) [draw, rounded corners, minimum width=2.5cm, minimum height=0.8cm, font=\large] {zero-freeness};
          \node (mixing) [draw, rectangle, below left=of zero, rounded corners, minimum width=2.5cm, minimum height=0.8cm, font=\large] {rapid mixing};
          \node (decay) [draw, rectangle, below right=of zero, rounded corners, minimum width=2.5cm, minimum height=0.8cm, font=\large] {decay of correlations};
        
          \draw[->] (zero) to[bend left=10] node[above, xshift=0.5cm] {\tiny{\cite{gamarnik2023correlation,regts2023absence}}} (decay);
          \draw[->] (decay) to[bend left=10] node[below, xshift=-0.8cm] {\tiny{\cite{peter2019conjecture,liu2019correlation,Shao2019ContractionAU}}} (zero);
          \draw[->] (decay) to[bend left=10] node[below] {\tiny{\cite{anari2020spectral, chen2020rapid}}} (mixing);
          \draw[->] (mixing) to[bend left=10] node[below] {\tiny{\cite{anari2024universality}}} (decay);
          \draw[->] (zero) to[bend left=10]  node[above, xshift=+1.5cm] {\tiny{\cite{Alimohammadi2021FractionallyLA,Chen2021SpectralIV}}} (mixing);
          \draw[->,color=red, dashed] (mixing) to[bend left=10]  node[above, xshift=-1.4cm] {This work} (zero);
    \end{tikzpicture}
    \caption{A rough summary of connections between three types of phase transitions. We do not distinguish the exact form of phase transitions within each type.}
    \label{fig:three-phase-transitions}
\end{figure}
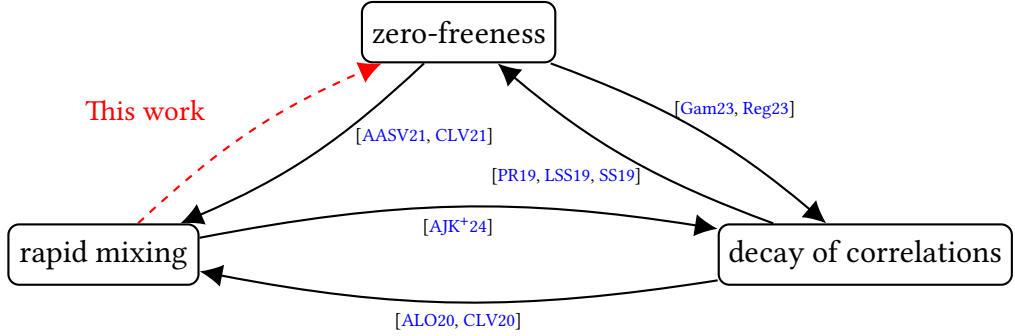

\subsection{Hypergraph independence polynomial}
%Can existing
Our motivating example is the independence polynomial on a $k$-uniform hypergraph.
Given a hypergraph $H=(V,\+E)$, we say that $H$ is $k$-uniform if every hyperedge $e\in \+E$ has size $\abs{e}=k$.
The maximum degree $\Delta$ is the maximum number of hyperedges incident to a vertex.
An \emph{independent set} in $H$ is a subset $S\subset V$ of vertices that does not contain any $e\in E$, that is, every hyperedge $e$ must have at least one endpoint not chosen by $S$. We use $\sigma \in \set{0,1}^n$ to indicate the set $S$, meaning that $\sigma(v)=1$ iff $v \in S$. Let $\+I(H)$ denote the set of independent sets in $H$. Then, the independence polynomial of $H$ is a generating polynomial in the variables $\*\lambda$:
\[
Z_H(\*\lambda) = \sum_{\sigma \in \+I(H)} \prod_{v: \sigma(v)=1} \lambda_v. %Z^\-{ly}_H(\*\lambda) 
\]
When $k=2$, this is the standard independence polynomial, which has been studied in many branches of mathematics, physics, and computer science.
To name a few, Shearer~\cite{shearer85} obtained instance-optimal sufficient criterion in \Lovasz{} local lemma using the largest root of the independence polynomial; a tight runtime analysis of the celebrated Moser-Tardos algorithm for the algorithmic \Lovasz{} local lemma is characterized by the independence polynomial~\cite{KS11}; independence polynomial is also known as the hardcore model for equilibrium of lattice-gas in statistical physics~\cite{Scott2003TheRL}; it is also the first example where a sharp computational complexity of approximate counting and sampling is known~\cite{weitz06counting,Sly2010computation}.
We will refer to the complex zeros of the independence polynomial in $\*\lambda$ as \emph{Lee-Yang} zeros, as $\*\lambda$'s are playing the role of external fields here.
Henceforth, we denote the above polynomial by $Z^\-{ly}_H(\*\lambda)$.

Scott and Sokal also proposed a soft-core version of independence polynomial~\cite{Scott2003TheRL}, with which they derived a weak dependency version of the local lemma.
This inspired us to study a soft-core independence polynomial parameterized by the interactions:
\[
Z_H^\-{fs}(\*\beta) =  \sum\limits_{S \subset V} \prod\limits_{\substack{e: e \subset S}} \beta_e.
\]
Intuitively, for every hyperedge $e$ completely contained in a set $S$, we assign a ``penalty'' $\beta_e$ for violating the ``hard-core'' constraint.
Zeros in the interaction parameter are also known as \emph{Fisher zeros}~\cite{fisher1965nature}. Compared to Lee-Yang zeros that have been extensively studied, general results for Fisher zeros have been limited until the recent introduction of contraction method~\cite{peter2019conjecture,liu2019correlation,Shao2019ContractionAU}. However, these contraction methods crucially rely on a self-avoiding walk construction, which breaks up the uniformity of hypergraphs and is therefore ill-suited for our purpose. 

The point $\*\lambda=1$ in $Z^\-{ly}_H(\*\lambda)$ and the point $\*\beta=0$ in $Z_H^\-{fs}(\*\beta)$ are of particular interests, as they correspond to the uniform enumeration of independent sets.
This is a prominent
example of models where the known regime of zero-freeness, corresponding to an algebraic phase transition, has significantly lagged behind the known regime where
efficient algorithms are available.
%This is arguably the most prominent example of the regime known as zero-freeness, which has significantly lagged behind the regime where efficient algorithms are available.
%The independence polynomial on $k$-uniform hypergraphs was
As one of the early examples where approximate counting and sampling algorithms were devised for a constraint satisfaction problem under a local lemma type condition, \cite{HSZ19} showed that the Glauber dynamics on independent sets of $k$-uniform hypergraph mixes rapidly provided the maximum degree $\Delta \lesssim 2^{k/2}$, and this is matched, up to the leading constants, by an earlier NP-hardness result for $\Delta \gtrsim 2^{k/2}$ in~\cite{NPH}.
Since then, several developments have followed suit, including perfect samplers \cite{HSW21,qiu2022perfect} and a local sampler \cite{feng2023towards}, all within similar regimes, albeit with poly$(k)$ factors.
Remarkably, the latter can be derandomized to yield a deterministic approximate counting algorithm.

For zero-freeness results, however,  progress has lagged significantly despite numerous efforts.
While there is a rich literature on models with pairwise interactions, %such as independence polynomial on graphs. 
%However, 
understanding of the more physically  relevant regime involving higher-order interactions remains limited, 
and techniques for locating complex zeros for higher-order interactions are less developed. %(See~\cite{brydges1984short,galvin2024zeroes} for a discussion of the difficulties). 
Only recently, Galvin, McKinley, Perkins, Sarantis and Tetali~\cite{galvin2024zeroes} established the existence of a zero-free disk for $\lambda$ centered at the origin with radius $\approx \frac{1}{\mathrm{e}\Delta}$ for hypergraphs of maximum degree $\Delta$.  
Later, Bencs and Buys~\cite{bencs2023optimal}  improved this result to match Shearer's bound for the independence polynomial on graphs. 
%\cite{galvin2024zeroes} also conjectured a larger zero-free disk for linear hypergraphs but was subsequently disproved by \cite{Zhang2023ANO}. 
For zero-freeness in a complex neighborhood around the positive real axis, 
it implicitly follows  from the lifting paradigm of~\cite{liu2019correlation,Shao2019ContractionAU} applied to the contraction method of~\cite{liu2015fptas,lu2016fptas}, that zero-freeness holds around $\*\lambda=1$ for $\Delta \le 5$.
This is also carried out more explicitly by~\cite{li2021complex,bencs2023optimal}. %In particular, this implies that if $\Delta \le 5$, then $Z^\-{ly}_H(\*\lambda)$ does not vanish around a complex neighborhood of $\*\lambda=1$. %and $Z_H^\-{fs}(\*\beta)$ does not vanish around the point $\*\beta=0$.
Despite these advances, a significant gap remains compared to the algorithmic transition, which holds up to $\Delta \lesssim 2^{k/2}$. Essentially, existing techniques for proving zero-freeness are insufficient to exploit the uniformity of hyperedges. 

\subsection{Our contributions}
We demonstrate how one can establish the absence of complex zeros through a powerful probabilistic tool in the analysis of Markov chains: percolation applied to a complex extension of Markov chains. %where we allow complex transition weights.
We show zero-free regions for both the Lee-Yang zeros of $Z^\-{ly}_H(\*\lambda)$, and the Fisher zeros of $Z_H^\-{fs}(\*\beta)$, in regimes that match the algorithmic transition of $\Delta \lesssim 2^{k/2}$ up to poly$(k)$ factors.
We also show convergence of a systematic scan Glauber dynamics with \emph{complex transition weights} in the same regime. To the best of our knowledge, this is the first such result for a complex dynamics.

%Information percolation
%Question: other analytical properties from Markov chain? e.g. log concavity?

\begin{theorem}[Lee-Yang zeros of hypergraph independence polynomial]\label{theorem:zero-freeness-his-lee-yang}
Fix $k\ge 2$ and $\Delta\ge 3$. 
Let $0<\epsilon<\frac{1}{9k^5\Delta^2}$,
and let $[0,\lambda_{c,\varepsilon})$ be the real segment such that the following holds for all $\lambda\in [0,\lambda_{c,\varepsilon})$:

\[
\left(\frac{\lambda+\varepsilon}{1+\lambda-\varepsilon}\right)^{k/2}<\frac{1}{2\sqrt{2}\mathrm{e}\Delta k^2}.
\]

Let $\+D_{\varepsilon}$ be the union of $\eps$-balls around the segment given by $\+D_{\varepsilon}=\left\{z\in\mathbb{C}\mid \exists \lambda\in[0,\lambda_{c,\varepsilon})\text{ s.t. }|z-\lambda|\le \varepsilon \right\}$.

Then, for any $k$-uniform hypergraph $H=(V,\+E)$ with maximum degree $\Delta$, the partition function $Z^\-{ly}_H(\*\lambda)$ is non-zero for all $\*\lambda \in \+D^V_{\varepsilon}$, i.e.~$\forall \*\lambda \in \+D^V_{\varepsilon}, Z^\-{ly}_H(\*\lambda)\neq 0$.

\end{theorem}
As a corollary, there is no phase transition in $\lambda \in [0,1]$ up till the ``algorithmic transition'' at $\Delta \lesssim 2^{k/2}$: 
\begin{corollary}
    Let $\delta>0$, $k \ge 2$ and $\Delta\ge3$ be constants with $\Delta\le \frac{1-\delta}{2\sqrt{2}\mathrm{e} k^2} \cdot 2^{\frac{k}{2}}$.
    For any $k$-uniform hypergraph $H=(V,\+E)$ with maximum degree $\Delta$, $Z^\-{ly}_H(\lambda)\neq 0$  around an open strip containing $[0,1]$.
\end{corollary}

This is a significant improvement on~\cite{galvin2024zeroes} for $k$-uniform hypergraphs. 

We also prove a Fisher zero-free region.

\begin{theorem}[Fisher zeros of hypergraph independence polynomial]\label{theorem:zero-freeness-his-fisher}
Fix $k\ge 2$ and $\Delta\ge 3$. Let $0<\epsilon < \frac{1}{16(k+1)^5\Delta^2}$ and $\+D_{\varepsilon}=\left\{z\in\mathbb{C}\mid \exists \beta\in[0,1]\text{ s.t. }|z-\beta|\le \varepsilon\right\}$ be the union of $\eps$-balls around $[0,1]$. 
If the following condition holds:
\[
\sqrt{1+2\varepsilon}2^{-k/2}<\frac{1}{2\sqrt{2}\mathrm{e}\Delta (k+1)^2},
\]
then  for any $k$-uniform hypergraph $H=(V,\+E)$ with maximum degree $\Delta$, the partition function $Z^\-{fs}_H(\*\beta)$ is non-zero for all $\*\beta\in \+D^{\+E}_{\varepsilon}$, i.e.~$\forall \*\beta \in \+D^{\+E}_{\varepsilon}, Z^\-{fs}_H(\*\beta)\neq 0$.
\end{theorem}

\begin{remark}[implications for deterministic counting]
FPTASes for the partition functions $Z^\-{ly}_{H}(\*\lambda)$ and $Z^\-{fs}_H(\*\beta)$ can be derived from the above zero-freeness results by applying Barvionk's interpolation method~\cite{barvinok2016combinatorics,patel2017deterministic,Liu2017TheIP}. 
The cumulants, such as the average size and variance of a random independent set, can also be approximated through a similar interpolation~\cite{vishesh2022approximate}. 
We note that in the regime of \Cref{theorem:zero-freeness-his-lee-yang} specifically at the point $\*\lambda=1$ 
(or in the regime of \Cref{theorem:zero-freeness-his-fisher} at the point $\*\beta=0$),
an FPTAS for the partition function that counts the number of independent sets in the hypergraph $H$ has already been found in \cite{feng2023towards}, 
 utilizing a different approach based on derandomization.
However, as showcased by the following examples, zero-freeness have much broader applications beyond deterministic approximation of partition functions.
\end{remark}

Through the well-known connection between central limit theorems and the zero-freeness of univariate polynomials \cite{lebowitz2016central,michelen2019central,vishesh2022approximate}, 
we derive central limit theorems for hypergraph independent sets. 
We consider the \emph{Gibbs measure} $\mu_{H,\*\lambda}$ associated with $Z^\-{ly}_H(\*\lambda)$, defined as follows:
\[
\forall \sigma\in \+I(H), \quad \mu_{H,\*\lambda} (\sigma) = \frac{1}{Z^\-{ly}_H(\*\lambda)}\prod_{v:\sigma(v)=1} \lambda_v.
\]
The measure  $\mu_{H,\*\lambda}$ can be analytically continued to the complex plane through a connected zero-free region, that is, regions where $Z^\-{ly}_H(\*\lambda)\neq 0$.  
Our multivariate zero-freeness for the hypergraph independence polynomial is especially powerful. In principle, one can derive a central limit theorem for any univariate projection of the polynomial.
We demonstrate a natural example by giving a quantitative central limit theorem (also known as Berry-Esseen inequality) for the size of a random independent set, drawn from the Gibbs measure of hypergraph independent sets.
Our new zero-free region and central limit theorem (CLT) can be lifted to a local CLT as in~\cite{vishesh2022approximate}. 
%A local CLT provides not only tail bounds, but also tight estimates of probability mass functions.  
%Combined, we follow a similar strategy as in \cite{vishesh2022approximate} to approximate coefficients of the partition function, such as the number of independent sets of size $t$.
We defer the proof to \Cref{sec:clt}. 
\begin{theorem}[Central limit theorem for hypergraph independent sets]
\label{theorem:clt}
Fix $k\ge 2$ and $\Delta \ge 3$.
Let $H = (V, \+E)$ be a $k$-uniform hypergraph with maximum degree $\Delta$. Let $n = \abs{V}$.
Fix any $\delta>0$, $\varepsilon\in\left(0,\frac{1}{9k^5\Delta^2}\right)$. Let $\lambda_{c, \varepsilon}$ be defined as in \Cref{theorem:zero-freeness-his-lee-yang}. 
For any $\lambda \in (\delta, \lambda_{c, \varepsilon}]$, let $I\sim \mu_{H,\lambda}$, and define $X = \abs{I}$, $\bar{\mu} = \=E[X]$ and $\sigma^2 = \-{Var}[X]$. Then we have $\sigma^2 = \Theta_{k, \Delta, \varepsilon}(\lambda n)$ and
\[
\sup_{t \in \=R} \abs{\=P[(X - \bar{\mu})/\sigma \le t] - \=P[\+Z \le t]} = O_{k,\Delta,\delta,\varepsilon}\left(\frac{\log n}{\sqrt{n}}\right),
\]
where $\+Z \sim N(0, 1)$ is a standard Gaussian random variable.

Furthermore, let $\+N(x) = \mathrm{e}^{-x^2/2}/\sqrt{2\pi}$ denote the density of the standard normal distribution, we have
    \[
    \sup\limits_{t \in \=Z}\abs{\=P[X = t] - \sigma^{-1}\+N((t - \bar{\mu})/\sigma)} = O_{k,\Delta, \varepsilon}\tp{\min\tp{\frac{(\log n)^{5/2}}{\sigma^2},\frac{1}{\sigma^2}+\frac{\sigma^{2k}(\log n)^2}{n^{k-1}}}}.
    \]
\end{theorem}

As a side note, while there is rich literature on Markov chain central limit theorems (CLT), these do not seem to apply to our context. Specifically, our CLT crucially captures the \emph{unimodality} of the stationary distribution itself, while Markov chain CLT concerns the sum of samples generated by a Markov chain, and does not seem to distinguish between unimodal and multimodal distributions.
Log-Sobolev type inequalities (LSI), if available, would also give concentration tail estimates. But the recent spectral independence framework for establishing LSI for Markov chains requires \emph{arbitrary pinnings}, which breaks the uniformity of hyperedges. 

%\textcolor{red}{
%\todo{yyx: I move the theorems related to LCLT here.}
%The proofs are deferred to \Cref{section:lclt}.
%\begin{theorem}[Local central limit theorem for hypergraph independent sets]\label{theorem:lclt}
%    Let $H$ be a $k$-uniform hypergraph $H=(V, \+E)$ with maximum degree $\Delta$ and $n = \abs{V}$. Fix $\varepsilon \in \tp{0, \frac{1}{9k^5\Delta^2}}$. Let $\lambda_{c,\varepsilon}$ be defined as in \Cref{theorem:zero-freeness-his-lee-yang}. Let $\+N(x) = \mathrm{e}^{-x^2/2}/\sqrt{2\pi}$ denote the density of the standard normal distribution. For any $\lambda \in (0, \lambda_{c,\varepsilon}]$, let $I \sim \mu_{H, \lambda}$, and define $X = \abs{I}$, $\bar{\mu} = \=E[X]$ and $\sigma^2 = \-{Var}[X]$, then we have
%    \[
%    \sup\limits_{t \in \=Z}\abs{\=P[X = t] - \sigma^{-1}\+N((t - \bar{\mu})/\sigma)} = O_{k,\Delta,\varepsilon}\tp{\min\tp{\frac{(\log n)^{5/2}}{\sigma^2},\frac{1}{\sigma^2}+\frac{\sigma^{2k}(\log n)^2}{n^{k-1}}}}.
%    \]
%\end{theorem}
%As a consequence of the local CLT, 
Inspired by~\cite{vishesh2022approximate}, we give an FPTAS based on zero-freeness and local CLT, to approximate the number of hypergraph independent sets of size $t$. The proofs are deferred to \Cref{section:lclt}.
\begin{theorem}\label{theorem:fptas-exact-k}
Fix $k\ge2$, $\Delta\ge 3$. Let $H=(V, \+E)$ be a $k$-uniform hypergraph with maximum degree $\Delta$ and $n = \abs{V}$. Let $\varepsilon$, $\lambda_{c, \varepsilon}$ be defined as in \Cref{theorem:zero-freeness-his-lee-yang}.
There exists a deterministic algorithm which, on input $H$, an integer $1\le t\le n\left(1-\frac{1}{1+\lambda_{c, \varepsilon}}\cdot \left(1+\frac{1}{4\mathrm{e}\Delta k^3}\right)\right)$, and an error parameter $\eta \in (0, 1)$, outputs an $\eta$-relative approximation to the number of hypergraph independent sets of size $t$ in time $(n/\eta)^{O_{k, \Delta, \varepsilon}(1)}$.
\end{theorem}
%We also give an FPRAS with a faster running time to calculate $i_t(H)$.
%\begin{theorem}
%Fix $k\ge2$, $\Delta\ge 3$. Let $H=(V, \+E)$ be a $k$-uniform hypergraph with maximum degree $\Delta$ and $n = \abs{V}$. Let $\varepsilon$, $\lambda_{c, \varepsilon}=1$ be defined as in \Cref{theorem:zero-freeness-his-lee-yang}.
%There exists an FPRAS which, on input $H$, an integer $1\le t\le n\alpha_H(1)$, and an error parameter $\eta \in (0, 1)$, outputs an $\eta$-relative approximation to $i_t(H)$ with probability $3/4$ in time $O_{k, \Delta, \varepsilon}\tp{n^{3}\eta^{-2}
%\log^2(n) \log(\log n) \log \tp{\frac{n}{\eta}}}$.
%\end{theorem}
% \begin{remark}
%     In \Cref{theorem:fpras-exact}, we also establish an FPRAS  with a faster running time, but under a more restrictive regime.
%    This restriction arises because we rely on the rapid mixing of Glauber dynamics, which is proven only when $\lambda = 1$~\cite{HSZ19}. However, their proof can be adapted to demonstrate rapid mixing for $0 < \lambda < 1$.
% \end{remark}
%}

A consequence of the Perron-Frobenius theorem for nonnegative matrices is that any ergodic Markov chain converges to a unique stationary distribution. However, the convergence behavior for complex transition matrices is much less understood.
Central to our analysis is the systematic scan Glauber dynamics with complex transition weights
(see \Cref{definition:complex-MCMC-transition-form} for a formal definition). 

We show that it converges to the stationary measure in the same regime of \Cref{theorem:zero-freeness-his-lee-yang}. 
\begin{theorem}[Convergence of systematic scan Glauber dynamics with complex transitions]
\label{theorem:convergence-lee-yang}
    Under the condition of \Cref{theorem:zero-freeness-his-lee-yang}, the systematic scan Glauber dynamics for the complex measure associated with the independence polynomial 
 $Z^\-{ly}_H(\*\lambda)$ converges.
\end{theorem}

\subsection{Technical overview}

A few challenges arise when trying to locate complex zeros through a percolation-type argument. To extend the notion of probability measures to the complex plane, one can formally define \emph{complex normalized measures} as ratios between partition functions. However, a generalization of statements such as ``stochastically dominated by a sub-critical branching process'' for complex measures appears very challenging. In particular, the monotonicity of probability measures crucially relies on the non-negativity axiom. Our key observation is that a factorization property, which arises in decomposing the Glauber dynamics, can be translated to the complex plane.

Our starting point for locating complex zeros of $Z^\-{ly}_H(\*\lambda)$ is an induction on marginal measures. This approach is implicit in the Lee-Yang theorem and the Asano-Ruelle lemma~\cite{Lee1952StatisticalTO,Asano1970LeeYangTA,ruelle1971extension}, and is applied more explicitly in the contraction method~\cite{peter2019conjecture,liu2019correlation,Shao2019ContractionAU}.
We give a quick review below.
\subsubsection{Locating complex zeros through marginal measures}
\label{section:self-reducibility}
Here we use the standard edge-wise self-reducibility, consider a hypergraph $H = (V, \+E)$ with $\+E = \{e_1, e_2, \dots, e_m\}$, and let $H_i=(V, \+E_i)$ where $\+E_i = \{e_1, e_2, \dots, e_i\}$. 
We write the partition function as: $Z_H = Z_{H_0} \prod_{i=1}^{m} \frac{Z_{H_i}}{Z_{H_{i-1}}} $. 
To establish $Z_H\neq 0$, it suffices to show $\frac{Z_{H_i}}{Z_{H_{i-1}}} \neq 0$ as it is clear that $Z_{H_0} \neq 0$.
The ratio $\frac{Z_{H_i}}{Z_{H_{i-1}}}$ corresponds to a marginal measure, which we explain in the context of hypergraph independence polynomial.
A hypergraph independent set $\sigma$ in $H_{i-1}$ is an independent set in $H_i$ if and only if $\sigma_{e_i} \neq 1^{k}$. Thus,
$\frac{Z_{H_i}}{Z_{H_{i-1}}} = 1 - \mu_{H_{i-1}}\tp{\sigma_{e_i} = 1^{k}}$, where $\mu_{H_{i-1}}$ is the measure associated with $Z_{H_{i-1}}$. Then, one can set up an induction on $i$: assuming that $Z_{H_{i-1}}\neq 0$, one shows that the marginal measure $\mu_{H_{i-1}}\tp{\sigma_{e_i} = 1^{k}}\neq 1$, this implies $Z_{H_{i}}\neq 0$.

\subsubsection{Marginal measures through information percolation on complex Markov chains}
Our departure from previous works on the absence of zeros is that we introduce a systematic scan Glauber dynamics to analyze the marginal measures.
Introducing Glauber dynamics is crucial in bypassing a barrier to a better zero-free region for the hypergraph independence polynomial: strong spatial mixing does not hold, and a computational tree construction does not preserve the uniformity of hyperedges.

Given a measure $\mu_{H,\*\lambda}$, Glauber dynamics is a canonical way of constructing a Markov chain with stationary measure $\mu_{H,\*\lambda}$. In particular, the transition matrix of the Glauber dynamics, denoted by $P_{\*\lambda}$, can also be analytically continued to the complex plane through a connected zero-free region as $\mu_{H,\*\lambda}$ is well-defined. In particular, $\mu_{H,\*\lambda}$ is a left eigenvector for $P_{\*\lambda}$ with eigenvalue $1$.
The analysis of Markov chains for $\lambda \in \=R$ mainly concerns the spectral gap of $P_{\*\lambda}$, but the spectral gap usually tends to zero as $n$ goes to infinity (in the thermodynamics limit).
Instead of attempting a complex extension of spectral theory, we work with the marginal measures generated by powers of the transition matrix $P_{\*\lambda}$.

To get a handle on the marginal measures, we take inspirations from the decomposition of Glauber dynamics that arises in information percolation arguments for Markov chains~\cite{lubetzky2016information,HSZ19, HSW21,qiu2022perfect,feng2023towards}. In these applications, one starts by formulating the Markov chain on a space-time slab (also known as a \emph{witness graph}) so that updates, when viewed backward in time, behave like a subcritical percolation. To do so, each step of the dynamics is decomposed into an \emph{oblivious} update part, which updates a site independent of its neighbors, and an adaptive (non-oblivious) part in which one tries to make up the correct transition probability. 
By revealing the randomness used in these updates backward in time, 
we either continue the revealing process due to an adaptive update or terminate it upon encountering an oblivious update.
Previously, this percolation argument has primarily been used to bound the mixing rates of classical Markov chains \cite{lubetzky2016information,HSZ19} 
and to analyze the time required for coalescence in grand coupling processes, 
such as coupling from the past (CFTP)~\cite{HSW21,qiu2022perfect} and its variant, coupling towards the past (CTTP)~\cite{feng2023towards}.

Our idea is to interpret a decomposition of Glauber dynamics as implicitly a decomposition of the transition matrix $P_{\*\lambda}$, also into an oblivious part and an adaptive part. Say we ``initialize'' the Glauber dynamics with a complex measure $\mu$, viewed as a row vector, and we consider the measure generated by $T$ steps of Glauber dynamics, which is the vector-matrix product $\mu P_{\*\lambda}^T$. By expanding this summation, one can see that, upon encountering an oblivious part, the contribution to the sum ``factorizes''. In fact, the result of  $\mu P_{\*\lambda}^T$ formally corresponds to summing over walks of length $T$ over a space-time slab, where each node is weighted by the corresponding entry in the transition matrix. And the factorization is what leads us to define ``independence'' for complex measures, which effectively allows us to ``terminate the percolation'' just as in a standard argument. 
Central to our analysis is to show that, after running the dynamics for sufficiently long, we can use the ``oblivious updates'' as a certificate/witness for the measure of any event, in the sense that these \emph{witness sequences} dominate the complex measure $\mu P_{\*\lambda}^T$. This is formalized as~\Cref{condition:convergence}. These oblivious updates themselves are much easier to analyze as they correspond to a product of complex measures.

By identifying the measure generated by $\mu P_{\*\lambda}^T$ as contributions from an information percolation on a space-time slab (formally defined as \emph{witness graphs} in \Cref{definition:witness-graph-indset})
, we introduce several dynamics-related quantities --- bad vertices, bad components, bad trees (\Cref{definition:bad-vertices-his})
--- to trace the information percolation process (formally through \Cref{lemma:characterization,lemma:modulus-bound-subset,lemma:diminishing-bad-tree}).
Then, we express the measure of any configuration by these quantities. When the information percolation process terminates quickly (in the sense of~\Cref{condition:convergence}), we can control the marginal measure using a product of complex measures.

\subsubsection{Convergence of the complex systematic scan Glauber dynamics}

The convergence of Markov chains in the real case is well understood thanks to the Perron-Frobenius theory and the coupling method. It is unclear what the right generalizations to the complex plane should be. 
Using the information percolation framework, we categorize the percolation processes as follows: 
\begin{enumerate}
    \item processes that terminate before reaching the starting time (\Cref{lemma:characterization});
    \item processes that do not terminate before reaching the starting time (\Cref{lemma:diminishing-bad-tree}).
\end{enumerate}
To establish convergence, it suffices to show that the contributions from type (2) processes diminish to zero. Unlike standard percolation theory where the existence of limits are guaranteed by monotone events, we have to give non-asymptotic bounds before taking an appropriate limit (see \Cref{lemma:diminishing-bad-tree}).  
Combined, this allows us to show that the measure of any event is dominated by witness sequences (\Cref{condition:convergence}), and we give a proof of convergence in \Cref{lemma:convergence}.

\section{Preliminaries and notations}\label{sec-pre}

\subsection{Complex normalized measures}

Our technique involves dealing with complex measures, so we provide some measure theory basics for our presentation.

Let $\mu:\Omega\to \=C$ be a complex measure over a measurable space $(\Omega,\+F)$, where $\Omega$ is a finite set and elements in $\+F$ are called events. The support of $\mu$ is defined as $\supp(\mu)\defeq\{x\in \Omega\mid \mu(x) \neq 0\}$.
We say $\mu$ is \emph{normalized} if
$\sum\limits_{\omega\in \Omega}\mu(\omega)=1$.
The measure on $A \in \+F$ is given by 
$\mu(A) = \sum_{\omega \in A} \mu(\omega)$.
Similar to probability, for any event $A\in \+F$ with $\mu(A)\neq 0$, we can define the \emph{conditional measure} of $\mu$ on $A$ as a restricted measure $\mu(\cdot \mid A)$ over the measure space $(\Omega,\+F_{A})$ where $\+F_{A}=\{B\cap A:B\in \+F\}$ such that for any $B\in \+F$,
\[
\mu(B\mid A)=\frac{\mu(B\cap A)}{\mu(A)}.
\]

Note that the conditional measure $\mu(\cdot \mid A)$ is always normalized when well-defined.

We say that two events $A_1, A_2 \in \+F$ are \emph{independent} if and only if,
$\mu(A_1 \cap A_2) = \mu(A_1) \cdot \mu(A_2)$.

More generally, for a finite sequence of events $A_1, A_2, \ldots A_m \in \+F$, we say they are \emph{mutually independent} if and only if, for any finite subset $I \subset \= \{1, 2, \ldots, m\}$, it holds that
\[
\mu\left( \bigcap_{i\in I} A_i \right) = \prod_{i\in I} \mu(A_i).
\]

For a finite sequence of events $A_1, A_2, \ldots, A_m \in \+F$ we say that they are \emph{mutually disjoint} if for any $i\neq j$, $A_i \cap A_j = \emptyset$.
We also define the \emph{law of total measure}. Let $A_1, A_2, \ldots, A_m \in \+F$ be a finite sequence of mutually disjoint events and let $\bigcup_{i=1}^m A_i = \Omega$. Then for any $B \in \+F$, we have that,

\[
\mu(B) = \sum\limits_{i=1}^m \mu(B \cap A_i).
\]

\subsection{Graphical model and complex zeros}

Let $H = (V,\+E)$ be a hypergraph, where each vertex $v \in V$ represents a random variable that takes its value from a finite domain $[q]=\{1,2,\ldots,q\}$ and each hyperedge $e \in \+E$ represents a local constraint on the set of variables $e \subseteq V$.
For each $v \in V$, there is a function $\phi_v: [q]\to \={C}$ that expresses vertex activity (external fields), 
and for each $e \in \+E$, there is a function $\phi_e: [q]^{e} \to \={C}$ that expresses (hyper)edge activity (nearest-neighbor interactions).
A graphical model is specified by the tuple $\+G = (H, (\phi_v)_{v \in V}, (\phi_e)_{e \in \+E})$, namely the hypergraph associated with the family of vertex and edge activities.
For each configuration $\sigma \in [q]^V$, define its weight by
\begin{align*}
	w_{\+G}(\sigma) \defeq \prod_{v \in V}\phi_v(\sigma_v)\prod_{e \in \+E}\phi_e(\sigma_e).
\end{align*}
Then the partition function $Z_{\+G}$ of the graphical model $\+G$ is given by 
\begin{align*}
	Z=Z_{\+G}\defeq \sum_{\sigma \in [q]^V}w_{\+G}(\sigma).
\end{align*}

We study the complex zeros of the partition function $Z$ of the graphical model in the following aspects:
\begin{itemize}
    \item with respect to the vertex activities $(\phi_v)_{v \in V}$, also known as Lee-Yang zeros;
    \item with respect to the edge activities $(\phi_e)_{e \in \+E}$, also known as Fisher zeros.
\end{itemize}

When the partition function $Z_{\+G}$ is non-zero, we can naturally associate it with a complex normalized measure $\mu=\mu_{\+G}$, called the \emph{Gibbs measure}, defined on the measurable space $([q]^V,2^{[q]^V})$, where
\[
\forall \sigma\in [q]^V,\quad \mu(\sigma)=\frac{w(\sigma)}{Z_{\+G}}.
\]

For any subset of variables $\Lambda\subseteq V$ and a partial restriction $\sigma\in [q]^\Lambda$ over $\Lambda$, we say $\sigma$ is \emph{feasible} if its measure is non-zero.
For any disjoint $S, \Lambda\subseteq V$  and any feasible $\sigma\in [q]^\Lambda$, 
we use $\mu_S^{\sigma}$ to denote the marginal measure induced by $\mu$ on $S$ conditioned on $\sigma$, i.e.,
\[
\forall \tau\in [q]^S,\quad \mu_S^{\sigma}(\tau)=\frac{\mu(X_S=\tau\mid X_{\Lambda}=\sigma)}{\mu(\sigma)}.
\]

\subsection{Glauber dynamics: random scan and systematic scan}
We need to work with a {systematic scan} variant of the {Glauber dynamics}, which is a fundamental Markov chain for high-dimensional measures.
We recall the definitions here.
Let $\mu$ be a distribution over $[q]^V$, with $V = \{v_1,v_2,\ldots,v_n\}$.
The \emph{Glauber dynamics} is a canonical construction of Markov chains with stationary distribution $\mu$. 
Starting from an initial state $X_0 \in [q]^V$ with $\mu(X_0)>0$,
the chain proceeds as follows at each step $t$:
\begin{itemize}
	\item pick a variable $v \in V$ uniformly at random and set $X_t(u) = X_{t-1}(u)$ for all $u \neq v$;
	\item update $X_t(v)$ by sampling from the distribution $\mu_{v}^{X_{t-1}(V \setminus \{v\})}$.
\end{itemize}

The \emph{systematic scan Glauber dynamics} is a variant of Glauber dynamics that, instead of updating a variable at random, one updates them in a canonical order. 
Specifically, at each step $t$, we choose the variable $v=v_{i(t)}$, where 
\begin{align}
    i(t) \defeq (t \mod n) + 1.\label{eq:scan-i(t)}
\end{align}
Then $X_{t-1}$ is updated to $X_{t}$ using the same rule as in  Glauber dynamics, based on the chosen $v$.

The Glauber dynamics is well known to be both aperiodic and reversible with respect to $\mu$. 
The systematic scan Glauber dynamics is not time-homogeneous, as variables are accessed in a cyclic order.
However, by bundling $n$ consecutive updates, we obtain a time-homogeneous Markov chain that is aperiodic and reversible. 

\subsection{$2$-tree}

We also need the notion of $2$-trees~\cite{Alon91}. 
Given a graph $G=(V,E)$, its square graph $G^2=(V,E_2)$ has the same vertex set, while an edge $(u,v)\in E_2$ if and only if $1\leq \dist_G(u,v)\leq 2$. 

\begin{definition}[$2$-tree]\label{definition:2-tree}
Let $G=(V, E)$ be a graph. A set of vertices $T\subseteq V$ is called a \emph{$2$-tree} of $G$, if
\begin{itemize}
    \item  for any $u,v\in T$, $\text{dist}_G(u,v)\geq 2$, and
    \item  $T$ is connected on $G^2$.
\end{itemize}
\end{definition}

Intuitively, a $2$-tree is an independent set that does not spread far away. We can construct a large $2$-tree in any connected graph as follows.

\begin{definition}[construction of a maximal $2$-tree in a connected graph {\cite[Lemma 4.5]{Vishesh21towards}}]\label{definition:2-tree-construction}
Let $G = (V, E)$ be a connected graph of maximum degree $D$ and $v\in V$. We can deterministically construct a $2$-tree $T$ of $V$ containing $v$
such that $\abs{T} \geq \lfloor |V|/(D + 1) \rfloor$ as follows:
\begin{itemize}
\item order the vertices in $V$ in lexicographical order. Start with $T=\{v\}$ and $U=V\setminus N^+(v)$, where $N^+(v)\defeq N(v)\cup \{v\}$ and $N(v)\defeq \{u\in V\mid (u,v)\in E\}$ ;
\item repeat until $U=\emptyset$: let $u$ be the vertex in $U$ with the smallest distance to $T$, with ties broken by the order on $V$. Set $T\gets T\cup \{u\}$ and $U\gets U\setminus N^+(v)$.
\end{itemize}
\end{definition}
The following two lemmas bound the number of subtrees and $2$-trees of a certain size containing a given vertex, respectively.

\begin{lemma}[\text{\cite[Lemma 2.1]{borgs2013left}},{\cite[Corollary 5.7]{feng2021rapid}}]\label{lemma:2-tree-number-bound}
Let $G=(V, E)$ be a graph with maximum degree $D$, and $v\in V$ be a vertex. 
The number of subtrees in $G$ of size $k \ge 2$ containing $v$ is at most $\frac{(\mathrm{e}D)^{k-1}}{2}$, and the number of 2-trees in $G$ of size $k \geq 2$ containing $v$ is at most $\frac{(\mathrm{e}D^2)^{k-1}}{2}$.
\end{lemma}

\section{Convergence of complex Markov chains}\label{sec:MCMC}

In this section, we present our framework for establishing new zero-free regions for certain polynomials under local lemma conditions,
as well as for proving the convergence of complex Markov chains. 
We begin by defining a complex Markov chain: the complex systematic scan Glauber dynamics.
Then, our proof is carried out by lifting the information percolation argument,
a powerful technique commonly used to prove the rapid mixing of Markov chains~\cite{HSZ19, HSW21,qiu2022perfect,feng2023towards},
to the complex plane.
Both zero-freeness and convergences are intrinsically related to bounding the marginal measures,
which can be interpreted as limiting the ``randomness" of the Markov chain.
We show that if transition measures can be decomposed (as captured by \Cref{definition:decomposition-scheme})
we can apply a complex variant of the information percolation argument to effectively bound the norm of the marginal measures.

\subsection{Complex extensions of Markov chains}

\subsubsection{Complex Markov chain}
We first define complex-valued transition matrices on a finite state space.
Let $\Omega$ be a finite state space. 
We say $P \in \=C^{\Omega \times \Omega}$ is a complex-valued transition matrix if 
\[\forall \sigma \in \Omega, \quad \sum_{\tau \in \Omega} P(\sigma, \tau) = 1,\] 
which is a direct extension of the classical row-stochastic matrix.

Fix $T \ge 1$. For a measurable space $(\Omega, \+F)$ with finite $\Omega$, 
we write $\Omega^T$ for the Cartesian product, and $\+F^{T}$ for the product $\sigma$-algebra. 
Let $\`P$ be a  complex normalized measure on $(\Omega^{T}, \+F^{T})$  and 
let $X_1, X_2, \ldots, X_T$  be a sequence  of measurable functions taking values over $\Omega$ following the measure~$\`P$. 
The sequence $(X_t)_{t=1}^T$ is said to be a $T$-step discrete-time complex Markov chain if there exists a complex-valued transition matrix $P \in \=C^{\Omega \times \Omega}$ such that
for any $1< j \le T$ and any $x_1, x_2, \dots, x_{j} \in \Omega$,  
\[
\`P(X_{j} = x_{j} \mid X_{1}=x_1, X_{2}=x_2, \dots, X_{j-1} = x_{j-1}) = \`P(X_{j} = x_{j} \mid X_{j-1} = x_{j-1}) = P(x_{j-1}, x_{j}),
\]

We often use $P$ to refer to the corresponding Markov chain. 
For a complex normalized measure  $\nu \in \=C^{\Omega}$ on~$\Omega$,
the measure $\nu P$ obtained via  a one-step transition of the Markov chain from $\nu$ is given by
\[\forall x \in \Omega,\quad (\nu P)(x) = \sum\limits_{y \in \Omega} \nu(y) P(y, x).\]
A complex normalized measure $\pi$  over $\Omega$  is a \emph{stationary measure} of $P$ if $\pi = \pi P$. 
It is important to note that for a generic complex row-stochastic matrices $P$,  it may not have a stationary measure\footnote{While there is a left-eigenvector with eigenvalue $1$, it can sum up to $0$, and cannot be normalized to a complex measure. This is also the main reason why convergence alone does not imply zero-freeness, as we need to rule out the possibility of converging to an eigenvector that cannot be normalized.}, and even if it does, it may not be unique.

Next, we define the convergence of the Markov chains with complex-valued transition matrices.

\begin{definition}[convergence of the complex Markov chains]
\label{definition:convergence-MC}
    A Markov chain with a complex-valued transition matrix $P$ and state space $\Omega$ is said to be \emph{convergent} if,
    for any two complex normalized measures $\mu$ and $\mu^*$ over $\Omega$, it holds that
        \[
    \lim_{T \to \infty} \abs{\mu P^T - \mu^* P^T}_1 = 0.
    \]
\end{definition}

\subsubsection{Complex Glauber dynamics}
 We introduce a complex extension of systematic scan Glauber dynamics for complex normalized measures. We do so through two equivalent viewpoints: formulating the transition matrices with complex transition weights, and also a dynamics-based formulation. The latter is more convenient for our analysis, and we note that the two are equivalent in the sense that they eventually generate the same complex normalized measures.

\begin{definition}[Complex extension of systematic scan Glauber dynamics] 
\label{definition:complex-MCMC-transition-form}
Let $\mu\in \=C^{[q]^V}$ be a complex normalized measure.
The complex systematic scan Glauber dynamics for the target measure $\mu$ is defined by a sequence of complex-valued transition matrices $P_t \in \=C^{\Omega \times \Omega}$ for $t\ge 1$, 
where with $v = v_{i(t)}$ (and $i(t)$ is as defined in \eqref{eq:scan-i(t)}), the transition matrix  $P_t$ is defined as
\[
P_t(\sigma, \tau) \defeq \begin{cases}
    \mu_v^{\sigma(V \backslash \{v\})}(\tau_v) & \text{ if } \forall u \neq v, \sigma_u = \tau_u,\\
    0 & \text{ otherwise.}
\end{cases}
\]

\end{definition}
Starting from an initial state $\tau\in \supp(\mu)$, 
the complex Markov chain generates an induced complex measure $\mu_t\in \=C^{[q]^V}$ which we define next.   At time $t=0$, we define $\mu_0(\tau) = 1$ and $\mu_0(\sigma) = 0$ for all $\sigma \in [q]^V\setminus\{\tau\}$, and for $t\ge 1$, we define 
$\mu_{t} \defeq \mu_{t-1} P_t$.

\begin{remark}[well-definedness of the complex systematic scan Glauber dynamics]\label{remark:MCMC-well-defined}
The complex systematic scan Glauber dynamics, as defined in \Cref{definition:complex-MCMC-transition-form}, 
is well-defined as long as the conditional measures $\mu_{v}^{\sigma(V\setminus \{v\})}$ are well-defined for each $\sigma\in \supp(\mu)$ and every $v\in V$.
Then, the induced complex measures $\mu_t$ remain normalized at any time.
\end{remark}

We now present the dynamics-based formulation of the complex systematic scan Glauber dynamics in \Cref{Alg:complex-GD}. Recall that the dynamics have a stationary measure $\mu$, an initial starting state $\tau$, and we denote the associated induced complex measure by $\mu^{\-{GD}}_{T,\tau}$. For technical convenience, we shift the timeline of the dynamics so that we are starting with a state $\sigma_{-T}$ and the final state is $\sigma_0$.

\begin{algorithm}
\caption{Complex systematic scan Glauber dynamics} \label{Alg:complex-GD}
\SetKwInOut{Input}{Input}
\Input{An arbitrary initial configuration $\tau\in\supp(\mu)\subseteq[q]^V$ and an integer $T\ge 1$.} 
Set $\sigma_{-T}\gets\tau$\;
\For{$t=-T+1,-T+2,\ldots,0$}{
let $\sigma_{t}\gets\sigma_{t-1}$ and $v\gets v_{i(t)}$, where $i(t)=(t \mod \abs{V}) + 1$\;
let $c_t$ follow the conditional measure $\mu_{v}^{\sigma_{t-1}(V\setminus \{v\})}$\;
update $\sigma_t(v)\gets c_t$\;
}
\end{algorithm}

\begin{remark} 
It is important to emphasize that \Cref{Alg:complex-GD} (as well as \Cref{Alg:complex-GD-decomposed,Alg:complex-GD-decomposed-his}, which are introduced later) is not meant to be an efficient algorithm for sampling from a complex measure.
Rather, the ``algorithms'' described in this paper serve as analytical tools for establishing zero-freeness.
When we state ``let $c$ follow a complex normalized measure $\mu$'' or ``$c$ is drawn from a complex normalized measure $\mu$'', we mean that the measure of $c$ is the same as $\mu$. 
This statement is conceptual rather than operational;
we do not attempt to explicitly generate samples during runtime. 
Any subsequent operation on $c$ should be understood as a transformation applied to the complex normalized measure of $c$.
%
%For example, suppose the value of $\sigma$ is updated based on certain rules involving $c$.
%In that case, the updated measure of $\sigma$ is derived by applying a suitable transformation to the measure of $\sigma$ before update and the measure of~$c$.
%
It is worth noting that any finite segment of the complex measures is computable on a deterministic Turing machine in exponential time,
provided all the involved measures can be described using Gaussian rational numbers. 
This can be achieved by explicitly enumerating all outcomes of the process.
    \end{remark}

\begin{remark} \label{remark:remark-equivalence-matrix-dynamics}
It is straightforward to verify that the processes described in \Cref{definition:complex-MCMC-transition-form} and \Cref{Alg:complex-GD} are essentially equivalent in the following sense: for any $\tau^* \in \supp(\mu)$,  we have $\mu^\-{GD}_{T, \tau}(\sigma_t = \tau^*) = \mu_{t+T}(\tau^*)$ for all $-T \le t \le 0$.
This can be routinely verified through induction on $t$.
    \end{remark}

In general, we lack convergence theorems for Markov chains with complex-valued transition matrices in the literature,  
so we cannot immediately assert whether the complex measure after $T$ steps converges to the target measure $\mu$ as $T\to\infty$.
However, we can consider a complex process initialized with the stationary measure $\mu$. 

\begin{definition}[stationary systematic scan Glauber dynamics]\label{definition:stationary-start}
Consider the process defined in \Cref{Alg:complex-GD}, 
but now with the initial state $\sigma_{-T}$ following the measure $\mu$. 
We call this modified process the $T$-step stationary systematic scan Glauber dynamics,
and denote its induced measure as $\mu^{\-{GD}}_{T}$.
\end{definition}

It is straightforward to verify that for all $t\in [-T,0]$, the measure induced on $\sigma_{t}$  under $\mu^{\-{GD}}_{T}$ precisely follows  the measure $\mu$.
\Cref{definition:stationary-start} will play an essential intermediate role in our proof of the convergence of Glauber dynamics. Note that the measure $\mu^{\-{GD}}_{T}$ is a linear combination of the measures $\mu^{\-{GD}}_{T,\sigma}$ over all starting states $\sigma\in \supp(\mu)$.   Later, by comparing with the stationary process in \Cref{definition:stationary-start}, we will identify a sufficient condition (\Cref{condition:convergence}) and prove (in \Cref{lemma:convergence}) that, under this condition, the Glauber dynamics starting from any initial state converges to a unique limiting measure, which is precisely the stationary measure $\mu$.

\subsection{Decomposition of transition measure}
\label{subsection:decomposition-schemes}
Inspired by the information percolation approach for the real case, 
we consider the decomposition of the transition measure for a complex Markov chain.
Specifically, in the context of complex systematic scan Glauber dynamics,  
each transition measure can be decomposed into two parts: 
an oblivious part, where the transition does not depend on the current configuration; 
and an adaptive part, where the transition measure depends on the current configuration. 
This leads to the following formal definition of a decomposition scheme.

\begin{definition}[decomposition scheme]\label{definition:decomposition-scheme}
Let $\mu\in \mathbb{C}^{[q]^V}$ be a complex normalized measure. For each $v\in V$, we associate a complex normalized complex measure $b_v:[q]\cup \{\bot\}\to \mathbb{C}$, and let $\bm{b}=(b_v)_{v\in V}$. 

We define the $\bm{b}$-decomposition scheme on $\mu$ as follows.
For each $v\in V$ and each feasible $\tau\in [q]^{V\setminus \{v\}}$ (meaning that $\tau$ can be extended to a $\sigma\in\supp(\mu)$), 
we define the measure $\mu_{v}^{\tau,\bot}$ as
\[
\forall c \in [q], \quad \mu_v^{\tau,\bot}(c) \defeq \frac{\mu^{\tau}_v(c) - b_v(c)}{b_v(\bot)}.
\]
Then, the marginal measure $\mu^{\tau}_v$ can be decomposed as:
\begin{equation}\label{eq:decomposition}
\forall c \in [q], \quad \mu^{\tau}_v(c)=b_v(c)+b_v(\bot)\cdot \mu_v^{\tau,\bot}(c),
\end{equation}
where we assume the convention $0\cdot \infty =0$ to ensure that \eqref{eq:decomposition} still holds when $b_v(\bot)=0$.
\end{definition}

\begin{remark}
Intuitively, for each $v\in V$, the measure $b_v$ serves as a ``baseline measure'' for all transition measures $\mu_v^{\tau}$ at $v$. 
Ideally, $b_v$ should be re-normalized from the lower envelope of all transition measures $\mu_v^{\tau}$ conditioned on an arbitrary feasible $\tau\in [q]^{V\setminus \{v\}}$, making it oblivious to the boundary condition~$\tau$. 
In contrast, the measure $\mu_v^{\tau,\bot}$ captures the ``excess'' of $\mu_v^{\tau}$ over this lower envelop, adapting to the boundary~$\tau$.
It can be verified that as long as each marginal measure $\mu_v^{\tau}$ is well-defined, the decomposition in \eqref{eq:decomposition} is always well-defined.  
\end{remark}

With this decomposition scheme, the complex systematic scan Glauber dynamics described in \Cref{Alg:complex-GD} can be reinterpreted as the process in \Cref{Alg:complex-GD-decomposed}.
We denote its induced measure as $\mu^{\-{GD}}_{T,\tau,\bm{b}}$. 
By equation~\eqref{eq:decomposition}, it is straightforward to verify that $\mu^{\-{GD}}_{T,\tau,\bm{b}} (\sigma_t = \cdot) = \mu^{\-{GD}}_{T,\tau} (\sigma_t = \cdot)$. 
Similarly, as in \Cref{definition:stationary-start},
we also consider the $\bm{b}$-decomposed stationary systematic scan Glauber dynamics, and denote its induced measure as $\mu^{\-{GD}}_{T,\bm{b}}$.

\begin{algorithm}
\caption{$\bm{b}$-decomposed complex systematic scan Glauber dynamics} \label{Alg:complex-GD-decomposed}
\SetKwInOut{Input}{Input}
\Input{ An arbitrary initial configuration $\tau\in\supp(\mu)\subseteq[q]^V$ and an integer $T\ge 1$.} 
Set $\sigma_{-T}\gets\tau$\;
\For{$t=-T+1,-T+2,\ldots,0$}{
let $\sigma_{t}\gets\sigma_{t-1}$ and $v\gets v_{i(t)}$, where $i(t)=(t \mod \abs{V}) + 1$\;
let $r_t$ follow the measure $b_{v}$\;
\eIf{$r_t\neq \bot$}{let $c_t\gets r_t$\;}{
      let $c_t$ follow the measure $ \mu_{v}^{\sigma_t,\bot}$\;}
update $\sigma_t(v)\gets c_t$\;
}
\end{algorithm}

This decomposition of Markov chain transitions is a complex extension of similar decompositions based on unconditional marginal lower bounds in the real case~\cite{anand2021perfect,he2022sampling,feng2023towards}.

The purpose of this decomposition of transition measures
is to determine each update's outcome without relying on knowledge of the current configuration. 
By carefully selecting the baseline complex measures $\bm{b}$, 
this approach allows us to infer the outcome of Glauber dynamics in the following scenarios:
\begin{itemize}
    \item Infer the final outcome $\sigma_0$ of the chain without knowing the initial configuration $\sigma_{-T}$ as $T\to\infty$, showing convergence of the complex Glauber dynamics.
    \item Infer $\sigma_0(v)$ in the complex Glauber dynamics,
    through independent measurable functions $r_t$ following the baseline measures $\bm{b}$, 
    which enables the application of information percolation analysis to bound marginal measures and establish zero-freeness.
\end{itemize}

\subsection{Convergence of systematic scan Glauber dynamics}
We now formalize the above arguments. 
We first define the situations in which a realization $\bm{\rho}$ of $\bm{r}=(r_t)_{t=-T+1}^{0}$, as used in \Cref{Alg:complex-GD-decomposed}, can certify the induced measure over the occurrence of a particular event regarding the final outcome $\sigma_0$, regardless of what the initial state is.

\begin{definition}[witness sequence]\label{definition:witness-sequence}
Fix $T\geq 1$. 
Let $\bm{b}=(b_v)_{v\in V}$ be a collection of complex normalized measures. 
Consider a $\bm{b}$-decomposed complex systematic scan Glauber dynamics as in \Cref{Alg:complex-GD-decomposed}.
For any event $A\subseteq [q]^V$, we say a sequence $\bm{\rho}=(\rho_t)_{t=-T+1}^{0}\in ([q]\cup \{\bot\})^{T}$ is a \emph{witness sequence} with respect to $A$ if
\[
\forall \sigma,\tau\in [q]^V,\quad \mu^{\-{GD}}_{T,\sigma,\bm{b}}(\sigma_0\in A\mid \bm{r}=\bm{\rho})=\mu^{\-{GD}}_{T,\tau,\bm{b}}(\sigma_0\in A\mid \bm{r}=\bm{\rho}),
\]
denoted as $\bm{\rho}\Rightarrow A$. 
Otherwise, we denote $\bm{\rho}\nRightarrow A$.
\end{definition}

The following is a sufficient condition for the convergence of the complex systematic scan Glauber dynamics.

\begin{condition}[a sufficient condition for convergence]\label{condition:convergence}
                                Assuming the systematic scan Glauber dynamics in \Cref{Alg:complex-GD-decomposed-his} is well-defined\footnote{Recall that for the Glauber dynamics to be well-defined, the measure $\mu$ must also be well-defined, which requires the partition function to be non-zero.
    This may seem odd, as \Cref{condition:convergence} will be used to imply the convergence of the Glauber dynamics, which, in turn, will be used to establish the zero-freeness of the partition function.
    However, as explained in the technical overview (\Cref{section:self-reducibility}), this implication will be carried out as an inductive argument that avoids circular reasoning.},
        there exists a sequence of sets $\set{B(T)}_{T\ge 1}$ such that for each $T\ge 1$,  $B(T)\subseteq ([q]\cup \{\bot\})^T$ 
    satisfies the following conditions for all configurations $\tau\in [q]^V$: 
    \begin{itemize}
        \item For all $\bm{\rho}\nRightarrow \{\tau\}$, it holds $\bm{\rho}\in B(T)$; thus, $B(T)$ contains all non-witness sequences for $\{\tau\}$.
                \item For any initial configuration $\sigma\in \supp(\mu)$, the following limit exists and satisfies:
    \[
    \lim\limits_{T\to \infty} \abs{\sum\limits_{\bm{\rho} \in B(T)}\mu^{\-{GD}}_{T,\sigma,\bm{b}}(\bm{r}=\bm{\rho})\cdot \mu^{\-{GD}}_{T,\sigma,\bm{b}}(\sigma_0=\tau\mid \bm{r}=\bm{\rho})}
        =0.
    \]
    \end{itemize}   
   
\end{condition}

By the law of total measure,
\Cref{condition:convergence} translates to:
\begin{align}
\abs{\mu^{\-{GD}}_{T,\sigma,\bm{b}}\tp{\sigma_0=\tau\land \bm{r}\in B(T)}}\to 0 \quad \text{ as } T\to \infty.\label{eq:condition:convergence}    
\end{align}
Ideally, for each individual sample $\tau\in[q]^V$, we are concerned with the set of all non-witnesses for the event $\{\tau\}$, and want to establish that $\abs{\mu^{\-{GD}}_{T,\sigma,\bm{b}}\tp{\sigma_0=\tau\land \bm{r}\nRightarrow \{\tau\}}}\to 0$ as $T\to \infty$.
Instead, \Cref{condition:convergence} guarantees the existence of $B(T)$ of non-witnesses for every $T$ with respect to any elementary event $\{\tau\}$. 
We note that although complex measures may not be monotone, this condition suffices as it effectively reduces to reasoning about a product measure on witness sequences.

The sufficiency of \Cref{condition:convergence} is formalized by the following lemma.

\begin{lemma}[convergence of complex systematic scan Glauber dynamics]\label{lemma:convergence}
    Assume that there exists a $\bm{b}$-decomposition scheme such that \Cref{condition:convergence} holds.
    Then, the complex systematic scan Glauber dynamics converges to $\mu$ as $T\to\infty$, starting from any initial configuration $\sigma\in \supp(\mu)$.
\end{lemma}

\begin{proof}
    Recall that applying the $\bm{b}$-decomposition scheme does not affect the induced measure on $\sigma_0$. 
    We claim that for any two initial configurations $\sigma,\sigma'\in \supp(\mu)$, the following always holds:
    \begin{equation}\label{eq:convergence-equivalent-condition}
       \forall\tau\in[q]^V,\quad  \lim\limits_{T\to \infty} \abs{\mu^{\-{GD}}_{T,\sigma,\bm{b}}(\sigma_0=\tau)-\mu^{\-{GD}}_{T,\sigma',\bm{b}}(\sigma_0=\tau)}=0.
    \end{equation}
     Assuming \eqref{eq:convergence-equivalent-condition}, we compare the chain $\mu^{\-{GD}}_{T,\sigma,\bm{b}}$ starting from an arbitrary initial state $\sigma\in \-{supp}(\mu)$ with the stationary chain (\Cref{definition:stationary-start}). By the triangle inequality, 
    \begin{align*}
    \abs{\mu_{T, \sigma, \*b}^\-{GD}(\sigma_0 = \tau) - \mu(\tau)}&=
        \abs{\mu_{T, \sigma, \*b}^\-{GD}(\sigma_0 = \tau) - \sum\limits_{\sigma' \in \-{supp}(\mu)}\mu(\sigma')\mu_{T, \sigma', \*b}^\-{GD}(\sigma_0 = \tau)  } \\
        &\le \sum\limits_{\sigma' \in \-{supp}(\mu)}\abs{\mu(\sigma')}\abs{\mu_{T, \sigma, \*b}^\-{GD}(\sigma_0 = \tau) - \mu_{T, \sigma', \*b}^\-{GD}(\sigma_0 = \tau)}.
    \end{align*}
    According to \eqref{eq:convergence-equivalent-condition}, as $T\to \infty$, the right-hand side approaches $0$. Thus, the complex systematic scan Glauber dynamics converges to $\mu$.
    
    We then complete the proof by establishing \eqref{eq:convergence-equivalent-condition}. For any $T\ge1$, let $B(T)\subseteq ([q]\cup \{\bot\})^T$ be the set of non-witness sequences satisfying \Cref{condition:convergence}. For any $\tau\in[q]^V$, we have: 
\begin{align*}
    &\lim\limits_{T\to \infty}\abs{\mu^{\-{GD}}_{T,\sigma,\bm{b}}(\sigma_0=\tau)-\mu^{\-{GD}}_{T,\sigma',\bm{b}}(\sigma_0=\tau)}\\
  (\star)  \quad= &\lim\limits_{T\to \infty}\abs{\sum\limits_{\bm{\rho}\in ([q]\cup \{\bot\})^{T}}\mu^{\-{GD}}_{T,\sigma,\bm{b}}(\bm{r}=\bm{\rho})\cdot \left(\mu^{\-{GD}}_{T,\sigma,\bm{b}}(\sigma_0=\tau\mid \bm{r}=\bm{\rho})-\mu^{\-{GD}}_{T,\sigma',\bm{b}}(\sigma_0=\tau\mid \bm{r}=\bm{\rho})\right)}\\
     (\blacktriangle) \quad=&\lim\limits_{T\to \infty}\abs{\sum\limits_{\*\rho\in B(T)}\mu^{\-{GD}}_{T,\sigma,\bm{b}}(\bm{r}=\bm{\rho})\cdot \left(\mu^{\-{GD}}_{T,\sigma,\bm{b}}(\sigma_0=\tau\mid \bm{r}=\bm{\rho})-\mu^{\-{GD}}_{T,\sigma',\bm{b}}(\sigma_0=\tau\mid \bm{r}=\bm{\rho})\right)}\\
   (\blacksquare) \quad\leq&0,
    \end{align*}
    which implies \eqref{eq:convergence-equivalent-condition}. Here, the $(\star)$ inequality follows from the law of total measure, along with the observation that $\mu^{\-{GD}}_{T,\sigma,\bm{b}}(\bm{r}=\bm{\rho})$ does not depend on $\sigma$. 
    The $(\blacktriangle)$ equality follows from \Cref{definition:witness-sequence} and that all $\*\rho\in ([q]\cup \{\bot\})^T\setminus B(T)$ satisfy $\rho\Rightarrow \{\tau\}$. 
    The $(\blacksquare)$ inequality follows from the triangle inequality and \Cref{condition:convergence}.
    This completes the proof.
\end{proof}

\subsection{Bounding the marginal measures}
Assume \Cref{condition:convergence}. 
We now explain how to bound the marginal measure of an event $A\subseteq[q]^V$. 
For any $T \ge 1$. Let $\+S$ be the set of sequences $\*\rho = (\rho_i)_{i=-T+1}^0$ where each $\rho_i \in [q]\cup \set{\bot}$. Note that \Cref{condition:convergence} immediately implies that there is a set $B(T)\subseteq \tp{[q]\cup \set{\bot}}^T$ such that for any event $A \subseteq [q]^V$, and for any $\*\rho \not\in B(T)$, it holds that $\*\rho \Rightarrow A$.  Consider the stationary systematic scan Glauber dynamics (\Cref{definition:stationary-start}), for any event $A\subseteq [q]^V$, by the triangle inequality,
\begin{align*}
\abs{\mu(A)} = &\abs{\sum\limits_{\sigma \in \-{supp}(\mu)} \mu(\sigma) \left( \sum\limits_{\rho \not\in B(T)} \mu^\-{GD}_{T, \sigma, \*b}(\sigma_0\in A \land \*r = \*\rho) + \sum\limits_{\rho \in B(T)}\mu^\-{GD}_{T, \sigma, \*b}(\sigma_0 \in A \land \*r = \*\rho) \right) } \\
\le &\abs{\sum\limits_{\sigma \in \-{supp}(\mu)} \mu(\sigma)\sum\limits_{\rho \not\in B(T)} \mu^\-{GD}_{T, \sigma, \*b}(\sigma_0\in A \land \*r = \*\rho)} \\
&+ \abs{\sum\limits_{\sigma \in \-{supp}(\mu)} \mu(\sigma)\sum\limits_{\rho \in B(T)}\mu^\-{GD}_{T, \sigma, \*b}(\sigma_0 \in A \land \*r = \*\rho)}.
\end{align*}
For any $\*\rho\not\in B(T)$, since $\*\rho$ is a witness sequence for $A$,  we have that for any  $\sigma,\tau  \in \-{supp}(\mu)$,
\[
  \sum\limits_{\rho \not\in B(T)} \mu^\-{GD}_{T, \tau, \*b}(\sigma_0\in A \land \*r = \*\rho) = \sum\limits_{\rho \not\in B(T)} \mu^\-{GD}_{T, \sigma, \*b}(\sigma_0\in A \land \*r = \*\rho).  
\]
By this equation and since  $\mu$ is a complex normalized measure, the previous bound for $\abs{\mu(A)}$ can be expressed as follows, after fixing an arbitrary $\tau\in  \-{supp}(\mu)$:
\begin{align*}
\abs{\mu(A)} &\le \abs{\sum\limits_{\rho \not\in B(T)} \mu^\-{GD}_{T, \tau, \*b}(\sigma_0\in A \land \*r = \*\rho)} + \abs{\sum\limits_{\sigma \in \-{supp}(\mu)} \mu(\sigma)\sum\limits_{\rho \in B(T)}\mu^\-{GD}_{T, \sigma, \*b}(\sigma_0 \in A \land \*r = \*\rho)}\\
&\le \abs{\sum\limits_{\rho \not\in B(T)} \mu^\-{GD}_{T, \tau, \*b}(\sigma_0\in A \land \*r = \*\rho)} + \sum\limits_{\sigma \in \-{supp}(\mu)} \mu(\sigma)\abs{\sum\limits_{\rho \in B(T)}\mu^\-{GD}_{T, \sigma, \*b}(\sigma_0 \in A \land \*r = \*\rho)}.
\end{align*}
As $T \to \infty$, according to \Cref{condition:convergence}, we know that for any $\sigma \in \-{supp}(\mu)$, 
\[\lim_{T \to \infty}\abs{\sum\limits_{\*\rho \in B(T)}\mu^\-{GD}_{T, \sigma, \*b}(\sigma_0 \in A \land \*r = \*\rho)} = 0.\] 
Therefore, as $T\to\infty$, we have
\begin{equation}\label{eq:norm-bound}
\abs{\mu(A)} \le  \abs{\sum\limits_{\rho \not\in B(T)} \mu^{\-{GD}}_{T, \tau, \*b}(\sigma_0 \in A \land \*r = \*\rho)}
= \abs{\sum\limits_{\rho \not\in B(T)} \mu^{\-{GD}}_{T,\tau,\bm{b}}(\bm{r}=\bm{\rho}) \cdot \mu^{\-{GD}}_{T, \tau, \*b}(\sigma_0 \in A \mid  \*r = \*\rho)}.
\end{equation}

This, in turn, enables us to establish zero-freeness results using edge-wise self-reducibility.
It remains to demonstrate how to establish \Cref{condition:convergence} and to upper bound the right-hand side in~\eqref{eq:norm-bound}. 
Note that for any sequence $\bm{\rho}\in ([q]\cup \{\bot\})^{T}$ and any initial configuration $\tau\in\-{supp}(\mu)$, the measure $\mu^{\-{GD}}_{T,\tau,\bm{b}}(\bm{r}=\bm{\rho})$ can be computed directly, as $\bm{r}$ follows a product measure. 
The primary technical challenge then lies in characterizing the bound on the measure $\mu^{\-{GD}}_{T,\tau,\bm{b}}(\sigma_0\in A\mid \bm{r}=\bm{\rho})$ through useful properties of witness sequences~$\bm{\rho}$.
However, this characterization may depend on the concrete models. 
Therefore, we do not aim to provide a generic method for establishing \Cref{condition:convergence} and bounding the marginal measure through \eqref{eq:norm-bound}. 
Instead, we will show in the following section how to apply this general framework to the hypergraph independence polynomials, yielding the desired zero-free and convergence results.

\section{Lee-Yang zeros of the hypergraph independence polynomial}\label{section:zeros}
  
In this section, we will show how to apply the general framework developed in the previous section to establish 
both the absence of Lee-Yang zeros and the convergence of complex Glauber dynamics for the hypergraph independence polynomials.
%
%zero-freeness for the hypergraph independence polynomials and the convergence of complex Glauber dynamics for the associated measure.  
Specifically, we will prove \Cref{theorem:zero-freeness-his-lee-yang,theorem:convergence-lee-yang}. 

%We begin by recalling some notations we have defined.
Consider a hypergraph $H=(V,\+E)$ with $|V|=n$.
A configuration $\sigma\in \{0,1\}^V$ is called an \emph{independent set} of $H$ if for every hyperedge $e \in \+E$, 
there exists at least one vertex $v \in e$ such that $\sigma(v) = 0$.
%contains at least one vertex $v \in e$ with $\sigma(v) = 0$. 
Let $\+I(H) \subseteq \{0,1\}^V$ denote the set of all independent sets of $H$. 
The independence polynomial of the hypergraph $H$ on Lee-Yang zeros is defined as
\[
Z^\-{ly}_H(\*\lambda)=\sum\limits_{\sigma\in \+I(H)} \prod\limits_{v:\sigma(v)=1} \lambda_v,
\]
where $\*\lambda=(\lambda_v)_{v\in V}$ is a vector of complex-valued parameters associated with the vertices.
%We note that this is the same independence polynomial as in the introduction, by identifying an independent set $S$ as the set of vertices assigned $1$ in $\sigma$.

Provided that $Z^\-{ly}_H(\*\lambda)\neq 0$, the associated (complex-valued) Gibbs measure $\mu=\mu_{H,\lambda}:\{0,1\}^V\to \=C$ is defined as
\[
\forall \sigma\in \+I(H),\quad \mu(\sigma)=\mu_{H,\lambda}(\sigma)=\frac{\prod_{v:\sigma(v)=1} \lambda_v}{Z^\-{ly}_H(\*\lambda)}.
\]

The following presents a sufficient condition for the hypergraph $H=(V,\+E)$ with complex vertex weights $\*\lambda\in\mathbb{C}^V$ such that $Z^\-{ly}_H(\*\lambda)\neq 0$, indicating that $\*\lambda$ is not a complex zero of $Z^\-{ly}_H$.
\begin{condition}\label{cond:his-lee-yang-marginal}
For the hypergraph $H=(V,\+E)$ with complex vertex weights $\*\lambda\in\mathbb{C}^V$,  the following holds. 
Fix an arbitrary ordering of hyperedges $\+E=\{e_1,e_2,\dots,e_m\}$.  
For each $0\leq i\leq m$, let $\+E_i=\{e_1,e_2,\dots,e_i\}$ and $H_i=(V,\+E_i)$.
For each $0\leq i<m$: 
%\todo{JL: Keep the original notation of $\sigma$ here}
\[
Z^\-{ly}_{H_{i}}(\*\lambda)\neq 0 \implies \abs{\mu_{H_i,\*\lambda}\left(\sigma(e_{i+1}) =  1^{e_{i+1}}\right)}< 1,
\]
where $1^{e}\in\{0,1\}^{e}$ denotes the partial configuration that assigns all vertices in $e$ to 1.
\end{condition}
%\begin{remark}
As discussed in \Cref{section:self-reducibility},
assuming \Cref{cond:his-lee-yang-marginal}, we can routinely establish $Z^\-{ly}_H(\*\lambda)\neq 0$ through induction
by using the edge-wise self-reducibility. A similar argument was used in~\cite{peter2019conjecture,liu2019correlation,Shao2019ContractionAU}.
%\end{remark}

\subsection{Complex Glauber dynamics on hypergraph independent sets}
Our proofs of zero-freeness and convergence of Glauber dynamics
%\Cref{theorem:zero-freeness-his-lee-yang,theorem:zero-freeness-his-fisher,theorem:convergence-lee-yang} 
utilize the framework of $\bm{b}$-decomposed complex systematic scan Glauber dynamics for the Gibbs measure $\mu$, introduced in \Cref{sec:MCMC}. 

We first verify that the transition measures are well-defined.
Note that $\supp(\mu)=\+I(H)$.
For any $v\in V$ and any $\sigma \in \supp(\mu)$, the marginal measure $\mu^{\sigma(V\backslash \{v\})}_v$ is defined as follows:
\begin{itemize}
\item 
If $\sigma$ remains  an independent set after setting $\sigma(v)\gets 1$, then 
\[
\mu^{\sigma(V\backslash \{v\})}_v(0)=\frac{1}{1+\lambda_v} \quad\text{ and }\quad \mu^{\sigma(V\backslash \{v\})}_v(1)=\frac{\lambda_v}{1+\lambda_v}.
\]
\item 
Otherwise, if $\sigma$ is not an independent set after setting  $\sigma(v)\gets 1$, then
\[
\mu^{\sigma(V\setminus\{v\})}_v(0)=1 \quad\text{ and }\quad \mu^{\sigma(V\setminus\{v\})}_v(1)=0.
\]
\end{itemize}
Thus, according to this definition,  for any $v\in V$ where $\lambda_v\neq -1$, the conditional measure $\mu^{\sigma(V\setminus \{v\})}_v$ is well-defined for any $\sigma\in \supp(\mu)$.

With these marginal measures, we construct the following $\bm{b}$-decomposition scheme (\Cref{definition:decomposition-scheme}).
Specifically, $\*b = (b_v)_{v\in V}$ is the collection of normalized measures over $\{0,1,\bot\}$ defined as follows:
\begin{equation}\label{eq:base-measure-his}
b_v(c)=
\begin{cases}
    \frac{1}{1+\lambda_v} & \text{if }c=0,\\
    0 & \text{if }c=1,\\
    \frac{\lambda_v}{1+\lambda_v} & \text{if }c=\bot.
\end{cases}
%b_v(0)=\frac{1}{1+\lambda_v}, \quad b_v(1)=0,\quad\text{ and }\quad b_v(\bot)=\frac{\lambda_v}{1+\lambda_v}.
\end{equation}
Then, for any $v\in V$ and independent set $\tau \in \{0,1\}^{V\backslash \{v\}}$, the excess marginal measure $\mu_v^{\tau, \bot}$ is given by:
\begin{itemize}
\item 
If $\tau$ remains  an independent set after including $v$, then 
\[
\mu_v^{\tau, \bot}(0)=0 \quad\text{ and }\quad  \mu_v^{\tau, \bot}(1)=1.
\]
\item 
Otherwise, if $\tau$ is no longer an independent set after including $v$, then
\[
\mu_v^{\tau, \bot}(0)=1 \quad\text{ and }\quad  \mu_v^{\tau, \bot}(1)=0.
\]
\end{itemize}
%
% we have
% \[\mu_v^{\tau, \bot}(0)=0 \quad\text{ and }\quad  \mu_v^{\tau, \bot}(1)=1,\]
% when setting $\sigma(v) \gets 1$ results in an independent set and
% \[\mu_v^{\tau, \bot}(0)=1 \quad\text{ and }\quad  \mu_v^{\tau, \bot}(1)=0,\]
% otherwise. 
%

Using this $\bm{b}$-decomposition scheme, the measure on $r_t=1$ is always $0$, so we can assume that $\*r=(r_t)_{t=-T+1}^{0}\in \{0,\bot\}^T$. Then,
the systematic scan Glauber dynamics for the hypergraph independence polynomial can be expressed as follows in \Cref{Alg:complex-GD-decomposed-his}, which is a specialization of \Cref{Alg:complex-GD-decomposed}.

\begin{algorithm}[H]
\caption{Systematic scan Glauber dynamics for hypergraph independence polynomial (under the $\bm{b}$-decomposition scheme specified in \eqref{eq:base-measure-his})} \label{Alg:complex-GD-decomposed-his}
%\SetKwInOut{Instance}{Instance}
\SetKwInOut{Input}{Input}
%\SetKwInOut{Output}{Output}
%\SetKwIF{WP}{ElseIf}{Else}{with probability}{do}{else if}{else}{endif}
%  \SetKwInput{KwPar}{Parameter}
%\Instance{A complex normalized measure $\mu\in \mathbb{C}^{[q]^V}$.}
\Input{ %A collection of complex measures $\bm{b}=(b_v)_{v\in V}$ where $b_v:[q]\cup \{\bot\}\to \mathbb{C}$ is normalized for each $v\in V$, 
An arbitrary independent set $\tau\in\+I(H)$ and an integer $T\ge 1$.} 
%\Output{a pair of assignments $(X,Y)\in [q]^V\times [q]^V$;}
%Initialize $X_0=\tau$ for an arbitrary $\tau\in\supp(\mu)\subseteq[q]^V$\;
Set $\sigma_{-T}\gets\tau$\;
\For{$t=-T+1,-T+2,\ldots,0$}{
let $\sigma_{t}\gets\sigma_{t-1}$ and
$v\gets v_{i(t)}$, where $i(t)=(t \mod n) + 1$\; 
let $r_t$ follow the measure $b_{v}$ specified in \eqref{eq:base-measure-his}\;
\uIf{$r_t\neq \bot$}{let $c_t\gets 0$\;}
\uElseIf{there exists a hyperedge $e\in \+E$ such that $v\in e$ and $\sigma_{t-1}(u)=1$ for all $u\in e\setminus \{v\}$\label{Line:independent-condition}}{
%$\exists e\in \+E \text{ s.t. }v\in e\text{ and }\forall u\in e\setminus \{v\},\sigma_{t-1}(u)=1$\label{Line:independent-condition}}{
      let $c_t\gets 0$\; }
      \Else{
      let $c_t\gets 1$\;
      }
update $\sigma_t(v)\gets c_t$\;
}
\end{algorithm}

\subsection{Zero-freeness of hypergraph independence polynomial}
We define the following quantities for a hypergraph $H=(V,\+E)$ with vertex weights $\bm{\lambda}=(\lambda_v)_{v\in V}$, 
for characterizing zero-freeness.

\begin{definition}\label{definition:parameters-his}
Let $H=(V,\+E)$ be a hypergraph with maximum edge size $k$ and maximum vertex degree $\Delta$. 
Let $\bm{\lambda}=(\lambda_v)_{v\in V}\in\mathbb{C}^V$ be the complex vertex weights such that $\lambda_v\neq -1$ for each $v\in V$.

We define the following quantities:
\begin{align*}
    N=N(H,\bm{\lambda})&\defeq \max\limits_{e\in \+E}\prod\limits_{v\in e}\abs{\frac{\lambda_v}{1+\lambda_v}},\\
    M=M(H,\bm{\lambda})&\defeq \max\limits_{v\in V}\left(\abs{\frac{\lambda_v}{1+\lambda_v}}+\abs{\frac{1}{1+\lambda_v}}\right),\\
    \alpha=\alpha(H,\bm{\lambda})&\defeq  N\cdot M^{4\Delta^2k^5}.
\end{align*}
% \begin{itemize}
%     \item The \emph{maximum edge-wise norm} is defined as: 
%     \[N=N(H,\bm{\lambda})\defeq \max\limits_{e\in \+E}\prod\limits_{v\in e}\abs{\frac{\lambda_v}{1+\lambda_v}}.\]
%     \item The \emph{maximum overflow norm} is defined as:
%     \[M=M(H,\bm{\lambda})\defeq \max\limits_{v\in V}\left(\abs{\frac{\lambda_v}{1+\lambda_v}}+\abs{\frac{1}{1+\lambda_v}}\right).\]
% \end{itemize}
\end{definition}

The regime for zero-freeness and convergence of Glauber dynamics, 
as stated in \Cref{theorem:zero-freeness-his-lee-yang,theorem:convergence-lee-yang},
is characterized by these quantities through the following lemma.

%The significance of \Cref{definition:parameters-his} is through the following lemma, that can simultaneously prove \Cref{theorem:zero-freeness-his-lee-yang,theorem:convergence-lee-yang}.

\begin{lemma}[Inductive step]\label{lemma:unified}
Let $H=(V,\+E)$ be a hypergraph with maximum edge size $k$ and maximum vertex degree $\Delta$. 
Let $\bm{\lambda}=(\lambda_v)_{v\in V}\in\mathbb{C}^V$ be the complex vertex weights such that $\forall v\in V, \lambda_v\neq -1$.
Suppose:
\begin{equation}\label{eq:bound-alpha}
8\mathrm{e}\Delta^2k^4\cdot\alpha<1.
\end{equation}
Then, we have:
\begin{enumerate}
    \item\label{item:unified-1} 
    % if $Z_H^\-{ly}(\*\lambda)\neq 0$, then 
    %\todo{yyx: I think this needs that the partition function is non-zero.}
    \Cref{condition:convergence} holds for the systematic scan Glauber dynamics described in \Cref{Alg:complex-GD-decomposed-his}.
%    the systematic scan Glauber dynamics for the complex measure associated with the independence polynomial $Z^\-{ly}_H(\*\lambda)$ converges; 
 \item \label{item:unified-2}
 \Cref{cond:his-lee-yang-marginal}  holds for the hypergraph $H=(V,\+E)$ with complex vertex weights $\*\lambda=(\lambda_v)_{v\in V}$. 
 %$Z^\-{ly}_{H}(\bm{\lambda})\neq 0$.
\end{enumerate}
\end{lemma}

The reason \Cref{lemma:unified} is referred to as the ``inductive step'' will be clarified in its following application to proving  our main results, \Cref{theorem:zero-freeness-his-lee-yang,theorem:convergence-lee-yang}, where the lemma will be used to carry out the inductive step and establish zero-freeness of the partition function $Z^\-{ly}_H(\*\lambda)$.
%We then show that our main results, \Cref{theorem:zero-freeness-his-lee-yang,theorem:convergence-lee-yang}, both follow from the inductive step \Cref{lemma:unified}.

%We first show that \Cref{lemma:unified} implies all of \Cref{theorem:zero-freeness-his-lee-yang,theorem:zero-freeness-his-fisher,theorem:convergence-lee-yang}.

\begin{proof}[Proofs of \Cref{theorem:zero-freeness-his-lee-yang} and \Cref{theorem:convergence-lee-yang}]
We first verify that the conditions in \Cref{theorem:zero-freeness-his-lee-yang} satisfy \eqref{eq:bound-alpha}.
%Then, \Cref{theorem:zero-freeness-his-lee-yang,theorem:convergence-lee-yang}  follows directly from \Cref{lemma:unified}-(\ref{item:unified-2}) and \Cref{lemma:unified}-(\ref{item:unified-1}), respectively. 
Assume the condition in \Cref{theorem:zero-freeness-his-lee-yang}. It is already ensured that $\lambda_v\neq -1$ for each $v\in V$.

First, we bound $N$.
For any $\lambda \in \+D_{\varepsilon}$, let $\lambda^* \in [0,\lambda_{c, \varepsilon}]$ be the nearest point to $\lambda$. We have that
\[
\abs{\frac{\lambda}{1 + \lambda}} \le \frac{\lambda^* + \varepsilon}{1 + \lambda^* - \varepsilon} \le \left(2 \sqrt{2} \mathrm{e} \Delta k^2\right)^{-2/k},
\]
where the first inequality is due to the triangle inequality, and the second is due to conditions in \Cref{theorem:zero-freeness-his-lee-yang}.
This implies $N\le \left(2 \sqrt{2} \mathrm{e} \Delta k^2\right)^{-2}$.

Then, we bound $M$. For any $\lambda \in \+D_{\varepsilon}$, let $\lambda^* \in [0,\lambda_{c, \varepsilon}]$ be the nearest point to $\lambda$. We have that
\[
\abs{\frac{\lambda}{1 + \lambda}} + \abs{\frac{1}{1 + \lambda}} = \frac{1 + \abs{\lambda}}{\abs{1 + \lambda}} \le \frac{1 + \lambda^* + \varepsilon}{1 + \lambda^* - \varepsilon} \le \frac{1 + \varepsilon}{1 - \varepsilon},
\]
which implies $M \le \frac{1 + \varepsilon}{1 - \varepsilon}$.

Recall $\alpha=N \cdot M^{4 \Delta^2 k^5}$. Now, we can verify the condition in \eqref{eq:bound-alpha}:
\begin{align*}
   8\mathrm{e}\Delta^2k^4 \cdot \alpha
   %8 \mathrm{e} \Delta^2 k^4 \cdot N \cdot M^{4 \Delta^2 k^5} 
   = 8 \mathrm{e} \Delta^2 k^4 \left(2\sqrt{2} \mathrm{e} \Delta k^2\right)^{-2} \left(\frac{1 + \varepsilon}{1 - \varepsilon}\right)^{4\Delta^2 k^5} 
    \le \mathrm{e}^{-1} \exp\left(\frac{2 \varepsilon}{1 -\varepsilon} 4 \Delta^2 k^5\right) 
    < 1.
\end{align*}
This shows \eqref{eq:bound-alpha} is satisfied. By \Cref{lemma:unified}, \Cref{condition:convergence} and \Cref{cond:his-lee-yang-marginal} hold directly. 
% and concludes the proof of \Cref{theorem:convergence-lee-yang}.

%\end{proof}

% \begin{lemma}\label{lemma:his-lee-yang-marginal}

% Suppose \eqref{eq:bound-alpha} holds for $(H=(V,\+E),\*\lambda)$. Let $\+E=\{e_1,e_2,\dots,e_m\}$.  For each $0\leq i\leq m$, let $\+E_i=\{e_1,e_2,\dots,e_i\}$ and let $H_i=(V,\+E_i)$.
% For each $0\leq i<m$, if $Z^\-{ly}_{H_{i}}(\*\lambda)\neq 0$, then
% \[
% \abs{\mu_{H_i,\*\lambda}\left(\sigma(e_{i+1}) =  1^{e_{i+1}}\right)}< 1.
% \]
% \end{lemma}

%We then show that \Cref{lemma:his-lee-yang-marginal} already implies \Cref{lemma:unified}-(\ref{item:unified-2}).
%\begin{proof}[Proof of \Cref{lemma:unified}-(\ref{item:unified-2})]
We then use \Cref{cond:his-lee-yang-marginal} to prove \Cref{theorem:zero-freeness-his-lee-yang}. 
This is by induction, using edge-wise self-reducibility.
Recall the $\+E=\{e_1,e_2,\dots,e_m\}$, $\+E_i=\{e_1,e_2,\dots,e_i\}$ and $H_i=(V,\+E_i)$, for each $0\leq i\leq m$, defined in \Cref{cond:his-lee-yang-marginal}. 
% Note that \eqref{eq:bound-alpha} holds for $(H,\*\lambda)=(H_m,\*\lambda)$ implies that \eqref{eq:bound-alpha} holds for all $(H_i,\*\lambda)$.
Now we prove by induction that $Z^\-{ly}_{H_i}(\*\lambda)\neq 0$ for each $0\leq i\leq m$.

For the induction basis, 
we have $Z^\-{ly}_{H_0}(\bm{\lambda})=\prod_{v \in V}(1+\lambda_v)\neq 0$,
since $\lambda_v\neq -1$ for all $v \in V$.

For the induction step, fix $0\leq i<m$ and assume that $Z^\-{ly}_{H_{i}}(\*\lambda)\neq 0$. 
Then the measure $\mu_{H_i,\*\lambda}$ is well-defined. Also, when $\lambda_v\neq -1$ for each $v\in V$, the systematic scan Glauber dynamics in \Cref{Alg:complex-GD-decomposed-his} is well-defined.
We further note that
\[
\frac{Z^\-{ly}_{H_{i+1}}(\*\lambda)}{Z^\-{ly}_{H_{i}}(\*\lambda)}=\frac{Z^\-{ly}_{H_{i}}(\*\lambda)-\sum\limits_{\substack{\sigma\in \{0,1\}^{\abs{V}}\\ \sigma(e_{i+1})=1^{e_{i+1}}}} \prod\limits_{v:\sigma(v)=1} \lambda_v }{Z^\-{ly}_{H_{i}}(\*\lambda)}=1-\mu_{H_i,\*\lambda}\left(\sigma(e_{i+1})=1^{e_{i+1}}\right)\neq 0,
\]
% \todo{"verify"}
where the last inequality follows from \Cref{cond:his-lee-yang-marginal}. Therefore, $Z^\-{ly}_{H_{i+1}}(\bm{\lambda})\neq 0$, completing the induction step and therefore proving \Cref{theorem:zero-freeness-his-lee-yang}.

Finally, \Cref{theorem:convergence-lee-yang} holds directly assuming \Cref{condition:convergence}, according to \Cref{lemma:convergence}.
\end{proof}

\subsection{Information percolation on the witness graph}
In the remainder of this section,
we will focus on the proof of \Cref{lemma:unified}.
To achieve this, we will construct an information percolation argument to analyze the complex systematic scan Glauber dynamics described in \Cref{Alg:complex-GD-decomposed-his}.

A key step in our analysis is to bound the marginal measure. Our idea is to split the contribution to the marginal measure into two types: those coming from \emph{small bad trees} and those coming from \emph{large bad trees}, which we formally define later. Then our proof strategy consists of the following three steps:
\begin{enumerate}
    \item Prove that small bad trees characterize witness sequences (corresponding to oblivious updates);
    \item Bound the contribution from small bad trees;
    \item Prove that contribution from large bad trees diminishes to zero under a suitable limit.
\end{enumerate}
These three steps correspond to \Cref{lemma:characterization,lemma:modulus-bound-subset,lemma:diminishing-bad-tree} respectively.
We setup the formulations necessary for a proof of  \Cref{lemma:unified}, and defer the proofs of these lemmas to \Cref{sec:proof}.
%In this subsection, we will discuss how the outcome $\sigma_0$ of complex-valued $\bm{b}$-decomposed systematic scan Glauber dynamics depends on the structure of $\bm{r}$. We will need the assumption that $\mu$ is a measure defined by a graphical model $\+G=(H=(V,\+E),(\phi_v)_{v \in V}, (\phi_e)_{e \in \+E})$ so that the following conditional independence property holds. 

%\begin{fact}[conditional independence]\label{fact:conditional independence}
%Let $\mu$ be a measure defined by a graphical model $\+G=(H=(V,\+E),(\phi_v)_{v \in V}, (\phi_e)_{e \in \+E})$. Then for any $v\in V$, any $\sigma\in [q]^{V\setminus \{v\}}$, it holds that
%\[
%\mu^{\sigma}_v=\mu^{\sigma_{\Lambda(v)}}_v,
%\]
%where $\Lambda(v)\defeq \{u\in V\setminus \{v\}\mid \exists e\in \+E \text{ s.t. }u,v\in e\}$ is the neighborhood of $v$.
%\end{fact}

\bigskip
For any $u\in V$ and integer $t$, we denote by $\upd_u(t)$ the last time before $t$ at which $u$ is updated, i.e.
\begin{align*}
\upd_u(t)\defeq \max\{s \mid s\leq t \hbox{ such that } v_{i(s)}=u\}.
\end{align*}
For any subset of vertices $U\subseteq V$ and $t\in \mathbb{Z}_{\leq 0}$, define
\begin{align*}
		\ts(U,t)\defeq\{\upd_v(t)\mid v\in U\}
\end{align*}
as the collection of ``\emph{timestamps}'' of the latest updates of vertices in  $U$ up to time $t$.

Recall the definition of witness sequences from \Cref{definition:witness-sequence}. 
For an event $A\subseteq \{0,1\}^V$, we try to characterize the witness sequence with respect to $A$. 
Let $S=\vbl(A)\subseteq V$ denote the set of variables on which the event $A$ is defined. Formally,
\[
\vbl(A)\defeq\{v\in V\mid \exists \sigma\in A \text{ s.t. }\sigma'\not\in A\text{ where }\sigma=\sigma'\text{ except at }v\}.
\]
If $r_t\neq \bot$ for all $t\in \ts(S,0)$, 
we can directly infer whether $\sigma_0\in A$ regardless of the initial state, 
indicating that $\bm{r}$ is a witness sequence. 
Otherwise, according to \Cref{Alg:complex-GD-decomposed-his}, for those $v\in S$ with $r_{\upd_v(0)}=\bot$, we need to determine if the condition in \Cref{Line:independent-condition} holds at time $\upd_v(0)$ for $\bm{r}$ to qualify as a witness sequence. 
The argument can be applied recursively, allowing the characterization of witness sequences to percolate through the time-space structure of the systematic scan Glauber dynamics. As such, we focus on keeping track of the complex measure of ``percolation of assigning $r_t=\bot$''.

To formalize this argument, we introduce the definition of a \emph{witness graph}, 
a combinatorial structure shown to be useful for analyzing Glauber dynamics~\cite{HSZ19,HSW21,qiu2022perfect,feng2023towards}. 

%To formalize this characterization, we need to introduce the definition of \emph{witness graph}, which is a combinatorial structure useful for analyzing the systematic scan Glauber dynamics~\cite{HSW21,qiu2022perfect,feng2023towards}. 

\begin{definition}[witness graph/space-time slab]\label{definition:witness-graph-indset}
Given a hypergraph $H=(V,\+E)$  and a subset of variables $S \subseteq V$, the \emph{witness graph} $G^S_H=\left(V^S_H,E^S_H\right)$ is an infinite graph with the vertex set 
\[
V^S_H=\{\ts(e,t)\mid e\in \+E,t\in \mathbb{Z}_{\leq 0}\}\cup \{\ts(S,0)\},
\] 
and $E^S_H$ consists of undirected edges between vertices $x,y \in V^S_H$ such that $x\neq y$ and $x\cap y\neq\emptyset$.
\end{definition}

%In this section, we prove the Lee-Yang zero-free region and Fisher zero-free region for hypergraph independent sets, specifically referring to  \Cref{theorem:zero-freeness-his-lee-yang} and \Cref{theorem:zero-freeness-his-fisher}. We first define the bad component and bad tree on the witness graph (\Cref{definition:witness-graph-indset}). We use them to trace the information percolation process.
% And for the Fisher zero-free region, we use a reduction from Fisher zero-free region to Lee-Yang zero-free region for hypergraph independent set.

%\subsection{Bad components and bad trees}
% \yxtodo{I am going to modify the notations in this subsection to follow new definitions.}
%In this subsection, we work on bad vertices, bad components and bad trees on the witness graph $G_H^S = (V_H^S, E_H^S)$. They are special cases of the aforementioned witness components and witness $2$-trees (\Cref{definition:framework-witness-component-2-trees}). The differences are that we only consider $q=2$, $\abs{S}\le k$ where $k$ is the maximum edge size and the $r$'s of in bad component are all $\bot$.

The following structural property of the witness graph has been established in \cite{feng2023towards}.

\begin{lemma}[{\cite[Corollary 6.15]{feng2023towards}}]\label{lemma:witness-graph-degree-bound} 
Assume the hypergraph $H=(V, \+E)$ has a maximum degree $\Delta$ and a maximum edge size $k$. 
Then, in the witness graph $G^S_H=\left(V^S_H,E^S_H\right)$, for any $v\in V_H^S\setminus \{\ts(S,0)\}$, the degree of $v$ is at most $2\Delta k^2-2$. Furthermore, the degree of \;$\ts(S,0)$ is at most $2\Delta k|S|-1$.
\end{lemma}

We formalize the following notions of bad structures within the witness graph.
\begin{definition}[bad vertices, bad components, and bad trees]\label{definition:bad-vertices-his}
Let $H=(V,\+E)$ be a hypergraph, and let $G_H^S = (V_H^S, E_H^S)$ be the witness graph as in \Cref{definition:witness-graph-indset}.
Let $T\geq 1$, and let $\bm{r}=(r_t)_{t=-T+1}^{0}\in\{0,\bot\}^T$. %where each $r_t\in \{0,1\}\cup \{\bot\}$.
%We define the following quantities related to the witness graph $G_H^S = (V_H^S, E_H^S)$ (\Cref{definition:witness-graph-indset}) and $\bm{r}=(r_t)_{t=-T+1}^{0}$ where each $r_t\in [q]\cup \{\bot\}$.
%We defined the following concepts:
\begin{itemize}
\item The set of \emph{bad vertices} $V^{\-{bad}}=V^{\-{bad}}(\bm{r})$ is defined as:
%let the set of \emph{bad vertices} $V^{\-{bad}}=V^{\-{bad}}(\bm{r})$ in $G^S_H$ be the set of vertices such that all $r_t$ evaluates to $\bot$ for every timestamp $t$ included, i.e.,  
\[
V^{\-{bad}}\defeq \{v\in V^S_H\mid\forall t\in v, -T+1\leq t\leq 0 \text{ and }r_t=\bot  \} \cup \{\ts(S, 0)\},
\]
which contains $\ts(S, 0)$ and the vertices in the witness graph such that all $r_t$ evaluates to $\bot$ for every timestamp $t$ associated with that vertex.
\item Let $G^S_H\left[V^{\-{bad}}\right]$ be the subgraph of the witness graph $G^S_H$ induced by $V^{\-{bad}}$.
\item The \emph{bad component}  $\+C^{\-{bad}}=\+C^{\-{bad}}(\bm{r})\subseteq V^{\-{bad}}$ is defined as the maximal set of vertices in $V^{\-{bad}}$ containing $\ts(S,0)$ that is connected in $G^S_H\left[V^{\-{bad}}\right]$. 
% Let $\+C^{\-{bad}}=\emptyset$ if $\ts{}(S,0)\notin V^{\-{bad}}$.
\item The \emph{bad tree} $\+T^{\-{bad}}=\+T^{\-{bad}}(\bm{r})\subseteq V^{\-{bad}}$ is defined as the $2$-tree of the induced subgraph $G^S_H\left[\+C^{\-{bad}}\right]$ containing $\ts(S,0)$, constructed deterministically using \Cref{definition:2-tree-construction}. 
We further denote this deterministic construction as a mapping $\=T$ from the bad component such that $\tbad=\=T(\cbad)$.
\end{itemize}
\end{definition}

At first glance, the constructions of these structures may seem technically involved and uneasy to decipher.
However, the intuition behind them is quite clear. 
Specifically, the bad tree $\+T^{\-{bad}}$ acts as a ``certificate'' for the undesirable situation 
where the percolation process, starting from time 0, actually reaches the initial time $-T$.
Consequently, the event at time $0$ cannot be successfully inferred solely from the ``randomness'' $\bm{r}$ used in the oblivious transitions in the Glauber dynamics in \Cref{Alg:complex-GD-decomposed-his}.

%With this structure of bad (2-)trees in the witness graph, 
For two measurable events $A$ and $B$, we say $A$ is \emph{determined by} $B$ if either $A\supseteq B$ or $A \cap B=\emptyset$.  We can then characterize the witness sequences (\Cref{definition:witness-sequence}) for an arbitrary event as follows.
\begin{lemma}[characterization of witness sequences]\label{lemma:characterization}
Fix any $T\geq n$ and any event $A\subseteq \{0,1\}^V$. 
Consider the witness graph $G_H^S = \tp{V_H^S, E_H^S}$ constructed using the set of variables $S=\vbl(A)$.
Let $\bm{\rho}=(\rho_t)_{t=-T+1}^{0}\in\{0,\bot\}^T$. 
%be a sequence where each $\rho_t\in \{0,1,\bot\}$. 
If the corresponding bad tree $\tbad(\bm{\rho})$ in the witness graph $G_H^S$ satisfies: 
\[
\left|\tbad(\bm{\rho})\right|\leq \frac{T}{2n}-2,
\]
where $n=|V|$, then the following holds:
\begin{enumerate}
    \item  $\bm{\rho}\Rightarrow A$, i.e., $\bm{\rho}$ is a witness sequence with respect to the event $A$;\label{item:characterization-1}
    \item for any initial  $\sigma\in \supp(\mu)$, the occurrence of $A$ at time $0$ is determined by $\cbad(\bm{\rho})$ and $\bm{\rho}_{\ts(S,0)}$.%\[
%mu^{\-{GD}}_{T,\sigma,\bm{b}}(\sigma_0\in A\mid \cbad(\bm{r})=\cbad(\bm{\rho})\color{red}\land \bm{r}_{\ts(S,0)}=\bm{\rho}_{\ts(S,0)} \color{black})\in \{0,1\}
%\]
%\color{red}when well-defined, i.e., when $\mu^{\-{GD}}_{T,\sigma,\bm{b}}\tp{\cbad(\bm{r})=\cbad(\bm{\rho})\color{red}\land \bm{r}_{\ts(S,0)}=\bm{\rho}_{\ts(S,0)}}\neq 0$.\color{black}
\label{item:characterization-2}
\end{enumerate}
\end{lemma}
The proof of \Cref{lemma:characterization} is deferred to \Cref{section:small-badtrees}.

\begin{remark}[zero-one law]\label{remark:zero-one-law}
    \Cref{lemma:characterization}-(\ref{item:characterization-2}) implies the following ``zero-one law'' for the measure of an event $A$ at time 0:
    \begin{equation}\label{eq:zero-one-law-small-trees}
    \mu^{\-{GD}}_{T, \sigma, \*b}\left(\sigma_0\in A \mid \cbad(\*r) = \cbad(\*\rho) \land \*r_{\ts(S, 0)} = \*\rho_{\ts(S, 0)}\right) \in \{0, 1\},
    \end{equation}
    provided the conditional measure is well-defined, i.e., the event $\cbad(\*r) = \cbad(\*\rho)\land\*r_{\ts(S, 0)} = \*\rho_{\ts(S,0)}$ has non-zero measure.
    % $\cbad(\bm{\rho})$ and $\*\rho_{\ts(S,0)}$ are well-defined, i.e., 
    % \[\mu^{\-{GD}}_{T, \sigma, \*b}(\cbad(\*r) = \cbad(\*\rho) \land \*r_{\ts(S, 0)} = \*\rho_{\ts(S,0)})\neq 0, \]
    % it holds that

    This ``zero-one law'' serves as a key tool that enables us to compare complex normalized measures of different events.
    For complex measures $\mu$, the standard monotonicity property $|\mu(A \cap B)|\le |\mu(B)|$ does not generally hold. 
    However, by assuming the ``zero-one law'', where $\mu(A\mid B) \in \{0,1\}$, monotonicity can be recovered as follows:
$$|{\mu(A \cap B)}| = |{\mu(B)}| |\mu(A\mid B)| \le |{\mu(B)}|.$$
    This is crucial for our analysis of zero-freeness.
\end{remark}
\color{black}

Next, we bound the total contribution to the marginal measure coming from a bad tree.
Recall that in \Cref{definition:bad-vertices-his},
we use $\=T$ to denote the deterministic construction of the bad tree $\tbad=\=T(\cbad)$ from a bad component $\cbad$. To do so, we upper bound on the measure of any set of bad components $\cbad$ that could potentially produce the given bad tree $\tbad$ through the function $\tbad=\=T(\cbad)$.

% change form to lemma 5.2
\begin{lemma}\label{lemma:modulus-bound-subset}
Fix any $S\subseteq V$ and $T\geq n$ with $|S|=k$. 
Consider the witness graph $G_H^S$, the bad component $\cbad=\cbad(\*r)$, and the bad tree $\+T^{\-{bad}}=\+T^{\-{bad}}(\bm{r})=\=T(\cbad)$ ,
where $\bm{r}=(r_t)_{t=-T+1}^{0}\in\{0,\bot\}^T$ is constructed as in \Cref{Alg:complex-GD-decomposed-his}, following the product measure in~\eqref{eq:base-measure-his}.

Then, for any $\sigma\in \supp(\mu)$, for any finite  $2$-tree $\+T$ in $G^S_H$ containing $\ts(S,0)$, we have
\begin{equation}\label{eq:2-tree-modulus-bound-subset}
\sum\limits_{\+C \in \=T^{-1}(\+T)}\abs{\mu^{\-{GD}}_{T,\sigma,\bm{b}}\tp{\+C^{\-{bad}}= \+C \land \*r_{\ts(S,0)} = \bot^{\ts(S, 0)}}}\leq \alpha^{|\+T|},%\quad \forall \+C^*\subseteq \=T^{-1}(\+T).
\end{equation}
where $\*r_{\ts(S,0)} = \bot^{\ts(S, 0)}$ represents the event that $r_t=\bot$ for all timestamps $t\in \ts(S,0)$.
%corresponds to the restriction of $\*r$ on the timestamps in $\ts(S,0)$ and $\bot^{\ts(S, 0)}$ represents the sequence that assigns every timestamp in $\ts(S,0)$ to $\bot$.
\end{lemma}

The proof of this lemma is deferred to \Cref{sec:bound-small-badtrees}.

We have the following bounds on the number of possible bad trees with a given size $i$.

\begin{lemma}\label{lemma:bad-tree-num-bound}
Let $\+T_i$ denote the set of possible $2$-trees  of size $i$ in $G^S_H$ containing $\ts(S,0)$. Then letting $D_1=4\Delta^2k^4$ and $D_2=4\Delta^2k^3|S|$, we have
\begin{equation}\label{eq:bad-tree-bound-large}
\abs{\+T_i}\leq (\mathrm{e}(D_2 + i - 2))^{D_2 - 1} \cdot (\mathrm{e}D_1)^{i-1},
\end{equation}
Also, when $|S|\leq k$, we have a refined bound that
\begin{equation}\label{eq:bad-tree-bound-small}
|\+T_i|\leq (\mathrm{e}D_1)^{i-1}.
\end{equation}
\end{lemma}
\begin{proof}
Note that by \Cref{definition:2-tree}, each possible $\+T\in \+T_i$ satisfies:
\begin{itemize}
    \item $\+T$ contains $\ts(S,0)$;
    \item $\+T$ does not contain any $v\in V^S_H$ such that $(\ts(S,0),v)\in E^S_H$;
    \item $\+T$ is connected in $\left(G_H^S\right)^2$, the square graph of $G_H^{S}$.
\end{itemize}

Recall that by \Cref{lemma:2-tree-number-bound}, for a graph with maximum degree $d$, the number of subtrees of size $i\ge 1$ containing a fixed vertex is upper bounded by $(\mathrm{e}d)^{i-1}$. By \Cref{lemma:witness-graph-degree-bound}, we have that all vertices, except those within distance $1$ of $\ts(S,0)$ in $G^S_H$, have a degree at most $D_1=4\Delta^2k^4$ in $\left(G_H^{S}\right)^2$. Also, $\ts(S,0)$ has a degree of at most $D_2=4\Delta^2k^3|S|$ in $\left(G_H^S\right)^2$. Then, we can bound the size of $\+T_i$ as:
\[
    |\+T_i|\leq \binom{D_2+i-2}{D_2 - 1}\cdot (\mathrm{e}D_1)^{i-1}\leq \left(\frac{\mathrm{e}(D_2+i-2)}{D_2 - 1}\right)^{D_2-1}\cdot (\mathrm{e}D_1)^{i-1}\leq (\mathrm{e}(D_2 + i - 2))^{D_2 - 1} \cdot (\mathrm{e}D_1)^{i-1},
\]
which finishes the proof of \eqref{eq:bad-tree-bound-large}. Here, the first inequality follows by enumerating the size of the subtree on each neighbor of $\ts(S, 0)$ and applying \Cref{lemma:2-tree-number-bound} to each neighbor of $\ts(S, 0)$ in $\left(G_H^S\right)^2$.

When $|S|\leq k$, by \Cref{lemma:witness-graph-degree-bound}, the maximum degree of $G^S_H$ is upper bounded by $2\Delta k^2$. Therefore, \eqref{eq:bad-tree-bound-small} is a direct consequence of \Cref{lemma:2-tree-number-bound}.
\end{proof}

We still need one more technical lemma that ensures the decay of percolation.
Recalling \Cref{lemma:characterization}\nobreakdash-(\ref{item:characterization-1}), small bad trees characterize witness sequences. The following lemma says that that under \eqref{eq:bound-alpha}, large bad trees have diminishing contributions to the marginal measure. 
 
 \begin{lemma}\label{lemma:diminishing-bad-tree}
Fix any $S\subseteq V$ and $T\geq n$. 
Consider the witness graph $G_H^S$ and the bad tree $\+T^{\-{bad}}=\+T^{\-{bad}}(\bm{r})$,
where $\bm{r}=(r_t)_{t=-T+1}^{0}\in\{0,\bot\}^T$ is constructed as in \Cref{Alg:complex-GD-decomposed-his}, following the product measure in~\eqref{eq:base-measure-his}.
Under the condition of $\alpha$ in \eqref{eq:bound-alpha}, it holds for any initial configuration $\sigma\in \supp(\mu)$ and any $\tau\in \{0,1\}^V$,
\begin{equation}\label{eq:bad-tree-growth-norm-bound}
    \lim\limits_{T\to \infty}\abs{\mu^{\-{GD}}_{T,\sigma,\bm{b}}\tp{\sigma_0=\tau\land |\tbad|>T/(2n)-2}}=0.
\end{equation}
 \end{lemma}

 %The proof of \Cref{lemma:diminishing-bad-tree} follows a similar strategy as in the proofs of \Cref{lemma:characterization} and~\ref{lemma:modulus-bound-subset}-(\ref{item:modulus-bound-2}). A key technical challenge is that we can no longer certify the outcome using only the information of the bad component. However, we will demonstrate that the information from the initial configuration, the bad component, and the $r_t$ values for the first $n$ timestamps, is still sufficient to serve as a certificate for the percolation process. 
 The proof of \Cref{lemma:diminishing-bad-tree} is deferred to \Cref{sec:decay-large-bad-tree}.

We are now ready to establish the inductive step for the convergence of Glauber dynamics and the zero-free region of the hypergraph independence polynomial, which is the content of \Cref{lemma:unified}.
\subsection{Establishing the inductive step (\Cref{lemma:unified})}
We start with the proof of \Cref{lemma:unified}-(\ref{item:unified-1}).

\begin{proof}[Proof of \Cref{lemma:unified}-(\ref{item:unified-1})]

% Let $S'\subseteq V^{S}_{H}$ be a finite subset of vertices in $G^S_H$  and configuration $\tau \in \{0,1\}^{S'}$ on it such that $\mu^{\-{GD}}_{\infty}\left(\+C^{\-{bad}}=S', \sigma(S') = \tau\right)\neq 0$. Then, for arbitrary $\tau \in \{0,1\}^{S'}$ it holds that
% We claim that the partition function is non-zero which we will prove in \Cref{section:zero-free-proof}. So the stationary measure exists and the systematic scan Glauber dynamics (\Cref{Alg:complex-GD}) is well defined.

It suffices to verify \Cref{condition:convergence}. We take $B(T)$ in \Cref{condition:convergence} as the following event on $\*r$:
\begin{equation}\label{eq:choice-B}
B(T): \quad \abs{\tbad(\* r)}>\frac{T}{2n}-2.
\end{equation}

According to \Cref{lemma:characterization}, the definition of $B(T)$ in \eqref{eq:choice-B} indeed contains all non-witness sequences, verifying the first item of \Cref{condition:convergence}. The second item of \Cref{condition:convergence} follows from \Cref{lemma:diminishing-bad-tree}.  
\end{proof}

Next, we prove \Cref{lemma:unified}-(\ref{item:unified-2}).

\begin{proof}[Proof of \Cref{lemma:unified}-(\ref{item:unified-2}).]
% \textcolor{blue}{
Throughout the proof, we use $N,M,\alpha$ to denote $N(H,\*\lambda)$, $M(H,\*\lambda)$ and $\alpha(H,\*\lambda)$, where $H=H_m$, respectively. 
Fix an arbitrary $0\leq i<m$. Let $\mu=\mu_{H_i,\*\lambda}$.
Assuming $Z^\-{ly}_{H_{i}}(\*\lambda)\neq 0$, $\mu$ is well-defined. Then, our goal is to prove that,  the marginal measure on $e_{i+1}$ being assigned all-one, denoted by $\mu_{e_{i+1}}\left( 1^{e_{i+1}}\right)$, is bounded away from $1$. 
%
%
%Consider the complex systematic scan Glauber dynamics from stationarity with $T=\infty$. 
Consider the complex systematic scan Glauber dynamics on $H_i$. Fix any initial configuration $\sigma\in \supp(\mu)$. Note that from \Cref{lemma:convergence} and \Cref{lemma:unified}-(\ref{item:unified-1}), we have
% \todo{If u Lemma 4.3 }
\begin{equation}\label{eq:gdss-convergence}
\lim\limits_{T\to \infty}\mu^{\-{GD}}_{T,\sigma,\*b}(\sigma_0(e_{i+1})=1^{e_{i+1}})=\mu_{e_{i+1}}\left( 1^{e_{i+1}}\right).
\end{equation}
Therefore, it suffices to show that $\abs{\lim\limits_{T\to \infty}\mu^{\-{GD}}_{T,\sigma,\*b}(\sigma_0(e_{i+1})=1^{e_{i+1}})}<1$. 
Next, we notice that for $x\in \=C$, the function $\abs{x}$ is a continuous function.
Because the limit in~\eqref{eq:gdss-convergence} exists, we have
\[
\abs{\lim\limits_{T\to \infty}\mu^{\-{GD}}_{T,\sigma,\*b}(\sigma_0(e_{i+1})=1^{e_{i+1}})} = \lim\limits_{T\to \infty}\abs{\mu^{\-{GD}}_{T,\sigma,\*b}(\sigma_0(e_{i+1})=1^{e_{i+1}})} .
\]
Therefore, it suffices to show that
\begin{equation}\label{eq:absolute-convergence}
%\lim\limits_{T\to \infty}\abs{\mu^{\-{GD}}_{T,\sigma,\*b}(\sigma_0(e_{i+1})=1^{e_{i+1}})} \text{ exists, and }
\lim\limits_{T\to \infty}\abs{\mu^{\-{GD}}_{T,\sigma,\*b}(\sigma_0(e_{i+1})=1^{e_{i+1}})}<1. 
\end{equation}
Note that by combining \eqref{eq:gdss-convergence} with \eqref{eq:absolute-convergence}, the lemma is proved. It remains to prove \eqref{eq:absolute-convergence}.\\
We set $S = e_{i+1}$. Recall the witness graph $G^S_{H_i}$ (\Cref{definition:witness-graph-indset}) and related definitions in \Cref{definition:bad-vertices-his}. For each $i\geq 1$, let $\+T_i$ denote the set of possible $2$-trees in $G^S_H$ containing $\ts(S,0)$. Recall the definition of the bad component $\cbad=\cbad(\*r)$ in \Cref{definition:bad-vertices-his} and that we use $\=T$ to denote the construction in \Cref{definition:2-tree-construction} such that $\=T(\cbad)=\tbad$. It is important to note that 
\begin{equation}\label{eq:bounded-marginal-implication}
\sigma_0 (e_{i+1})=1^{e_{i+1}}\implies r_t=\bot\text{ for all }t\in \ts({e_{i+1}},0). 
\end{equation}

To see this, recall that in \Cref{Alg:complex-GD-decomposed-his},
in order to update a value to $1$ at time $t$, $r_t=\bot$.
By \Cref{lemma:diminishing-bad-tree}, 
\begin{equation}\label{eq:large-bad-tree-convergence}
\begin{aligned}
&\lim\limits_{T\to \infty}\abs{\mu^{\-{GD}}_{T,\sigma,\*b}(\sigma_0(e_{i+1})=1^{e_{i+1}})}\\
\le & \lim\limits_{T\to \infty}\abs{\mu^{\-{GD}}_{T,\sigma,\*b}\tp{\sigma_0(e_{i+1})=1^{e_{i+1}}\land \abs{\tbad}\leq \frac{T}{2n}-2}}\\
\leq &\lim\limits_{T\to \infty}\sum\limits_{j=1}^{\lfloor \frac{T}{2n}-2\rfloor}\sum\limits_{\+T\in \+T_j}\abs{\mu^{\-{GD}}_{T,\sigma,\bm{b}}\tp{\sigma_0(e_{i+1})=1^{e_{i+1}}\land  \tbad=\+T}}\\
\text{(by \eqref{eq:bounded-marginal-implication})}\quad = & \lim\limits_{T\to \infty}\sum\limits_{j= 1}^{\lfloor \frac{T}{2n}-2\rfloor}\sum\limits_{\+T\in \+T_j}\abs{\mu^{\-{GD}}_{T,\sigma,\bm{b}}\tp{\sigma_0(e_{i+1})=1^{e_{i+1}}\land \bm{r}_{\ts(S,0)}=\bot^{\ts(S,0)}\land  \tbad=\+T}}\\
=&\lim\limits_{T\to \infty}\sum\limits_{j= 1}^{\lfloor \frac{T}{2n}-2\rfloor}\sum\limits_{\+T\in \+T_j}\sum\limits_{\+C\in \=T^{-1}(\+T)} \left\vert\mu^{\-{GD}}_{T,\sigma,\bm{b}}\tp{\cbad=\+C\land \bm{r}_{\ts(S,0)}=\bot^{\ts(S,0)}}\right.\\
&\quad\quad\left.\cdot \mu^{\-{GD}}_{T,\sigma,\bm{b}}\tp{\sigma_0(e_{i+1})=1^{e_{i+1}}\mid \cbad=\+C\land \bm{r}_{\ts(S,0)}=\bot^{\ts(S,0)}}\right\vert.
\end{aligned}
\end{equation}
\sloppy
Recall the zero-one law stated in \eqref{eq:zero-one-law-small-trees}.
For any finite $2$-tree $\+T$ that includes  $\ts(S,0)=\ts(e_{i+1},0)$ and satisfies $\abs{\+T}\leq  \frac{T}{2n}-2$, the conditional measure $\mu^{\-{GD}}_{T, \sigma,\bm{b}}\tp{\sigma_0(e_{i+1})=1^{e_{i+1}}\mid \cbad=\+C\land \bm{r}_{\ts(S,0)}=\bot^{\ts(S,0)}}$ evaluates to either $0$ or $1$ for each $\+C\in \mathbb{T}^{-1}(\+T)$, provided it is well-defined. Then, we have
% We then denote by $\+C^{T, \sigma, S}(\+T) = \+C^{T, \sigma, e_{i+1}}(\+T)$ as the set of bad components $\+C\in\mathbb{T}^{-1}(\+T)$ such that $\mu^{\-{GD}}_{T, \sigma ,\bm{b}}\tp{\sigma_0({e_{i+1}})=1^{e_{i+1}}\mid \cbad=\+C}=1$. 

\begin{align*}%\label{eq:unified-2-proof}
 &\lim\limits_{T\to \infty}\sum\limits_{j= 1}^{\lfloor \frac{T}{2n}-2\rfloor}\sum\limits_{\+T\in \+T_j}\sum\limits_{\+C\in \=T^{-1}(\+T)} \left\vert\mu^{\-{GD}}_{T,\sigma,\bm{b}}\tp{\cbad=\+C\land \bm{r}_{\ts(S,0)}=\bot^{\ts(S,0)}}\right.\\
&\quad\quad\left.\cdot \mu^{\-{GD}}_{T,\sigma,\bm{b}}\tp{\sigma_0(e_{i+1})=1^{e_{i+1}}\mid \cbad=\+C\land \bm{r}_{\ts(S,0)}=\bot^{\ts(S,0)}}\right\vert\\
 (\text{by \eqref{eq:zero-one-law-small-trees}})\quad\leq & \lim\limits_{T\to \infty}\sum\limits_{j= 1}^{\lfloor \frac{T}{2n}-2\rfloor}\sum\limits_{\+T\in \+T_j}\sum\limits_{\+C\in \=T^{-1}(\+T)} \abs{\mu^{\-{GD}}_{T,\sigma,\bm{b}}\tp{\cbad=\+C \land \*r_{\ts(S, 0)}=\bot^{\ts(S, 0)}}}\\
 (\text{\Cref{lemma:modulus-bound-subset}}) \quad  \leq &  \lim\limits_{T\to \infty}\sum\limits_{j= 1}^{\lfloor \frac{T}{2n}-2\rfloor} \sum\limits_{\+T\in \+T_j} \alpha^{j}\\
(\text{\Cref{lemma:bad-tree-num-bound}})\quad \leq & \sum\limits_{j\geq 1} (4\mathrm{e}\Delta^2k^4)^{j-1}\cdot \alpha^{j}\\
(\text{by \eqref{eq:bound-alpha}})\quad < & 1,
\end{align*}
where the second-to-last inequality additionally uses that $k$ upper bounds all $|e_i|$. 
Then, by combining with \eqref{eq:large-bad-tree-convergence}, both \eqref{eq:absolute-convergence} and the lemma are proved.
\end{proof}

\section{Decay of complex measures in  percolation}\label{sec:proof}

In this section, we prove a series of percolation properties (\Cref{lemma:characterization,lemma:modulus-bound-subset,lemma:diminishing-bad-tree}) assumed in the previous section, which controls the contributions of bad trees to the marginal measure. 
From the perspective of information percolation, these lemmas establish a decay of percolation.
%Recall that we use bad trees to keep track of a ``percolation probability'', so this lemma controls the ``decay of percolation''. 

At a high level, we split the contributions to the marginal measure into those coming from small bad trees and those coming from large bad trees. We show that small bad trees characterize witness sequences; therefore, they correspond to ``oblivious updates'' that factorize the measure and do not depend on the initial measure.
Then, we bound the contributions to the marginal measures coming from small bad trees.
Finally, we show that those contributions coming from large bad trees go to zero as $T\to \infty$. Combined, this gives us a bound on the marginal measure through much simpler product measures on the witness graph.

\subsection{Small bad trees characterize witness sequences}
\label{section:small-badtrees}
We first prove the characterization property for small bad trees as in \Cref{lemma:characterization}.
\begin{proof}[Proof of \Cref{lemma:characterization}]
By the definition of $\upd_u(t)$ and $G^S_H$, and the nature of the systematic scan, we have for each $v\in V^S_H$,
\begin{align}
\label{eq:timestamps-max-min}
    \max\{t:t\in v\}-\min\{t:t\in v\}\leq n.
\end{align}

Note that according to the definition of $\tbad(\bm{\rho})$ in \Cref{definition:bad-vertices-his}, we have $\tbad(\bm{\rho})$ is connected in
the square graph of $G_H^S$.
% $(G_H^{S})^2 = (V_H^S, (E_H^{S})^2)$ where $(E_H^{S})^2 \defeq \{(u,v) \mid u, v \in V_H^S,  1\le \text{dist}_{G_H^S}(u, v) \le 2 \}$.
% square graph of $G^S_H$.  
So, if we have $|\tbad(\bm{\rho})|\leq T/(2n)-2$ then it holds that $t\geq -T+1+2n$ for all $t\in v, v\in \tbad(\bm{\rho})$.
    Furthermore, by \Cref{definition:bad-vertices-his}, $t\geq -T+1+2n$ for all $t\in v,v\in \tbad(\bm{\rho})$ implies that $t\geq -T-1+n$ for all $t\in v,v\in \cbad(\bm{\rho})$.
    % as each $v\in \tbad$ must share timestamps with some $v\in \cbad$ in $G^S_H$
    To see this, 
    first observe that each $v\in \cbad(\*\rho)$ must share timestamps with some $v\in \tbad(\*\rho)$ in $G^S_H$, according to the construction in \Cref{definition:2-tree-construction}. Then we apply \eqref{eq:timestamps-max-min}.

Next, we claim that, with $t\geq -T+1+n$ for all $t\in v,v\in \cbad(\bm{\rho})$, we can deduce the result of update at each timestamp $t\in \BadTS{\*\rho} \defeq \bigcup\limits_{s\in \+C^{\-{bad}}(\*\rho)}s$. Then the lemma directly follows from this claim, as the updates alone can completely determine the event $A$.

We then prove the claim. By \Cref{Alg:complex-GD-decomposed-his}, given some timestamp $t$, 
% let $v=v_{i(t)}$. Note 
$r_t = 0 \text{ or } \bot$ due to the $\*b$-decomposition scheme specified in \eqref{eq:base-measure-his}. If $r_t=0$, then  $\sigma_t(v_{i(t)})$ is updated to $0$ at time $t$; otherwise  $r_t=\bot$, and $\sigma_t(v_{i(t)})$ is updated to $0$ if and only if the following event occurs:
\begin{equation}\label{eq:his-bad-event}
    \exists e\in \+E \text{ s.t. }v_{i(t)}\in e:  \forall u\in e\setminus \{v_{i(t)}\}, \sigma_t(u)=1,
\end{equation} meaning that the value of all the vertices in $e\setminus \{v_{i(t)}\}$ got updated to $1$ at their last updates. 

% \color{red}For $r_t = 0$, it holds that $\sigma_t(v_{i(t)})=0$ otherwise $r_t = \bot$, \color{black} 
We argue that it suffices to check the event in \eqref{eq:his-bad-event} for the hyperedges $e$ such that $\ts(e,t)\in \+C^{\-{bad}}(\*\rho)$. We note that when $t\geq -T+1+n$ and $t\in \BadTS{\*\rho}$, an edge $e'$ such that $v_{i(t)}\in e'$ and
$\ts(e',t)\not\in \+C^{\-{bad}}(\*\rho)$  means there exists $t' \in \ts(e',t)$ such that $r_{t'} = 0$, which means $\sigma_t(v_{i(t')}) = \sigma_{t'}(v_{i(t')})=0$.  Therefore, this edge $e'$ does not satisfy \eqref{eq:his-bad-event} as one of its vertices is already assigned $0$.

Then, to check the event in \eqref{eq:his-bad-event}, we notice that we only need values of $\sigma_t(v_{i(t)})$ for $t\in \BadTS{\*\rho}$. Because for any edge $e'$ that involves values at $t\not\in \BadTS{\*\rho}$, we have $\ts(e',t)\not\in \+C^{\-{bad}}(\*\rho)$, and by the same argument as above, $e'$ does not satisfy \eqref{eq:his-bad-event}.

Therefore, we can deduce the result of updates at each $t\in \BadTS{\*\rho}$ in chronological order.
%if and only if $r_{t'}=\bot$ for all $t'\in \ts(e,t)$.
%a necessary condition \todo{JL: check this}
%for this event to occur is that $r_{t'}=\bot$ for all $t'\in \ts(e,t)$. This concludes the proof of the claim and the lemma.
\end{proof}

\subsection{Bounding complex measures for small bad trees}
\label{sec:bound-small-badtrees}
Next, we establish the exponential decay of the total measure stated in \Cref{lemma:modulus-bound-subset} for each bad tree.
\begin{proof}[Proof of \Cref{lemma:modulus-bound-subset}]
   According to the deterministic process in \Cref{definition:2-tree-construction} for the construction of the $2$-tree $\+T^{\-{bad}}=\=T(\+C^{\-{bad}})$, all $\+C\in \=T^{-1}(\+T^\-{bad})$ must only contain vertices in $G^S_H$ that are within distance $1$ of $\+T^\-{bad}$, and to determine $\cbad$ it is sufficient to check for all vertices in $G_H^S$ that are within distance $2$ of $\+T^\-{bad}$ whether they are in $V^{\-{bad}}$ or not. Therefore, we can bound the left-hand side in \eqref{eq:2-tree-modulus-bound-subset} by translating to the following condition on $\bm{r}$:
   \begin{enumerate}
       \item for each $s\in \+T$, we have $r_t=\bot$ for all $t\in s$;\label{item:modulus-bound-1}
       \item $r_t$s on all other vertices in $V_H^S$ that are within distance $2$ of $\+T$ satisfy certain restrictions so that $\cbad\in \=T^{-1}(\+T)$. \label{item:modulus-bound-2}
   \end{enumerate}
Here, the \Cref{item:modulus-bound-1} contributes a factor of 
\[
\prod\limits_{t\in s:s\in \+T}\abs{\frac{\lambda_{v_{i(t)}}}{1+\lambda_{v_{i(t)}}}}\leq N^{|\+T|},
\]by the definition of $2$-trees that all vertices in $\+T$ contain disjoint timestamps. We can use a simple triangle inequality to bound \Cref{item:modulus-bound-2}. 
% Let $(G_H^{S})^2 = (V_H^S, (E_H^{S})^2)$ be the square graph of $G_H^S$, where $(E_H^{S})^2 = \{(u,v) \mid  u,v\in V_H^S, 1\le \text{dist}_{G_H^S}(u, v) \le 2 \}$.
By \Cref{lemma:witness-graph-degree-bound}, the maximum degree of $\tp{G_H^S}^2$ is $4 \Delta^2 k^4$, where $\tp{G_H^S}^2$ is the square graph of $G_H^S$. For vertices in $\+T$, the number of their neighboring vertices in $\tp{G_H^S}^2$ can be upper bounded by $4 \Delta^2 k^4$. Therefore, there are at most $4 \Delta^2 k^5$ timestamps. And each timestamp $t$ contributes at most $\abs{\frac{1}{1 + \lambda_{v_{i(t)}}}} + \abs{\frac{\lambda_{v_{i(t)}}}{1 + \lambda_{v_{i(t)}}}} \le M$ by a triangle inequality. So \Cref{item:modulus-bound-2} contributes a factor of at most $M^{4\Delta^2k^5 |\+T|}$.
To summarize, we have
\begin{align*}
\abs{\mu^{\-{GD}}_{T,\sigma,\bm{b}}\tp{\+C^{\-{bad}}\in \=T^{-1}(\+T) \land \*r_{\ts(S,0)} = \bot^{\ts(S, 0)}}}\leq  N^{|\+T|}\cdot M^{4\Delta^2k^5 |\+T|}\leq \alpha^{|\+T|}. 
&\qedhere
\end{align*}
%finishing the proof of the lemma.
\end{proof}

\subsection{Decay of large bad trees}\label{sec:decay-large-bad-tree}
 Finally, we prove \Cref{lemma:diminishing-bad-tree}, which asserts that large bad trees have diminishing contributions. 
Our approach follows a similar framework to the one used for bounding the contributions of small bad trees (\Cref{lemma:characterization,lemma:modulus-bound-subset}). A key technical distinction here is that when the bad tree is large, the outcome may no longer be determined solely by the bad component. To address this, we show that in such cases, the information from the initial configuration, the bad component, and the $r_t$ values for the first $n$ timestamps are sufficient to serve as a certificate of the outcome.

Fix $T\ge n$. Let $I = [-T+1, -T+n]$. Let $\*r_I = (r_t)_{t = -T+1}^{-T+n}\in \{0,\bot\}^I$ denote the first $n$ timestamps of~$\*r$.  We first show a counterpart of \Cref{lemma:characterization}-(\ref{item:characterization-2}) for large bad trees.

\begin{lemma}\label{lemma:outcome-correspondence}
Fix any $T\geq n$ and any event $A\subseteq \{0,1\}^V$. 
Consider the witness graph $G_H^S = \tp{V_H^S, E_H^S}$ constructed using the set of variables $S=\vbl(A)$.
Let $\bm{\rho}=(\rho_t)_{t=-T+1}^{0}\in\{0,\bot\}^T$. 
%be a sequence where each $\rho_t\in \{0,1,\bot\}$. 

Then, for any initial configuration $\sigma\in \supp(\mu)$, the occurrence of $A$ at time $0$ is determined by $\cbad(\bm{\rho})$, $\*\rho_{I}$  and $\*\rho_{\ts(S,0)}$. %thus\color{black},
%\begin{equation}\label{eq:4-11-2-tree-modulus-bound-subset}
% Fix
%\mu_{T, \sigma, \*b}^\-{GD} \tuple{\sigma_0\in A \mid \+C^\-{bad}(\*r) = \cbad(\*\rho) \land \*r_I = \*\rho_I \color{red}\land \*r_{\ts(S,0)}=\*\rho_{\ts(S,0)}\color{black} } \in \{0, 1\}.
%\end{equation}
%\color{red}when well-defined, i.e., when $\mu^{\-{GD}}_{T,\sigma,\bm{b}}\tp{ \+C^\-{bad}(\*r) = \cbad(\*\rho) \land \*r_I = \*\rho_I \color{red}\land \*r_{\ts(S,0)}=\*\rho_{\ts(S,0)}}\neq 0$.\color{black}

\end{lemma}
\begin{remark}
    As in \Cref{lemma:characterization}-(\ref{item:characterization-2}), \Cref{lemma:outcome-correspondence} also implies the following ``zero-one law'' for the measure of an event $A$:
    \begin{equation}\label{eq:zero-one-law-large-trees}
    \mu^{\-{GD}}_{T, \sigma, \*b}\left(\sigma_0\in A \mid \cbad(\*r) = \cbad(\*\rho) \land \*r_{I} = \*\rho_I \land \*r_{\ts(S, 0)} = \*\rho_{\ts(S, 0)}\right) \in \{0, 1\},
    \end{equation}
    provided the conditional measure is well-defined, i.e., the following event has non-zero measure:
    $$\cbad(\*r) = \cbad(\*\rho) \land \*r_{I} = \*\rho_{I} \land \*r_{\ts(S, 0)} = \*\rho_{\ts(S,0)}.$$ 
\end{remark}
\begin{proof}[Proof of \Cref{lemma:outcome-correspondence}]
Let $\BadTS{\*\rho} \defeq \bigcup\limits_{s \in \cbad(\*\rho)} s$.
    Conditioning on $\+C^\-{bad}
    (\*r) = \cbad(\*\rho) \land \*r_I = \*\rho_I$, we know that for all $t\in \BadTS{\*\rho}$ that $t \le -T + n$, it holds that $r_t = \*\rho_{I}(t)$.

    If $t \ge -T + n + 1$ for all $t\in \BadTS{\*\rho}$, then we can determine $\sigma_0$ by $\cbad(\*\rho)$ following the proof of \Cref{lemma:characterization}-(\ref{item:characterization-2}). Otherwise, we claim that, given $\*\rho_I$, we can still deduce the result of the update at each timestamp $t\in \BadTS{\*\rho}$. Then, the lemma directly follows from this claim, as the updates alone can completely determine $\sigma_0$.

We then prove the claim. By \Cref{Alg:complex-GD-decomposed-his}, given some timestamp $t$, 
% let $v=v_{i(t)}$. Note 
$r_t = 0 \text{ or } \bot$ due to the $\*b$-decomposition scheme specified in \eqref{eq:base-measure-his}. If $r_t=0$, then  $\sigma_t(v_{i(t)})$ is updated to $0$ at time $t$; otherwise  $r_t=\bot$, and $\sigma_t(v_{i(t)})$ is updated to $0$ if and only if the following event occurs:
\begin{equation}\label{eq:4-11-his-bad-event}
    \exists e\in \+E \text{ s.t. }v_{i(t)}\in e:  \forall u\in e\setminus \{v_{i(t)}\}, \sigma_t(u)=1,
\end{equation} meaning that the value of all the vertices in $e\setminus \{v_{i(t)}\}$ got updated to $1$ at their last updates. 

Fix a specific $t \in \BadTS{\*\rho}$, 
\begin{itemize}
    \item For $t\geq -T+1+n$, if $r_t = 0$ then $\sigma_t(v_{i(t)})=0$, otherwise $r_t = \bot$, 
 we argue that it suffices to check the event in \eqref{eq:4-11-his-bad-event} for the hyperedges $e$ such that $\ts(e,t)\in \+C^{\-{bad}}(\*\rho)$. We note that when $t\geq -T+1+n$ and $t\in \BadTS{\*\rho}$, an edge $e'$ such that $v_{i(t)}\in e'$ and
$\ts(e',t)\not\in \+C^{\-{bad}}(\*\rho)$  means there exists $t' \in \ts(e',t)$ such that $r_{t'} = 0$, which means $\sigma_t(v_{i(t')}) = \sigma_{t'}(v_{i(t')})=0$.  Therefore, this edge $e'$ does not satisfy \eqref{eq:4-11-his-bad-event} as one of its vertices is already assigned $0$,
\item For $t\le-T+n$,  all events in \eqref{eq:4-11-his-bad-event} can be determined from $\*\rho_I$ and the initial configuration $\sigma$. 
\end{itemize}
% So we only need the initial configuration and the values of timestamps in $\BadTS{\*\rho}$, $I$ to deduce the result of updates.
Therefore, we can deduce the result of the updates at each $t\in \BadTS{\*\rho}$ in chronological order.
%
% Because for any edge $e'$ that involves values at $t\not\in \BadTS{\*\rho}$, we have $\ts(e',t)\not\in \+C^{\-{bad}}(\*\rho)$, and by the same argument as above, $e'$ does not satisfy \eqref{eq:his-bad-event}.
\end{proof}

We also have the following lemma that serves as a counterpart of \Cref{lemma:modulus-bound-subset}.
\begin{lemma}\label{lemma:modulus-bound-component-rn}
Fix any $S\subseteq V$ and $T\geq n$. 
Consider the witness graph $G_H^S$, the bad component $\cbad=\cbad(\*r)$, and the bad tree $\+T^{\-{bad}}=\+T^{\-{bad}}(\bm{r})=\=T(\cbad)$ ,
where $\bm{r}=(r_t)_{t=-T+1}^{0}\in\{0,\bot\}^T$ is constructed as in \Cref{Alg:complex-GD-decomposed-his}, following the product measure in~\eqref{eq:base-measure-his}.

Then, for any $\sigma\in \supp(\mu)$, for any $2$-tree $\+T$ in $G^S_H$ containing $\ts(S,0)$, we have
\begin{equation}\label{eq:2-tree-modulus-bound-component-rn}
\sum\limits_{\+C \in \=T^{-1}(\+T)}
\sum\limits_{\substack{\*\rho_I \in \{0,\bot\}^I \\ \*\rho_{\ts(S,0)} \in \{0, \bot\}^{\ts(S,0)}}}
\abs{\mu^{\-{GD}}_{T,\sigma,\bm{b}} \tp{\+C^{\-{bad}} = \+C \land \*r_I = \*\rho_I \land \*r_{\ts(S,0)}=\*\rho_{\ts(S,0)} }}\leq \alpha^{|\+T|-1}\cdot M^{4\Delta^2k^4|S| + n}.
\end{equation}
\end{lemma}
\begin{proof}
According to the deterministic process in \Cref{definition:2-tree-construction} for the construction of the $2$-tree $\+T^{\-{bad}}=\=T(\+C^{\-{bad}})$, all $\+C\in \=T^{-1}(\+T^\-{bad})$ must only contain vertices in $G^S_H$  within distance $1$ of $\+T^\-{bad}$, and to determine $\cbad$ it is sufficient to check for all vertices in $G_H^S$ within distance $2$ of $\+T^\-{bad}$ whether they are in $V^{\-{bad}}$ or not. Thus, we  bound the left-hand side in \eqref{eq:2-tree-modulus-bound-component-rn} with the following condition on $\bm{r}$:
\begin{enumerate}
   \item for each $s\in \+T\setminus \ts(S,0)$, we have $r_t=\bot$ for all $t\in s$;\label{item:4-11-modulus-bound-1}
   \item $r_t=\bm{\rho}_I(t)$ for each $t\in I$; \label{item:4-11-modulus-bound-2}
   \item $r_t=\*\rho_{\ts(S,0)}(t)$ for each $t\in \ts(S,0)$; \label{item:4-11-modulus-bound-3}
   \item $r_t$'s on vertices within distance $2$ of $\+T$ satisfy certain restrictions so that $\cbad\in \=T^{-1}(\+T)$.  \label{item:4-11-modulus-bound-4}
\end{enumerate}

Here, the \Cref{item:4-11-modulus-bound-1} contributes a factor of 
\[
\prod\limits_{t\in s:s\in \+T\setminus \ts(S,0)}\abs{\frac{\lambda_{v_{i(t)}}}{1+\lambda_{v_{i(t)}}}}\leq N^{|\+T|-1},
\]by the definition of $2$-trees that all vertices in $\+T$ contain disjoint timestamps.
% Let $(G_H^{S})^2 = (V_H^S, (E_H^{S})^2)$ be the square graph of $G_H^S$, where $(E_H^{S})^2 = \{(u,v) \mid  u,v\in V_H^S, 1\le \text{dist}_{G_H^S}(u, v) \le 2 \}$.
By \Cref{lemma:witness-graph-degree-bound}, the degree of $\ts(S,0)$ in $\left(G_H^S\right)^2$ can be upper bounded by $4\Delta^2k^3|S|$, where $\left(G_H^S\right)^2$ is the square graph of $G_H^S$. For vertices other than $\ts(S,0)$, the number of their neighboring vertices in $\left(G_H^S\right)^2$, excluding those already adjacent to $\ts(S,0)$, can be upper bounded by $4\Delta^2k^4$. Therefore, there are at most $4\Delta^2k^4(|\+T|-1)+4\Delta^2k^3|S|$ vertices and at most $4\Delta^2k^5(|\+T|-1)+4\Delta^2k^4|S|$ timestamps included. And each timestamp $t$ contributes at most $\abs{\frac{1}{1 + \lambda_{v_{i(t)}}}} + \abs{\frac{\lambda_{v_{i(t)}}}{1 + \lambda_{v_{i(t)}}}} \le M$ by a simple triangle inequality. So \Cref{item:4-11-modulus-bound-2,item:4-11-modulus-bound-3,item:4-11-modulus-bound-4} together contribute a factor of at most 
\[M^{4\Delta^2k^5(|\+T|-1)+4\Delta^2k^4|S| + n}.
\]
Summarizing, we have
\begin{align*}
&\sum\limits_{\+C \in \=T^{-1}(\+T)}
\sum\limits_{\substack{\*\rho_I \in \{0,\bot\}^I \\ \*\rho_{\ts(S,0)} \in \{0, \bot\}^{\ts(S,0)}}}
\abs{\mu^{\-{GD}}_{T,\sigma,\bm{b}} \tp{\+C^{\-{bad}} = \+C \land \*r_I = \*\rho_I \land \*r_{\ts(S,0)}=\*\rho_{\ts(S,0)} }} \\
&\leq  N^{|\+T|-1}\cdot M^{4\Delta^2k^5(|\+T|-1)+4\Delta^2k^4|S| + n}\\
&\leq \alpha^{|\+T|-1}\cdot M^{4\Delta^2k^4|S|+n}. 
&\qedhere
\end{align*}
%finishing the proof of the lemma.
\end{proof}

% \begin{}

We are now ready to prove \Cref{lemma:diminishing-bad-tree}.
 \begin{proof}[Proof of \Cref{lemma:diminishing-bad-tree}]
For each $i\geq 1$, let $\+T_i$ denote the set of $2$-trees of size $i$ in $G^S_H$ containing $\ts(S,0)$. Recall the definition of the bad component $\cbad=\cbad(\*r)$ in \Cref{definition:bad-vertices-his} and $\=T$ as the construction in \Cref{definition:2-tree-construction} such that $\=T(\cbad)=\tbad$.  
 \begin{equation}\label{eq:bad-tree-growth-transform}
\begin{aligned}
 &  \lim\limits_{T\to \infty}\abs{\mu^{\-{GD}}_{T,\sigma,\bm{b}}\tp{\sigma_0=\tau\land |\tbad|>T/(2n)-2}}\\
 \le&\lim\limits_{T\to \infty} \sum\limits_{i>T/(2n)-2}\sum\limits_{\+T\in \+T_i}\sum\limits_{\+C\in \=T^{-1}(\+T)} 
 \sum\limits_{\substack{\*\rho_I \in \{0,\bot\}^I \\ \*\rho_{\ts(S,0)} \in \{0,\bot\}^{\ts(S,0)}}}
 \abs{\mu^{\-{GD}}_{T,\sigma,\bm{b}}(\sigma_0 = \tau \mid \+C^\-{bad} = \+C \land \*r_I = \*\rho_I \land \*r_{\ts(S,0)} = \*\rho_{\ts(S,0)} )}\\
 &\qquad \cdot \abs{\mu^{\-{GD}}_{T,\sigma,\bm{b}}(\+C^\-{bad} = \+C \land \*r_I = \*\rho_I \land \*r_{\ts(S,0)} = \*\rho_{\ts(S,0)} )}\\
 \le& \lim\limits_{T\to \infty} \sum\limits_{i>T/(2n)-2}\sum\limits_{\+T\in \+T_i}\sum\limits_{\+C\in \=T^{-1}(\+T)} 
 \sum\limits_{\substack{\*\rho_I \in \{0,\bot\}^I \\ \*\rho_{\ts(S,0)} \in \{0,\bot\}^{\ts(S,0)}}}
 \abs{\mu^{\-{GD}}_{T,\sigma,\bm{b}}(\+C^\-{bad} = \+C \land \*r_I = \*\rho_I \land \*r_{\ts(S,0)} = \*\rho_{\ts(S,0)} )},
 \end{aligned}
 \end{equation}
where the last inequality is due to \Cref{lemma:outcome-correspondence} and the zero-one law stated in \eqref{eq:zero-one-law-large-trees}. Specifically, the zero-one law ensures that $\mu^{\-{GD}}_{T,\sigma,\bm{b}}(\sigma_0 = \tau \mid \+C^\-{bad} = \+C \land \*r_I = \*\rho_I \land \*r_{\ts(S,0)} = \*\rho_{\ts(S,0)})$  evaluates to either $0$ or $1$, provided the conditional measure is well-defined. 
%By \Cref{lemma:modulus-bound-component-rn}, the above can be bounded as
Recall the definition of $\alpha$ and $M$ in \Cref{definition:parameters-his}. Then letting $D_1=4\Delta^2k^4$ and $D_2=4\Delta^2k^3|S|$,  \begin{align*}
&\lim\limits_{T\to \infty}\abs{\mu^{\-{GD}}_{T,\sigma,\bm{b}}\tp{\sigma_0=\tau\land |\tbad|>T/(2n)-2}}\\
  (\text{by \eqref{eq:bad-tree-growth-transform} and \Cref{lemma:modulus-bound-component-rn}}) \quad  \leq 
  &  \lim\limits_{T\to \infty}\sum\limits_{i> T/(2n)-2} \sum\limits_{\+T\in \+T_i} \alpha^{i-1}\cdot M^{4\Delta^2k^4n + n}\\
 (\text{by \Cref{lemma:bad-tree-num-bound}}) \quad  \leq & \lim\limits_{T\to \infty}\sum\limits_{i> T/(2n)-2}  (\mathrm{e}(D_2 + i - 2))^{D_2 - 1} \cdot (\mathrm{e}D_1)^{i-1}\cdot \alpha^{i-1}\cdot M^{4\Delta^2k^4n + n} \\
(\text{by \eqref{eq:bound-alpha}})\quad \le& \lim\limits_{T\to \infty}\sum\limits_{i> T/(2n)-2}  (\mathrm{e}(D_2 + i - 2))^{D_2 - 1}\cdot M^{4\Delta^2k^4n + n} \frac{1}{2^{i-1}}\\
(\text{by taking limit})\quad =&0.
 &\qedhere
\end{align*}
 \end{proof}

\section{Reduction from Fisher zeros to Lee-Yang zeros}

In this section, we prove \Cref{theorem:zero-freeness-his-fisher}, which addresses the Fisher zeros of hypergraph independence polynomials.
%\begin{proof}[Proof of \Cref{theorem:zero-freeness-his-fisher}]
To achieve this,  we introduce the following reduction from the independence polynomial parameterized by edge strengths to that parameterized by vertex weights.
The resulting parameters are beyond the zero-free region that we have proved in \Cref{theorem:zero-freeness-his-lee-yang}, but we are able to verify  the key condition \eqref{eq:bound-alpha} as required in the inductive step \Cref{lemma:unified}, and still carry out the induction.

\begin{definition}[reduction from Fisher zeros to Lee-Yang zeros]\label{definition:reduction}
    Given a hypergraph $H=(V,\+E)$ and $\*\beta=(\beta_e)_{e\in \+E}\in \=C^{\+E}$ for each $e\in \+E$, we deterministically construct a hypergraph $H'=(V',\+E')$ and a complex vector $\*\lambda=(\lambda_v)_{v\in V'}\in \=C^{V'}$, denoted by $(H',\*\lambda)=\-{red}(H,\*\beta)$, as follows:
    \begin{itemize}
        \item $V'=V\cup \{v_e\mid e\in \+E\text{ s.t. }\beta_e\neq 0\},\+E'= \{ e\mid e\in\+E \text{ s.t. } \beta_e=0 \} \cup \{e\cup v_e\mid e\in \+E\text{ s.t. }\beta_e\neq 0\}$;
        \item $\lambda_v=\begin{cases}1 & v\in V\\ \frac{1-\beta_e}{\beta_e}& v=v_e\text{ for some }e\in \+E\text{ s.t. }\beta_e\neq 0\end{cases}$ 
    \end{itemize}
\end{definition}

The following lemma establishes the desirable properties of the reduction in \Cref{definition:reduction}.
\begin{lemma}[property of the reduction]\label{lemma:reduction-property}
Let $H=(V,\+E)$ be a hypergraph with a maximum degree $\Delta\geq 1$ and a maximum edge size $k$ and let $(H',\*\lambda)=\-{red}(H,\*\beta)$. Then, 
\begin{enumerate}
    \item $H'$ has a maximum degree $\Delta$ and a maximum edge size at most $k+1$;\label{item:reduction-1}
    \item the partition functions $Z_{H'}^{\-{ly}}(\*\lambda)$ and $=Z_{H}^\-{fs}(\*\beta)$ satisfy:
    \[Z_{H'}^{\-{ly}}(\*\lambda)=Z_{H}^\-{fs}(\*\beta)\cdot \prod\limits_{e\in \+E:\beta_e\neq 0}\beta_e^{-1}.\]\label{item:reduction-2}
\end{enumerate}
\end{lemma}
\begin{proof}
    \Cref{item:reduction-1} follows straightforwardly from the reduction in \Cref{definition:reduction}. 
    We then prove \Cref{item:reduction-2}. Let 
    $$V^*=V'\setminus V=\{v_e\mid e\in \+E\text{ s.t. }\beta_e\neq 0\}\quad\text{ and }\quad\+E^*=\+E'\setminus \+E=\{e\cup v_e\mid e\in \+E\text{ s.t. }\beta_e\neq 0\}.$$ 
    For any logical expression $P$, we define the Iverson bracket $[P]=1$ if $P$ is true, otherwise $[P]=0$.
    Recall $\+I(H') \subseteq \{0,1\}^{V'}$ denotes the set of all independent sets in $H'$. 
    
    The  partition function $Z_{H}^\-{fs}(\*\beta)$ on $H$ with edge strengths $\*\beta$ can be expressed as:
    \begin{align*}
    Z_{H}^\-{fs}(\*\beta)=
    % \sum\limits_{S\subseteq V}\prod\limits_{e:e\in S}\beta_e=
    \sum\limits_{\sigma\in \{0,1\}^V}\prod\limits_{e\in \+E} (\beta_e [\sigma(e)=1^e] + [\sigma(e)\neq 1^e]).
    % \begin{cases}\beta_e & \sigma(e)=1^{e}\\ 1 & \sigma(e)\neq 1^{e} \end{cases};
      \end{align*}
    Meanwhile, the partition function $Z_{H'}^{\-{ly}}(\*\lambda)$ on $H'$ with vertex weights $\*\lambda$ can be expressed as:
          \begin{align}
         Z_{H'}^{\-{ly}}(\*\lambda)=&\sum\limits_{\sigma\in \+I(H')}\prod\limits_{v:\sigma(v)=1 }\lambda_v\notag\\
         =&\sum\limits_{\sigma\in \{0,1\}^{V'}}\prod\limits_{v\in V':\sigma(v)=1}\lambda_v\prod\limits_{e\in \+E'}[\sigma(e)\neq 1^e]\notag\\
         % \begin{cases}0 & \sigma(e)=1^{e}\\ 1 & \sigma(e)\neq 1^{e}. \end{cases}\\
        % (\text{by \Cref{definition:reduction}})\quad 
        =&\sum\limits_{\sigma\in \{0,1\}^{V}}\prod\limits_{v\in V:\sigma(v)=1}\lambda_v\prod\limits_{e\in \+E:\beta_e\neq 0}\left(1/\beta_e[\sigma(e)\neq 1^e] + [\sigma(e)=1^e]\right)
        % \begin{cases}1 & \sigma(e)=1^{e}\\ 1/\beta_e & \sigma(e)\neq 1^{e} 
         % \end{cases}
         \prod\limits_{e\in \+E:\beta_e= 0}[\sigma(e)\neq 1^e]\label{eq:proof-ly-fs-reduction}\\
         % \begin{cases}0 & \sigma(e)=1^{e}\\ 1 & \sigma(e)\neq 1^{e} 
         % \end{cases}.
         =&\prod\limits_{e\in \+E:\beta_e\neq 0}\beta_e^{-1} \cdot \sum\limits_{\sigma\in \{0,1\}^{V}} \prod\limits_{e\in \+E:\beta_e\neq 0}\left([\sigma(e)\neq 1^e] + \beta_e[\sigma(e)=1^e]\right) \notag
        % \begin{cases}1 & \sigma(e)=1^{e}\\ 1/\beta_e & \sigma(e)\neq 1^{e} 
         % \end{cases}
         \prod\limits_{e\in \+E:\beta_e= 0}[\sigma(e)\neq 1^e]\label{eq:proof-ly-fs-reduction}\\
         =&Z_{H}^\-{fs}(\*\beta)\cdot \prod\limits_{e\in \+E:\beta_e\neq 0}\beta_e^{-1},\notag
         \end{align}
        where the equality in \eqref{eq:proof-ly-fs-reduction}  follows from \Cref{definition:reduction}.
%        
%        Then, \Cref{lemma:reduction-property} is proved by comparing $Z_{H'}^{\-{ly}}(\*\lambda)$ with $Z_{H}^\-{fs}(\*\beta)$.
\end{proof}

Let $(H',\*\lambda)=\-{red}(H,\*\beta)$. Note that according to \Cref{lemma:reduction-property}-(\ref{item:reduction-2}), we have $Z_{H'}^{\-{ly}}(\*\lambda)\neq 0$ if and only if $Z_{H}^\-{fs}(\*\beta)\neq 0$. We then prove that $Z_{H'}^{\-{ly}}(\*\lambda)\neq 0$. Note that the condition in \Cref{theorem:zero-freeness-his-fisher} together with \Cref{definition:reduction} implies $\lambda_v\neq -1$ for each $v\in V'$. It suffices to verify that the condition in \Cref{theorem:zero-freeness-his-fisher} implies \eqref{eq:bound-alpha} for the instance $(H',\*\lambda)$, then \Cref{theorem:zero-freeness-his-fisher} directly follows from \Cref{lemma:unified}-(\ref{item:unified-2}) and the same edge-wise self-reduction as in the proof of \Cref{theorem:zero-freeness-his-lee-yang}.

Now we bound the quantity $N=N(H',\*\lambda)$ in \Cref{definition:parameters-his}. 
By definition of $N$ and the reduction constructed in \Cref{definition:reduction}, it holds that
\begin{align*}
N &\le \min\left( \left(\frac{1}{2}\right)^k, \left(\frac{1}{2}\right)^k \max_{\substack{\beta \neq 0\\ \beta \in \+D_{\varepsilon}}} \abs{\frac{(1-\beta)/\beta}{1 + (1-\beta)/\beta}} \right) 
\le 2^{-k} \max\limits_{\substack{\beta \neq 0 \\ \beta \in \+D_\varepsilon}} \abs{1 - \beta}.
%\le (1 + 2\varepsilon) 2^{-k}.
%\le \left(2 \sqrt{2} \mathrm{e} \Delta (k+1)\right)^{-2},
\end{align*}
%where the first inequality directly follows from the definition of $N$ and each case represents one type edge in $\+E^*$. The last inequality is due to the conditions in \Cref{theorem:zero-freeness-his-fisher}.
%
% $0<\epsilon < \frac{1}{16(k+1)^5\Delta^2}$
By the condition of \Cref{theorem:zero-freeness-his-fisher}, we have 
$$N\le 2^{-k} \max\limits_{\substack{\beta \neq 0 \\ \beta \in \+D_\varepsilon}} \abs{1 - \beta}\le (1 + 2\varepsilon) 2^{-k}\le \left(2 \sqrt{2} \mathrm{e} \Delta (k+1)\right)^{-2}.$$

Next, we bound the quantity $M=M(H',\*\lambda)$ in \Cref{definition:parameters-his}. 
By the definition of $M$ and the reduction constructed in \Cref{definition:reduction}, it holds that
\[
M \le \max\left( \frac{1}{2}+\frac{1}{2}, \max_{\substack{\beta \neq 0\\ \beta \in \+D_{\varepsilon}}}\left( \abs{\frac{1}{1+(1-\beta)/\beta}}+ \abs{\frac{(1-\beta)/\beta}{1 + (1-\beta)/\beta}}\right) \right) \le \max\limits_{\substack{\beta \neq 0 \\ \beta \in \+D_\varepsilon}}(\abs{\beta} + \abs{1 - \beta}).
\]
%where the first inequality directly follows from the definition of $M$. 
By the condition of \Cref{theorem:zero-freeness-his-fisher}, we have 
\[
M\le \max\limits_{\substack{\beta \neq 0 \\ \beta \in \+D_\varepsilon}}(\abs{\beta} + \abs{1 - \beta}) \le 1 + 4 \epsilon.
\]
%The last inequality is due to the conditions in \Cref{theorem:zero-freeness-his-fisher}.

Recall the quantity $\alpha=\alpha(H',\*\lambda)$ in \Cref{definition:parameters-his}.  
We have
\begin{align*}
   &8\mathrm{e}\Delta^2(k+1)^4\cdot\alpha(H',\*\lambda)\\
   \le & 8 \mathrm{e} \Delta^2 (k+1)^4 \left(2 \sqrt{2} \mathrm{e} \Delta (k+1)^2\right)^{-2} (1 + 4 \varepsilon)^{4 \Delta^2 (k+1)^5} \\
    \le &\mathrm{e}^{-1} \exp\left( 16 \varepsilon \Delta^2 (k+1)^5 \right) \\
    <&1.
\end{align*}
This establishes the condition in \eqref{eq:bound-alpha} for the instance $(H',\*\lambda)$.
As discussed above, it proves \Cref{theorem:zero-freeness-his-fisher}.

\bibliographystyle{alpha}
\bibliography{refs.bib}

\newcommand{\etalchar}[1]{$^{#1}$}
\begin{thebibliography}{FGW{\etalchar{+}}23}

\bibitem[AASV21]{Alimohammadi2021FractionallyLA}
Yeganeh Alimohammadi, Nima Anari, Kirankumar Shiragur, and Thuy~Duong Vuong.
\newblock Fractionally log-concave and sector-stable polynomials: counting
  planar matchings and more.
\newblock In {\em STOC}, page 433–446. ACM, 2021.

\bibitem[AJ22]{anand2021perfect}
Konrad Anand and Mark Jerrum.
\newblock Perfect sampling in infinite spin systems via strong spatial mixing.
\newblock {\em SIAM J. Comput.}, 51(4):1280--1295, 2022.

\bibitem[AJK{\etalchar{+}}24]{anari2024universality}
Nima Anari, Vishesh Jain, Frederic Koehler, Huy~Tuan Pham, and Thuy-Duong
  Vuong.
\newblock Universality of spectral independence with applications to fast
  mixing in spin glasses.
\newblock In {\em SODA}, pages 5029--5056. SIAM, 2024.

\bibitem[ALGV18]{anari2018log}
Nima Anari, Kuikui Liu, Shayan~Oveis Gharan, and Cynthia Vinzant.
\newblock Log-concave polynomials iii: Mason's ultra-log-concavity conjecture
  for independent sets of matroids.
\newblock {\em arXiv preprint arXiv:1811.01600}, 2018.

\bibitem[Alo91]{Alon91}
Noga Alon.
\newblock A parallel algorithmic version of the local lemma.
\newblock In {\em FOCS}, pages 586--593. IEEE, 1991.

\bibitem[ALO20]{anari2020spectral}
Nima Anari, Kuikui Liu, and Shayan {Oveis Gharan}.
\newblock Spectral independence in high-dimensional expanders and applications
  to the hardcore model.
\newblock In {\em FOCS}, pages 1319--1330. IEEE, 2020.

\bibitem[Asa70]{Asano1970LeeYangTA}
Taro Asano.
\newblock Lee-yang theorem and the griffiths inequality for the anisotropic
  heisenberg ferromagnet.
\newblock {\em Phys. Rev. Lett.}, 24:1409--1411, 1970.

\bibitem[Bar16]{barvinok2016combinatorics}
Alexander Barvinok.
\newblock {\em Combinatorics and complexity of partition functions}, volume~30
  of {\em Algorithms and Combinatorics}.
\newblock Springer, Cham, 2016.

\bibitem[Bar19]{barvinok2019personal}
Alexander Barvinok.
\newblock Personal communication during a simons institute program: Geometry of
  polynomials.
\newblock 2019.

\bibitem[BB23]{bencs2023optimal}
Ferenc Bencs and Pjotr Buys.
\newblock Optimal zero-free regions for the independence polynomial of bounded
  degree hypergraphs.
\newblock {\em arXiv preprint arXiv:2306.00122}, 2023.

\bibitem[BBL09]{Borcea2007NegativeDA}
Julius Borcea, Petter Br{\"a}nd{\'e}n, and Thomas Liggett.
\newblock Negative dependence and the geometry of polynomials.
\newblock {\em J. Am. Math. Soc.}, 22(2):521--567, 2009.

\bibitem[BCKL13]{borgs2013left}
Christian Borgs, Jennifer Chayes, Jeff Kahn, and L\'{a}szl\'{o} Lov\'{a}sz.
\newblock Left and right convergence of graphs with bounded degree.
\newblock {\em Random Struct. Algorithms}, 42(1):1--28, 2013.

\bibitem[Ber16]{berkowitz2016quantitative}
Ross Berkowitz.
\newblock A quantitative local limit theorem for triangles in random graphs.
\newblock {\em arXiv preprint arXiv:1610.01281}, 2016.

\bibitem[BGG{\etalchar{+}}19]{NPH}
Ivona Bez{\'{a}}kov{\'{a}}, Andreas Galanis, Leslie~Ann Goldberg, Heng Guo, and
  Daniel Stefankovic.
\newblock Approximation via correlation decay when strong spatial mixing fails.
\newblock {\em SIAM J. Comput.}, 48:279--349, 2019.

\bibitem[BH20]{branden2020lorentzian}
Petter Br{\"a}nd{\'e}n and June Huh.
\newblock Lorentzian polynomials.
\newblock {\em Ann. Math.}, 192(3):821--891, 2020.

\bibitem[CLV20]{chen2020rapid}
Zongchen Chen, Kuikui Liu, and Eric Vigoda.
\newblock Rapid mixing of glauber dynamics up to uniqueness via contraction.
\newblock In {\em FOCS}, pages 1307--1318. IEEE, 2020.

\bibitem[CLV21]{Chen2021SpectralIV}
Zongchen Chen, Kuikui Liu, and Eric Vigoda.
\newblock Spectral independence via stability and applications to holant-type
  problems.
\newblock In {\em FOCS}, pages 149--160. IEEE, 2021.

\bibitem[DP23]{davies2023approximately}
Ewan Davies and Will Perkins.
\newblock Approximately counting independent sets of a given size in
  bounded-degree graphs.
\newblock {\em SIAM Journal on Computing}, 52(2):618--640, 2023.

\bibitem[DS85]{dobrushin1985completely}
Roland~L Dobrushin and Senya~B Shlosman.
\newblock Completely analytical gibbs fields.
\newblock In {\em Statistical Physics and Dynamical Systems: Rigorous Results},
  pages 371--403. Springer, 1985.

\bibitem[DS87]{dobrushin1987completely}
Roland~L Dobrushin and Senya~B Shlosman.
\newblock Completely analytical interactions: constructive description.
\newblock {\em J. Stat. Phys.}, 46:983--1014, 1987.

\bibitem[DT77]{dobrushin1977central}
RL~Dobrushin and Brunello Tirozzi.
\newblock The central limit theorem and the problem of equivalence of
  ensembles.
\newblock {\em Communications in Mathematical Physics}, 54:173--192, 1977.

\bibitem[EL75]{LocalLemma}
Paul Erd\H{o}s and L\'aszl\'o Lov\'asz.
\newblock Problems and results on 3-chromatic hypergraphs and some related
  questions.
\newblock {\em Infinite and finite sets, volume 10 of Colloquia Mathematica
  Societatis J\'anos Bolyai}, pages 609--628, 1975.

\bibitem[FGW{\etalchar{+}}23]{feng2023towards}
Weiming Feng, Heng Guo, Chunyang Wang, Jiaheng Wang, and Yitong Yin.
\newblock Towards derandomising {M}arkov chain {M}onte {C}arlo.
\newblock In {\em FOCS}, pages 1963--1990. IEEE, 2023.

\bibitem[FGYZ21]{feng2021rapid}
Weiming Feng, Heng Guo, Yitong Yin, and Chihao Zhang.
\newblock Rapid mixing from spectral independence beyond the {B}oolean
  {D}omain.
\newblock In {\em SODA}, pages 1558--1577. SIAM, 2021.

\bibitem[Fis65]{fisher1965nature}
Michael~E. Fisher.
\newblock {\em The nature of critical points}.
\newblock University of Colorado Press, 1965.

\bibitem[Gam23]{gamarnik2023correlation}
David Gamarnik.
\newblock Correlation decay and the absence of zeros property of partition
  functions.
\newblock {\em Random Struct. \& Algorithms}, 62(1):155--180, 2023.

\bibitem[GMP{\etalchar{+}}24]{galvin2024zeroes}
David Galvin, Gwen McKinley, Will Perkins, Michail Sarantis, and Prasad Tetali.
\newblock On the zeroes of hypergraph independence polynomials.
\newblock {\em Comb. Probab. Comput.}, 33(1):65–84, 2024.

\bibitem[GSS11]{gharan2011randomized}
Shayan~Oveis Gharan, Amin Saberi, and Mohit Singh.
\newblock A randomized rounding approach to the traveling salesman problem.
\newblock In {\em FOCS}, pages 550--559. IEEE, 2011.

\bibitem[Gur06]{gurvits2006van}
Leonid Gurvits.
\newblock The van der waerden conjecture for mixed discriminants.
\newblock {\em Adv. Math.}, 200(2):435--454, 2006.

\bibitem[HL72]{Heilmann1972TheoryOM}
Ole~J. Heilmann and Elliott~H. Lieb.
\newblock Theory of monomer-dimer systems.
\newblock {\em Commun. Math. Phys.}, 27:166, 1972.

\bibitem[HSS11]{haeupler2011new}
Bernhard Haeupler, Barna Saha, and Aravind Srinivasan.
\newblock New constructive aspects of the lov{\'{a}}sz local lemma.
\newblock {\em J. {ACM}}, 58(6):28:1--28:28, 2011.
\newblock (Conference version in \emph{FOCS}'10).

\bibitem[HSW21]{HSW21}
Kun He, Xiaoming Sun, and Kewen Wu.
\newblock Perfect sampling for (atomic) {L}ov\'{a}sz local lemma.
\newblock {\em arXiv}, abs/{2107.03932}, 2021.

\bibitem[HSZ19]{HSZ19}
Jonathan Hermon, Allan Sly, and Yumeng Zhang.
\newblock Rapid mixing of hypergraph independent sets.
\newblock {\em Random Struct. Algorithms}, 54(4):730--767, 2019.

\bibitem[Hur95]{hurwitz1895conditions}
Adolf Hurwitz.
\newblock On the conditions under which an equation has only roots with
  negative real parts.
\newblock {\em Mathematische Annalen, English translation in Selected papers on
  mathematical trends in control theory}, 46:273--284, 1895.

\bibitem[HWY22]{he2022sampling}
Kun He, Chunyang Wang, and Yitong Yin.
\newblock Sampling lovász local lemma for general constraint satisfaction
  solutions in near-linear time.
\newblock In {\em FOCS}, pages 147--158. {IEEE}, 2022.

\bibitem[JPSS22]{vishesh2022approximate}
Vishesh Jain, Will Perkins, Ashwin Sah, and Mehtaab Sawhney.
\newblock Approximate counting and sampling via local central limit theorems.
\newblock In {\em {STOC}}, pages 1473--1486. {ACM}, 2022.

\bibitem[JPV21]{Vishesh21towards}
Vishesh Jain, Huy~Tuan Pham, and Thuy~Duong Vuong.
\newblock Towards the sampling lov{\'{a}}sz local lemma.
\newblock In {\em FOCS}, pages 173--183. {IEEE}, 2021.

\bibitem[Kah00]{kahn2000normal}
Jeff Kahn.
\newblock A normal law for matchings.
\newblock {\em Combinatorica}, 20(3):339--392, 2000.

\bibitem[KKG21]{karlin2021slightly}
Anna~R Karlin, Nathan Klein, and Shayan~Oveis Gharan.
\newblock A (slightly) improved approximation algorithm for metric tsp.
\newblock In {\em STOC}, pages 32--45, 2021.

\bibitem[KS11]{KS11}
Kashyap Kolipaka and Mario Szegedy.
\newblock Moser and {T}ardos meet {L}ov\'{a}sz.
\newblock In {\em STOC}, pages 235--243. ACM, 2011.

\bibitem[KS18]{kyng2018matrix}
Rasmus Kyng and Zhao Song.
\newblock A matrix chernoff bound for strongly rayleigh distributions and
  spectral sparsifiers from a few random spanning trees.
\newblock In {\em FOCS}, pages 373--384. IEEE, 2018.

\bibitem[LL15]{liu2015fptas}
Jingcheng Liu and Pinyan Lu.
\newblock Fptas for counting monotone cnf.
\newblock In {\em SODA}, pages 1531--1548, 2015.

\bibitem[LPRS16]{lebowitz2016central}
Joel~L. Lebowitz, Boris Pittel, David Ruelle, and Eugene~R. Speer.
\newblock Central limit theorems, lee--yang zeros, and graph-counting
  polynomials.
\newblock {\em Journal of Combinatorial Theory, Series A}, 141:147--183, 2016.

\bibitem[LS81]{lieb1981general}
Elliott~H Lieb and Alan~D Sokal.
\newblock A general lee-yang theorem for one-component and multicomponent
  ferromagnets.
\newblock {\em Commun. Math. Phys.}, 80(2):153--179, 1981.

\bibitem[LS16]{lubetzky2016information}
Eyal Lubetzky and Allan Sly.
\newblock Information percolation and cutoff for the stochastic ising model.
\newblock {\em J. Am. Math. Soc.}, 29(3):729--774, 2016.

\bibitem[LSS17]{Liu2017TheIP}
Jingcheng Liu, Alistair Sinclair, and Piyush Srivastava.
\newblock The ising partition function: Zeros and deterministic approximation.
\newblock {\em J. Stat. Phys.}, 174:287 -- 315, 2017.

\bibitem[LSS19]{liu2019correlation}
Jingcheng Liu, Alistair Sinclair, and Piyush Srivastava.
\newblock Correlation decay and partition function zeros: Algorithms and phase
  transitions.
\newblock {\em SIAM J. Comput.}, pages 200--252, 2019.

\bibitem[LX21]{li2021complex}
Liang Li and Guangzeng Xie.
\newblock Complex contraction on trees without proof of correlation decay.
\newblock {\em arXiv preprint arXiv:2112.15347}, 2021.

\bibitem[LY52]{Lee1952StatisticalTO}
T.~D. Lee and Chen~Ning Yang.
\newblock Statistical theory of equations of state and phase transitions. ii.
  lattice gas and ising model.
\newblock {\em Phys. Rev.}, 87:410--419, 1952.

\bibitem[LYZ16]{lu2016fptas}
Pinyan Lu, Kuan Yang, and Chihao Zhang.
\newblock Fptas for hardcore and ising models on hypergraphs.
\newblock In {\em STACS}. Schloss-Dagstuhl-Leibniz Zentrum f{\"u}r Informatik,
  2016.

\bibitem[MS19]{michelen2019central}
Marcus Michelen and Julian Sahasrabudhe.
\newblock Central limit theorems and the geometry of polynomials.
\newblock {\em arXiv preprint arXiv:1908.09020}, 2019.

\bibitem[MSS15a]{marcus2015interlacing1}
Adam~W Marcus, Daniel~A Spielman, and Nikhil Srivastava.
\newblock Interlacing families i: Bipartite ramanujan graphs of all degrees.
\newblock {\em Ann. Math.}, 182(1):307--325, 2015.

\bibitem[MSS15b]{marcus2015interlacing2}
Adam~W Marcus, Daniel~A Spielman, and Nikhil Srivastava.
\newblock Interlacing families ii: Mixed characteristic polynomials and the
  kadison—singer problem.
\newblock {\em Ann. Math.}, pages 327--350, 2015.

\bibitem[MSS18]{marcus2018interlacing}
Adam~W Marcus, Daniel~A Spielman, and Nikhil Srivastava.
\newblock Interlacing families iv: Bipartite ramanujan graphs of all sizes.
\newblock {\em SIAM J. Comput.}, 47(6):2488--2509, 2018.

\bibitem[Pet75]{Petrov1975}
V.~V. Petrov.
\newblock {\em Sums of Independent Random Variables}.
\newblock De Gruyter, Berlin, Boston, 1975.

\bibitem[PR17]{patel2017deterministic}
Viresh Patel and Guus Regts.
\newblock Deterministic polynomial-time approximation algorithms for partition
  functions and graph polynomials.
\newblock {\em SIAM J. Comput.}, 46(6):1893--1919, 2017.

\bibitem[PR19]{peter2019conjecture}
Han Peters and Guus Regts.
\newblock On a conjecture of {S}okal concerning roots of the independence
  polynomial.
\newblock {\em Michigan Math. J.}, 68(1):33--55, 2019.

\bibitem[QWZ22]{qiu2022perfect}
Guoliang Qiu, Yanheng Wang, and Chihao Zhang.
\newblock A perfect sampler for hypergraph independent sets.
\newblock In {\em ICALP}, pages 103:1--103:16. Schloss Dagstuhl -
  Leibniz-Zentrum f{\"{u}}r Informatik, 2022.

\bibitem[Reg23]{regts2023absence}
Guus Regts.
\newblock Absence of zeros implies strong spatial mixing.
\newblock {\em Probab. Theory Relat. Fields}, 186(1):621--641, 2023.

\bibitem[Rou77]{routh1877treatise}
Edward~John Routh.
\newblock {\em A treatise on the stability of a given state of motion,
  particularly steady motion: being the essay to which the Adams prize was
  adjudged in 1877, in the University of Cambridge}.
\newblock Macmillan and Company, 1877.

\bibitem[Rue71]{ruelle1971extension}
David Ruelle.
\newblock Extension of the lee-yang circle theorem.
\newblock {\em Phys. Rev. Lett.}, 26(6):303, 1971.

\bibitem[She85]{shearer85}
James~B. Shearer.
\newblock On a problem of {S}pencer.
\newblock {\em Combinatorica}, 5(3):241--245, 1985.

\bibitem[Sly10]{Sly2010computation}
Allan Sly.
\newblock Computational transition at the uniqueness threshold.
\newblock In {\em FOCS}, pages 287--296. {IEEE}, 2010.

\bibitem[SS03]{Scott2003TheRL}
Alex~D. Scott and Alan~D. Sokal.
\newblock The repulsive lattice gas, the independent-set polynomial, and the
  lov{\'a}sz local lemma.
\newblock {\em J. Stat. Phys.}, 118:1151--1261, 2003.

\bibitem[SS19]{Shao2019ContractionAU}
Shuai Shao and Yuxin Sun.
\newblock Contraction: A unified perspective of correlation decay and
  zero-freeness of 2-spin systems.
\newblock {\em J. Stat. Phys.}, 185, 2019.

\bibitem[SZ92]{stroock1992logarithmic}
Daniel~W Stroock and Boguslaw Zegarlinski.
\newblock The logarithmic sobolev inequality for discrete spin systems on a
  lattice.
\newblock {\em Commun. Math. Phys.}, 149:175--193, 1992.

\bibitem[Wag09]{wagner2009weighted}
David~G Wagner.
\newblock Weighted enumeration of spanning subgraphs with degree constraints.
\newblock {\em J. Comb. Theory, Ser B}, 99(2):347--357, 2009.

\bibitem[Wei06]{weitz06counting}
Dror Weitz.
\newblock Counting independent sets up to the tree threshold.
\newblock In {\em {STOC}}, pages 140--149. {ACM}, 2006.

\end{thebibliography}

\pagebreak
\appendix
\section{A central limit theorem for hypergraph independent sets}
\label{sec:clt}
In this section, we establish a central limit theorem (CLT) for hypergraph independent sets, using our new zero-free region for hypergraph independence polynomials (namely, \Cref{theorem:zero-freeness-his-lee-yang}). Specifically, we will prove \Cref{theorem:clt}.

The tool we use is the following relation between zero-free regions and central limit theorems, which is also the starting point in~\cite{vishesh2022approximate,davies2023approximately} for establishing CLTs. 

\begin{lemma}[{\cite[Theorem 1.2]{michelen2019central}}]\label{lemma:zero-free-to-clt}
    Let $X \in \{0,\ldots, n\}$ be a random variable with mean $\bar{\mu}$, standard deviation $\sigma$ and probability generating function $f$ and set $X^*=(X - \bar{\mu})\sigma^{-1}$. For $\delta \in (0, 1)$ such that $\abs{1 - \zeta} \ge \delta$ for all roots $\zeta$ of $f$,
    \[
    \sup_{t\in \=R} \abs{\=P[X^* \le t] - \=P[\+Z \le t]} \le O\left( \frac{\log n}{\delta \sigma} \right),
    \]
    where $\+Z \sim N(0, 1)$ is a standard Gaussian random variable.
\end{lemma}

\subsection{Central limit theorem}
In this subsection, we prove a multivariate version of the first part of \Cref{theorem:clt}.

\begin{theorem}[first part of \Cref{theorem:clt}]\label{theorem:clt-part1}
Fix $k\ge 2$, $\Delta \ge 3$. Let $H = (V, \+E)$ be a $k$-uniform hypergraph with maximum degree $\Delta$. Let $n=\abs{V}.$ Fix any $\delta>0$, $\varepsilon\in\left(0,\frac{1}{9k^5\Delta^2}\right)$. Let $\lambda_{c, \varepsilon}$ be defined as in \Cref{theorem:zero-freeness-his-lee-yang}. 
For any $\*\lambda \in (\delta, \lambda_{c, \varepsilon}]^V$, let $I\sim \mu_{H,\*\lambda}$, and define $X = \abs{I}$, $\bar{\mu} = \=E[X]$ and $\sigma^2 = \-{Var}[X]$. Then we have
\[
\sup_{t \in \=R} \abs{\=P[(X - \bar{\mu})/\sigma \le t] - \=P[\+Z \le t]} = O_{k,\Delta,\delta,\varepsilon}\left(\frac{\log n}{\sqrt{n}}\right),
\]
where $\+Z \sim N(0, 1)$ is a standard Gaussian random variable.
\end{theorem}

We need to lower bound the variance of $X=|I|$, the size of a random independent set following the Gibbs measure. The following bound is a generalization of ~\cite[Lemma 3.2]{vishesh2022approximate} to hypergraphs and with vertex-dependent external fields.

\begin{lemma}\label{lemma:bound-variance}
% Suppose the conditions in \Cref{theorem:zero-freeness-his-lee-yang}, let $\varepsilon, \+D_\varepsilon$ be defined as in $$
Let $\lambda_{\min}=\min_{v\in V}\lambda_v$ and $\lambda_{\max}=\max_{v\in V}\lambda_v$. Under the condition of \Cref{theorem:clt-part1},  we have
\[
 \Omega_{k,\Delta,\eps}(\lambda_{\min} n)\leq \-{Var}[X] \leq O_{k,\Delta,\eps}(\lambda_{\max} n).
\]
%Furthermore, let $\lambda_{\min} = \min_v \lambda_v$ and $\lambda_{\max} = \min_v \lambda_v$. It holds that
%\[\Omega_{\Delta, k, \varepsilon}\tp{\lambda_{\min}n} \le \-{Var}[X] \le O_{\Delta, k, \varepsilon}\tp{\lambda_{\max} n} .\]
%For $\lambda_{\max}> 1$, it holds that
%\[\Omega_{\Delta, k, \varepsilon}\tp{\lambda_{\min}n} \le \-{Var}[X] \le O_{\Delta, k, \varepsilon}\tp{n}.\]
\end{lemma}
\begin{proof}
We first show that $\-{Var}[X] = O_{k,\Delta,\eps}(\lambda_{\max}n)$. 
We consider the generating polynomial of $X$:
\[
f(x) = \sum\limits_{\sigma \in \+I(H)}\prod\limits_{v:\sigma(v)=1}(\lambda_v \cdot x).
\]
 We have
% \[
% f(1)=1, \quad f'(1)=\=E[X],\quad f''(1)= \=E[X^2]-\=E[X];
% \]
\[
\left.\ln(f(x))'\right|_{x = 1} = \frac{f'(1)}{f(1)}=\=E[X],\quad \left.\ln(f(x))''\right|_{x = 1}=\frac{f(1)f''(1)-(f'(1))^2}{(f(1))^2}=\=E[X^2] - \=E[X] - (\=E[X])^2,
\]
and therefore
\[\-{Var}[X] = \left.\ln(f(x))''\right|_{x = 1} + \left.\ln(f(x))'\right|_{x = 1}.\]

%For a $\sigma\in \+I(H)$, let $\sharp \sigma = \abs{\{v\mid v\in V, \sigma_v = 1\}}$ and $w(\sigma) = \prod\limits_{v:\sigma(v)=1}\lambda_v$. We have that
%\[
%\ln(f(x))' = \frac{1}{f(x)} \sum\limits_{\substack{\sigma\in\+I(H) \\ \sharp\sigma \ge 1}} \sharp\sigma~ x^{\sharp \sigma - 1} w(\sigma).
%\]
%So we have $\left.\ln(f(x))'\right|_{x = 1} = \=E[X]$. We can also verify that,
%\[
%\ln(f(x))'' = \frac{\sum\limits_{\substack{\sigma\in\+I(H) \\ \sharp\sigma \ge 2}}\sharp\sigma(\sharp \sigma - 1) x^{\sharp\sigma - 2}w(\sigma) }{f(x)} - \left(\frac{\sum\limits_{\substack{\sigma\in\+I(H) \\ \sharp\sigma \ge 1}} \sharp\sigma~ x^{\sharp \sigma - 1} w(\sigma) }{f(x)}\right)^2.
%\]
%We can verify that $\left.\ln(f(x))''\right|_{x = 1} = \=E[X^2] - \=E[X] - (\=E[X])^2$.
%So it holds that 
%\[\-{Var}[X] = \left.\ln(f(x))''\right|_{x = 1} + \left.\ln(f(x))'\right|_{x = 1}.\]

Note that the degree of $f$ is equal to $N$, the size of the maximum independent set in $H$, and that $f(0)=1$. By the fundamental theorem of algebra, we can write $f(x)$ as,
\[
f(x) = \prod\limits_{j = 1}^N (1 - r_j x),
\]
where $r_1, r_2, \ldots, r_{N}$ be the inverses of the complex roots of $f(x)$. Then we have that,
\begin{align*}
    \-{Var}[X] &= \left.(\ln(f(x)))''\right|_{x=1} + \left.\ln(f(x))'\right|_{x = 1} \\
    &= \sum\limits_{j = 1}^N \frac{r_j}{(1 - r_j)^2}.
\end{align*}
To upper bound $\-{Var}[X]$, it suffices to bound $\max\limits_{1 \le j \le N}\abs{\frac{r_j}{(1 - r_j)^2}}$. Note that $f(x)=Z_H^\-{ly}(\*\lambda \cdot x)$. Therefore, according to \Cref{theorem:zero-freeness-his-lee-yang}, $f(x)$ is zero-free on $D = \{x \mid x\in \=C, \abs{x - 1}\le \varepsilon/(\sqrt{2}\lambda_{\max})\}$. 
Let $D^c = \=C \backslash D$.
With the zero-free region for $f(x)$, we show that,
\[
\max\limits_{1 \le j \le N} \abs{\frac{r_j}{(1 - r_j)^2}} \le \max\limits_{x \in D^c}\frac{\abs{x}}{\abs{ x-1}^2} \le  \max\limits_{x \in D^c} \left( \frac{1}{\abs{x - 1}} + \frac{1}{\abs{x - 1}^2} \right) \le \sqrt{2}\lambda_{\max}\varepsilon^{-1} + 2\lambda_{\max}^2\varepsilon^{-2}.
\]

Note $\lambda_{\max} \le \lambda_{c, \varepsilon} = O_{k, \Delta, \varepsilon}(1)$. So it holds that 
\[\-{Var}[X] \le (\sqrt{2}\lambda_{\max}\varepsilon^{-1} + 2\lambda_{\max}^2\varepsilon^{-2})N = O_{k, \Delta, \varepsilon}(\lambda_{\max} n).\]
 Next, we show $\-{Var}[X] = \Omega_{k,\Delta,\eps}(\lambda_{\min} n)$. 
 Let $J\subseteq V$ be a maximum hypergraph matching, that is, a subset of vertices satisfying
 \begin{equation}\label{eq:J-condition}
 \forall e\in \+E,\quad |J\cap e|\leq 1,
 \end{equation}
 with maximum size. We write $|J|=M$.
 
 We additionally let $K = I \setminus J$. Note that $X = |K| + |I \cap J|$, and according to \eqref{eq:J-condition}, conditioned on $K$ the set $I \cap J$ is distributed according to the hypergraph independent set on 
\[
U=J \setminus \{v\in J\mid \exists e\in \+E\text{ s.t. } e\subseteq K\cup \{v\} \},
\]
which is the distribution each vertex $v\in U$ independently taking $1$ with probability $\frac{\lambda_v}{1 + \lambda_v}$ and taking $0$ with probability $\frac{1}{1 + \lambda_v}$. Therefore,
\[
\-{Var}[X\mid U]=\sum\limits_{v\in U}\frac{\lambda_v}{(1+\lambda_v)^2}.
\]

 By the law of total variance,
    \[
    \-{Var}[X] = \=E[\-{Var}[X \mid U]] + \-{Var}[\=E[X \mid U]] \ge \=E[\-{Var}[X \mid U]]=\sum\limits_{v\in J} \left(\frac{\lambda_v}{(1+\lambda_v)^2}\cdot \=P[v\in U]\right).
    \]

By the condition that $\lambda_v\in [\lambda_{\min},\lambda_{\max}]$, we further have
\[
\-{Var}[X]\geq \sum\limits_{v\in J} \left(\frac{\lambda_v}{(1+\lambda_v)^2}\cdot \=P[v\in U]\right)\geq \frac{1}{1+\lambda_{c, \varepsilon}}\cdot \frac{\lambda_{\min}}{1+\lambda_{\min}}\sum\limits_{v\in J}\=P[v\in U]=\Omega_{k,\Delta,\varepsilon}(\lambda_{\min})\cdot \=E[|U|].
\]

Note that vertex $v\in J$ is not in $U$ precisely when there exists some $e\in \+E$ such that $u\in e$ and $e\subseteq K\cup \{v\}$. Note further that under the distribution of random independent sets of $H$ with fugacity $\lambda$, the probability of each vertex $v$ being occupied is at most 
\[\frac{\lambda_v}{1+\lambda_v}\leq \frac{\lambda_{c,\varepsilon}}{1+\lambda_{c,\varepsilon}},\] under any conditioning on the value of other vertices. This means that 
\[
\mathbb{E}[\abs{U}] \geq |J|\left(1-\Delta\cdot \left(\frac{\lambda_{c,\varepsilon}}{1+\lambda_{c,\varepsilon}}\right)^{k-1}\right)=\left(1-\Delta\cdot \left(\frac{\lambda_{c,\varepsilon}}{1+\lambda_{c,\varepsilon}}\right)^{k-1}\right)M> \left(1-\frac{1}{8\mathrm{e}^2\Delta k^4}\right)M,
\] 
where the last inequality is from the condition in \Cref{theorem:clt-part1}. Hence,
\[
\-{Var}[X] > \Omega_{k, \Delta, \varepsilon}(\lambda_{\min})\cdot\left(1-\frac{1}{8\mathrm{e}^2\Delta k^4}\right)M=\Omega_{k,\Delta,\varepsilon}(\lambda_{\min})\cdot\left(1-\frac{1}{8\mathrm{e}^2\Delta k^4}\right)M.
\]

Now, it remains to notice that we can bound $M \geq \frac{n}{(k-1)\Delta+1}$ from the following greedy process of constructing a $J$ satisfying \eqref{eq:J-condition}: each time pick an arbitrary remaining vertex and remove all hyperedges containing it together with all vertices inside them.
% \begin{align*}
%     \max\limits_{1 \le j \le N}\abs{\frac{r_j}{(1 - \lambda r_j)^2}}&\le \max\limits_{1 \le j \le N}\abs{\frac{1}{(\lambda \sqrt{r_j} - \frac{1}{r_j})^2}}
% \end{align*}
\end{proof}

Now we can prove \Cref{theorem:clt-part1}.
\begin{proof}[Proof of \Cref{theorem:clt-part1}]
Let $f$ denote the generating polynomial of $X$, which is
\[
f(x)=\sum\limits_{\sigma\in \+I(H)}\prod\limits_{v:\sigma(v)=1}\left(\lambda_v\cdot x\right).
\]
Note that $f(x)=Z_H^\-{ly}(\*\lambda \cdot x)$. Let $g(x) = f(x)/f(1)$ be its probability generating function. Therefore, according to \Cref{theorem:zero-freeness-his-lee-yang}, $g(x)$ is zero-free on $D = \{x \mid x\in \=C, \abs{x - 1}\le \varepsilon/(\sqrt{2}\lambda_{c, \varepsilon})\}$. Also, note for any $\lambda_v$, it holds that $\lambda_v > \delta$. Hence, \Cref{theorem:clt-part1} follows directly from \Cref{lemma:zero-free-to-clt,lemma:bound-variance}. 
\end{proof}

\subsection{Local central limit theorem}
In this subsection, we prove a local central limit theorem for hypergraph independent polynomials, %with our improved zero-free region and the standard argument in \cite{vishesh2022approximate}.
%Specially, we prove the next theorem 
which is the second part of \Cref{theorem:clt}.
\begin{theorem}[second part of \Cref{theorem:clt}]
\label{theorem:clt-part2}
Fix $k\ge 2$ and $\Delta \ge 3$. Let $H = (V, \+E)$ be a $k$-uniform hypergraph with maximum degree $\Delta$. Let $n=|V|$.
Fix any $\delta>0$, $\varepsilon\in\left(0,\frac{1}{9k^5\Delta^2}\right)$. Let $\lambda_{c, \varepsilon}$ be defined as in \Cref{theorem:zero-freeness-his-lee-yang}. 
For any $\lambda \in (0, \lambda_{c, \varepsilon}]$, let $I\sim \mu_{H,\*\lambda}$, and define $X = \abs{I}$, $\bar{\mu} = \=E[X]$ and $\sigma^2 = \-{Var}[X]$. 
Let $\+N(x) = \mathrm{e}^{-x^2/2}/\sqrt{2\pi}$ denote the density of the standard normal distribution, we have
\[
\sup\limits_{t \in \=Z}\abs{\=P[X = t] - \sigma^{-1}\+N((t - \bar{\mu})/\sigma)} = O_{k,\Delta,\varepsilon}\tp{\min\tp{\frac{(\log n)^{5/2}}{\sigma^2},\frac{1}{\sigma^2}+\frac{\sigma^{2k}(\log n)^2}{n^{k-1}}}}.
\]
\end{theorem}

In \cite{vishesh2022approximate}, they use the following standard lemma from \cite{berkowitz2016quantitative}, which quantifies a local central limit theorem via approximations to characteristic functions. 

\begin{lemma}[{\cite[Lemma 3]{berkowitz2016quantitative}}]
\label{lemma:point-wise-diff-upperbound}
    Let $X$ be a random variable supported on the lattice $\+L = \alpha + \beta \=Z$ and let $\+N(x) = \mathrm{e}^{-x^2/2}/\sqrt{2\pi}$ denote the density of the standard normal distribution. Then 
    \[
    \sup\limits_{x \in \+L} \abs{\beta \+N(x) - \=P[X = x]} \le \beta \int_{-\pi/\beta}^{\pi/\beta} \abs{\=E\inbr{\mathrm{e}^{itX}} - \=E\inbr{\mathrm{e}^{it \+Z}}}dt + \mathrm{e}^{-\pi^2/(2\beta^2)},
    \]
    where $\+Z \sim N(0, 1)$ is a standard Gaussian random variable.
\end{lemma}

Then they prove that the high Fourier phases of the characteristic function, $\abs{\=E\inbr{\mathrm{e}^{itX}}}$ with large $t$'s, are negligible via a combinatorial argument for the independent set polynomial in \cite{dobrushin1977central}. And for low Fourier phases, $\abs{\=E\inbr{\mathrm{e}^{itX}}}$ with small $t$'s, they use the central limit theorem to bound $
\abs{\=E\inbr{\mathrm{e}^{itX}} - \=E\inbr{\mathrm{e}^{it\+Z}}}$.

We follow the high-level idea in \cite{vishesh2022approximate}. For high Fourier phases, we  show that they are negligible in \Cref{lemma:characterized-func-bound}, and the main distinction is an analogous combinatorial argument for hypergraph independent sets (\Cref{lemma:scattered-vertices}). For low Fourier phases, we again use the central limit theorem (\Cref{lemma:characterized-func-diff-bound}).

% Then they bound low Fourier phases and high Fourier phases separately. For low Fourier phases, they use the central limit theorem. For high Fourier phases, they bound the characteristic function by bounding each term via a combinatorial argument. 
% We introduce the next lemma which uses the characteristic function to establish the local central limit theorem.

% {\cite[Theorem 1.2]{michelen2019central}}

% \yxtodo}
Before bounding the characteristic functions, we first bound the variance. From \Cref{lemma:bound-variance}, we directly obtain the following bound for the univariate hypergraph independent set polynomial. 
\begin{corollary}\label{corollary:univariate-variance-bound}
Let $H = (V, \+E)$ be a $k$-uniform hypergraph with maximum degree $\Delta$. Let $n = \abs{V}$. 
Fix any $\varepsilon\in\left(0,\frac{1}{9k^5\Delta^2}\right)$. Let $\lambda_{c, \varepsilon}$ be defined as in \Cref{theorem:zero-freeness-his-lee-yang}. 
For any $\lambda \in (0, \lambda_{c, \varepsilon}]$, let $I\sim \mu_{H,\lambda}$, and define $X = \abs{I}$. We have
\[
\-{Var}[X] = \Theta_{k,\Delta,\varepsilon}(\lambda n).
\]    
%If $\lambda > 1$, we have
%\[
%\Omega_{k,\Delta,\varepsilon}(\lambda^2 n) \le \-{Var}%[X] \le O_{k,\Delta,\varepsilon}(n).
%\] 
\end{corollary}

% \todo{wcy: This is changed, original is wrong}

The next lemma bounds the low Fourier phases by the central limit theorem.
\begin{lemma}\label{lemma:characterized-func-diff-bound}
Fix $k\ge 2$ and $\Delta \ge 3$. Let $H = (V, \+E)$ be a $k$-uniform hypergraph with maximum degree $\Delta$. Let $n = \abs{V}$. Fix any $\varepsilon \in \tp{0, \frac{1}{9k^5\Delta^2}}$, let $\lambda_{c, \varepsilon}$ be defined as in \Cref{theorem:zero-freeness-his-lee-yang}. For any $\lambda \in (0, \lambda_{c, \varepsilon}]$, let $I \sim \mu_{H, \lambda}$, and define $X = \abs{I}$, $\bar{\mu} = \=E[X]$ and $\sigma^2 = \-{Var}[X]$. Let $Y = (X - \bar{\mu})/\sigma$, $\+Z \sim N(0, 1)$ be a standard Gaussian random variable. For any $t\in \=R$, we have that
\[
\abs{\=E\inbr{\mathrm{e}^{itY}} - \=E\inbr{\mathrm{e}^{it\+Z}}} = O_{k,\Delta,\varepsilon}\tp{\frac{\abs{t}\tp{\log n}^{3/2} + \log n}{\sigma}}.
\]
\end{lemma}
\begin{proof}
We first recall the central limit theorem for $Y$. By \Cref{lemma:zero-free-to-clt}, we have that
\begin{align}\label{eq:clt}
\sup\limits_{t\in\=R}\abs{\=P[Y\le t] - \=P[\+Z\le t]} \le O_{k,\Delta,\varepsilon}\tp{\frac{\log n}{\sigma}}.
\end{align}

Next, we express the $\=E\inbr{\mathrm{e}^{itY}}$ into an integration.
Let $Y'$ be $Y$ convolved with a centered Gaussian of infinitesimally small variance so that $Y'$ has a density function with respect to the Lebesgue measure on $\=R$; it suffices to consider $Y'$ and then pass to the limit. We have that,
\begin{align*}
\mathbb{E}\inbr{\mathrm{e}^{itY'}} =& \int_{-\infty}^{\infty} \mathrm{e}^{itz} p_{Y'}(z) \, dz \\
=& \int_{|z| \leq \tau} \mathrm{e}^{itz} p_{Y'}(z) \, dz \pm \mathrm{e}^{i\theta'}\mathbb{P}[|Y'| \geq \tau] \\
=& \left[ \mathrm{e}^{itz} \left( \int_{-\tau}^{z} p_{Y'}(z') \, dz' \right) \right]_{z=-\tau}^{z=\tau} - \int_{-\tau}^{\tau} it \mathrm{e}^{itz} \left( \int_{-\tau}^{z} p_{Y'}(z') \, dz' \right) dz \pm \mathrm{e}^{i\theta'}\mathbb{P}[|Y'| \geq \tau] \\
=& \mathrm{e}^{it\tau} - \int_{-\tau}^{\tau} it \mathrm{e}^{itz} \mathbb{P}[Y' \in [-\tau, z]] \, dz \pm \mathrm{e}^{i\theta'}\mathbb{P}[|Y'| \geq \tau] - \mathrm{e}^{it\tau} \mathbb{P}[|Y'| \geq \tau] \\
(\text{by \eqref{eq:clt}})\quad=& \mathrm{e}^{it\tau} - \int_{-\tau}^{\tau} it \mathrm{e}^{itz} \mathbb{P}[Y' \in [-\tau, z]] \, dz + \mathrm{e}^{i\theta} \cdot O_{k,\Delta,\varepsilon} \left( \frac{\log n}{\sigma} + \mathrm{e}^{-\tau^2/4} \right),
\end{align*}
for some $\theta',\theta \in [0, 2\pi)$.
Then we apply the same calculation to $\+Z$ instead of $Y'$ and taking the difference, we find that
\begin{align*}
\abs{\mathbb{E}\inbr{\mathrm{e}^{itY'}} - \mathbb{E}\inbr{\mathrm{e}^{it\+Z}}} &\leq |t| \left| \int_{-\tau}^{\tau} |\mathbb{P}[Y' \in [-\tau, z]] - \mathbb{P}[\+Z \in [-\tau, z]]| \, dz \right|  + O_{k,\Delta,\varepsilon} \left( \frac{\log n}{\sigma} + \mathrm{e}^{-\tau^2/4} \right) \\
&\leq O_{k,\Delta,\varepsilon} \left( \frac{(|\tau t| + 1) \log n}{\sigma} + \mathrm{e}^{-\tau^2/4} \right).
\end{align*}
By \Cref{corollary:univariate-variance-bound}, setting $\tau = \sqrt{8 \log n}$ gives the desired conclusion.
\end{proof}

Then, we bound the high Fourier phases following a similar strategy as in~\cite{dobrushin1977central,vishesh2022approximate}. 

\begin{lemma}[hypergraph version of {\cite[Lemma 3.4]{vishesh2022approximate}}]\label{lemma:scattered-vertices}
Fix $k\ge2$, $\Delta \ge 3$. Let $H = (V, \+E)$ be a $k$-uniform hypergraph with maximum degree $\Delta$. Let $n=\abs{V}$. Then, there exists a subset $S \subseteq V$ of size $\Omega(n/(\Delta k)^3)$ such that all vertices in $S$ have pairwise distance at least $4$ with respect to the hypergraph distance. Moreover, there is an algorithm to find such a subset $S$ in time $O_{\Delta, k}(n)$.
\end{lemma}
\begin{proof}
    Let $v_1, v_2, \ldots, v_n$ donate an arbitrary enumeration of the vertices. Initialize $S = \emptyset$. Consider the greedy algorithm which, at each time step, adds the first variable vertex to the set $S$ and removes all vertices within distance $3$ of this vertex from consideration. The algorithm stops when there are no more available vertices. The algorithm runs in time $O_{\Delta, k}(n)$ and outputs a set $S$ such that any two vertices in $S$ have graph distance at least $4$. Moreover, since at each time, $O(\tp{\Delta k}^3)$ vertices are removed, it follows that $\abs{S} = \Omega(n/(\Delta k)^3)$.
\end{proof}
\begin{lemma}[hypergraph version of {\cite[Lemma 3.5]{vishesh2022approximate}}]
\label{lemma:characterized-func-bound}
    Fix $k\ge 2$, $\Delta \ge 3$. Let $H=(V, \+E)$ be a $k$-uniform hypergraph with maximum degree $\Delta$. Let $n = \abs{V}$. Fix any $\varepsilon\in\tp{0,\frac{1}{9k^5\Delta^2}}$. Let $\lambda_{c,\varepsilon}$ be defined as in \Cref{theorem:zero-freeness-his-lee-yang}. For any $\lambda \in (0, \lambda_{c, \varepsilon}]$, there exists a constant $c = c_{k,\Delta,\varepsilon} > 0$ satisfying the following. Let $I \sim \mu_{H, \lambda}$, $X = \abs{I}$ and define $\bar{\mu} = \=E[X]$, $\sigma^2 = \-{Var}[X]$. Let $Y = (X - \bar{\mu})/\sigma$. Then, for all $t \in [-\pi \sigma, \pi\sigma]$, we have
    \[
    \abs{\=E\inbr{\mathrm{e}^{-itY}}} \le \exp\tp{-c \lambda n t^2/\sigma^2}.
    \]
\end{lemma}
\begin{proof}
    It suffices to show that for all $t\in\=R$, $\abs{t}\le\pi$,
    \[
    \abs{\=E\inbr{\mathrm{e}^{-itX}}}\le\exp\tp{-c\lambda n t^2}.
    \]
    Let $S$ be a $4$-separated set of vertices of $H$ of size $s = \Omega(n/(\Delta k)^3)$ from \Cref{lemma:scattered-vertices}. Let $T$ be the set of vertices that are at distance at least $2$ from $S$ in $H$ and let $H[T]$ denote the graph on $H$ induced by $T$. Let $\zeta$ denote the distribution on $H[T]$ induced by the Gibbs distribution $\mu_{H, \lambda}$. We sample $I$ by first sampling $J \sim \zeta$ and then sampling from the conditional distribution (induced by the Gibbs distribution $\mu_{H, \lambda}$ and $J$) on $H[v\cup N(v)]$ for each $v \in S$. The key observation is that these conditional distributions are mutually independent. In particular, given $J$, we can write,
    \[
    X = \abs{J} + X_1 + X_2 + \dots + X_s,
    \]
    where each $X_j$ is an independent random variable with support in ${0, 1, \ldots, k\Delta}$.
    % ., a probability mass at $0$ of $\Omega_{\Delta, C}(1)$, and probability mass of $\Omega_{\Delta, C}(1)$ at 1. Note that the implicate 
    We claim that for all $\abs{t}\le \pi$ and all $j \in [s]$, for any realization of $J$,
    \[
    \abs{\=E\inbr{\mathrm{e}^{-itX_j}}} \le 1 - c'\lambda t^2,
    \]
    for some absolute $c' = c'_{k,\Delta,\varepsilon} > 0$. For any realization of $J$, letting $X_j'$ denote an independent copy of $X_j$, we have
    \begin{align*}
        \abs{\=E\inbr{\mathrm{e}^{-itX_j}}}^2 =& \=E\inbr{\mathrm{e}^{it(X_j - X_j')}}\\
        =&\=P[X_j=X_j'] + \sum\limits_{k = 1}^{k\Delta}\tp{\=P[X_j - X_j' = k] + \=P[X_j'-X_j = k]}\cos(kt) \\
        \le&\=P[X_j=X_j'] + \sum\limits_{k=2}^{k\Delta}\tp{\=P[X_j - X_j' = k] + \=P[X_j'-X_j = k]} + 2\=P[X_j-X_j'=1]\cos(t)\\
        =& 1 - 2\=P[X_j-X_j'=1](1 - \cos(t))\\
        \le& 1 - \frac{1}{4} \=P[X_j-X_j'=1]t^2 \\
        \le& 1 - \frac{1}{4}\=P[X_j=1]\=P[X_j'=0]t^2.
    \end{align*}
    Note that for any vertex $v$ in $H$, under arbitrary configurations of its neighbors, there is at least probability $\frac{1}{1+\lambda}$ such that the configuration of $v$ is $0$. So $\=P[X_j'=0]\ge \frac{1}{1+\lambda}$. For $\=P[X_j = 1]$, let $v_j$ be the vertex in $S$ and contributes to $X_j$. $\=P[X_j = 1]$ is lower bounded by the probability that only the configuration of $v_j$ is $1$ and for other vertices in $N(v_j)$, their configurations are $0$. So $\=P[X_j = 1] \ge \frac{\lambda}{1+\lambda} \tp{\frac{1}{1+\lambda}}^{\Delta k}$. So we have that,
    \[
    \abs{\=E\inbr{\mathrm{e}^{-itX_j}}}^2 \le 1 - c'\lambda t^2.
    \]
    % as claimed.

    Finally, we have that for any $t\in[-\pi, \pi]$,
    \[
    \=E[\mathrm{e}^{-itX}] \le \max\limits_J \abs{\=E[\mathrm{e}^{-itX}\mid J]}
    =\max\limits_J\prod\limits_{j=1}^s \abs{\=E[\mathrm{e}^{-itX_j}]} 
    \le (1 - c'\lambda t^2)^{s/2}
    \le \exp(-c n\lambda t^2),
    \]
    for an appropriate  $c = c_{k,\Delta,\varepsilon} > 0$ and the result follows.
\end{proof}
Now we are ready to prove \Cref{theorem:clt-part2}.
\begin{proof}[Proof of \Cref{theorem:clt-part2}]
    We first show the first part of the inequality. For $\sigma\ge2$, applying \Cref{lemma:point-wise-diff-upperbound} to $Y = (X - \bar{\mu})/\sigma \in \alpha + \beta \=Z$, where $\alpha = -\bar{\mu}/\sigma$ and $\beta = 1/\sigma$. We have that
    % and using \Cref{lemma:quantitative-bound-variance}, \Cref{lemma:characterized-func-diff-bound} and \Cref{lemma:characterized-func-bound} we see that for $\sigma\ge2$,
    \begin{align*}
        \sup\limits_{t\in\+L}\abs{\beta \+N(t) - \=P[Y = t]}
        \le \frac{1}{\sigma}\int_{-\pi\sigma}^{\pi\sigma} \abs{\=E\inbr{\mathrm{e}^{itY}}-\=E\inbr{\mathrm{e}^{it\+{Z}}}}dt + \mathrm{e}^{-\pi^2\sigma^2/2}.
    \end{align*}
    With \Cref{lemma:characterized-func-bound,lemma:characterized-func-diff-bound}, we can bound these two characteristic functions, then we have that,
    \begin{align*}
    \sup\limits_{t\in\+L}\abs{\beta N(t) - \=P[Y = t]}
        &\lesssim_{k,\Delta,\varepsilon} \frac{1}{\sigma}\int_{-\pi\sigma}^{\pi\sigma} \min\tp{\frac{\abs{t}\tp{\log n}^{3/2} + \log n}{\sigma}, \mathrm{e}^{-c \lambda n t^2/\sigma^2} + \mathrm{e}^{-t^2/2}}dt + \mathrm{e}^{-\pi^2\sigma^2/2} \\
        &\lesssim_{k,\Delta,\varepsilon} \frac{1}{\sigma}\int_{-c''\sqrt{\log \sigma}}^{c''\sqrt{\log \sigma}} \frac{\abs{t}\tp{\log n}^{3/2} + \log n}{\sigma} dt + \frac{1}{\sigma^2} \\
        &\lesssim_{k,\Delta,\varepsilon} \frac{\tp{\log n}^{5/2}}{\sigma^2}.
    \end{align*}
    % \textcolor{red}{the second term is not finished, we must need the second one?}

    For the second term, we may assume that $1 \le \sigma \le \log n$. Let $\lambda' = \lambda/(1 + \lambda)$ and observe that the Gibbs distribution $\mu_{H, \lambda}$ is identical to the product distribution $\mathrm{Ber}(\lambda')^{\otimes V}$ conditioned on the configuration being an independent set. Here $\mathrm{Ber}(\lambda')$ is the random variable which is $1$ (or occupied) with probability $\lambda'$ and $0$ (or unoccupied) otherwise. A trivial union bound argument shows that a random sample from $\mathrm{Ber}(\lambda')^{\otimes V}$ is an independent set with probability at least $1 - \lambda'^{k}\Delta n = 1 - O_{k,\Delta,\eps}(\lambda^k n) = 1 - O_{k,\Delta,\eps}(\sigma^{2k}/n^{k-1})$, where we used \Cref{corollary:univariate-variance-bound}. Therefore, the probability of any configuration under the hypergraph independent set model is within a factor of $1 \pm O_{k,\Delta,\eps}(\sigma^{2k}/n^{k-1})$ of the probability of the same configuration under $\-{Ber}(\lambda')^{\otimes V}$.

    Let $X'$ denote the random variable counting the number of $1$'s in a random sample from $\-{Ber}(\lambda')^{\otimes V}$, and let $\mu'$ and $\sigma'$ denote the mean and standard deviation of $X'$. Then by the classical DeMoivre-Laplacian central limit theorem~\cite{Petrov1975}, we get that for any integer $t$,
    \begin{equation}\label{eq:clt-product}
    \abs{\=P[X' = t] - \frac{1}{\sigma'} \+N\tp{\frac{t - \mu'}{\sigma'}}} = O\tp{\frac{1}{\sigma'^2}}.
    \end{equation}
    Let $Y'$ be the sample drawn from $\-{Ber}(\lambda')^{\otimes V}$ such that $X'=|Y'|$. From the comparison between the hypergraph independent set model and $\-{Ber}(\lambda')^{\otimes V}$ mentioned above we have
    \begin{equation}\label{eq:probability-bound}
    \begin{aligned}
         \=P[X = t]= & \frac{\=P[X'=t\land Y'\in \+I(H)]}{\=P[Y'\in \+I(H)]}\\
         =& (1\pm O_{k,\Delta,\eps}(\sigma^{2k}/n^{k-1}))\=P[X'=t\land Y'\in \+I(H)]\\
         =& (1\pm O_{k,\Delta,\eps}(\sigma^{2k}/n^{k-1}))(
         \=P[X'=t]-O_{k,\Delta,\eps}(\sigma^{2k}/n^{k-1}))\\
         =&\=P[X'=t]\pm O_{k,\Delta,\eps}(\sigma^{2k}/n^{k-1}).
         \end{aligned}
    \end{equation}
Note that by the Chernoff bound for the Binomial distribution we have
\begin{equation}\label{eq:mu-bound-1}
    \begin{aligned}
 \mu-\mu' =& (1\pm O(1/n^{k}))\sum_{t=1}^{(1+O_k(\log n))\mu'}t\cdot \left(\=P[X=t]-\=P[X'=t]\right)]).
 \end{aligned}
    \end{equation}

    Note that by a similar argument to \eqref{eq:probability-bound} we have 
    \begin{equation}\label{eq:mu-bound-2}
    \forall S\subseteq \=N,\quad \=P[X\in S]-\=P[X'\in S]=\pm O_{k,\Delta,\eps}(\sigma^{2k}/n^{k-1}).
    \end{equation}
    Therefore we have by combining \eqref{eq:mu-bound-1}
 and \eqref{eq:mu-bound-2}, and that $\mu\geq 1$:    
 \begin{equation}\label{eq:mu-bound-3}
       \mu=(1\pm O_{k,\Delta,\eps}(\sigma^{2k}\log n/n^{k-1}))\mu'.
 \end{equation}
      Note that we can bound the second moment similarly such that
    \[
    \=E[X^2]-\=E[(X')^2]=\pm O_{k,\Delta,\eps}(\sigma^{2k}\log^2n/n^{k-1}),
    \]
    and therefore by $\sigma\geq 1$ we have
    \begin{equation}\label{eq:sigma2-bound}
    \sigma^2=(1\pm O_{k,\Delta,\eps}(\sigma^{2k}\log^2n/n^{k-1}))(\sigma')^2.
     \end{equation}
     Substituting \eqref{eq:mu-bound-3} and \eqref{eq:sigma2-bound} into \eqref{eq:clt-product} yields the desired result.
\end{proof}
\section{Algorithmic implications of local central limit theorem}\label{section:lclt}
% In this section, we prove the local central limit theorem for hypergraph independent polynomials (namely, \Cref{theorem:lclt}) with our improved zero-free region and standard arguments in \cite{vishesh2022approximate}. 
In this section, we show the algorithmic implication from the local central limit theorem (\Cref{theorem:clt-part2}). Specifically, we prove \Cref{theorem:fptas-exact-k}.

For a hypergraph $H = (V, \+E)$ and an external field $\lambda$ with $n = \abs{V}$, we define $\alpha_H(\lambda)\defeq \frac{1}{n} \=E_{I\sim \mu_{H, \lambda}} [\abs{I}]$ to be the \emph{occupancy fraction} where $\mu_{H, \lambda}$ is the Gibbs distribution. 
% The next theorem gives an FPTAS to approximate $i_t(H)$.
% We also define $\alpha_H(\lambda)$ as the \emph{occupancy fraction} of a random hypergraph independent set drawn from the Gibbs distribution $\mu_{H, \lambda}$. 

We follow the high-level ideas in \cite{vishesh2022approximate}. 
We first describe their proof. For a graph $G=(V, E)$ with $n = \abs{V}$ and an integer $t$, they first find the appropriate $\lambda^*$ such that $\abs{n \alpha_{G}(\lambda^*) - t} \le 1/2$. To see this, they first show that $\alpha_G(\lambda)$ is non-decreasing with $\lambda$, and the derivative of $\alpha_G(\lambda)$ is bounded by the variance. They also show the variance is bounded, then they can use a grid search to find the appropriate $\lambda^*$.
To complete this step, they give an FPTAS to approximate $\alpha_G(\lambda)$ by the cluster expansion and the interpolation method to check whether $\abs{n\alpha_G(\lambda)-t}<1/2$. %They also give an FPRAS to find $\lambda^*$ based on the Glauber dynamics and the Chernoff bound.

Let $I$ be a random independent set in $G$ from the Gibbs distribution $\mu_{G, \lambda}$. Let $X = \abs{I}$. From the local central limit theorem, they show that $\=P[X = t] = \Omega\tuple{\frac{1}{\sqrt{\-{Var}[X]}}}$ where $\-{Var}[X] = \Theta(\lambda^* n)$. So an FPRAS follows directly from the rapid mixing Glauber dynamics and the rejection sampling. Now, we consider the FPTAS. Let $\bar{\mu}=\=E[X], \sigma^2=\-{Var}[X]$, $y = (t-\bar{\mu})/\sigma$ and $Y = (X-\bar{\mu})/\sigma$. Let $i_{t}(G)$ be the number of hypergraph independent sets in $G$ of size $t$. It holds that
\begin{align}\label{equation:exact-t-prob}
 \=P[X=t]=\=P[Y=y]=\frac{i_t(G)(\lambda^*)^t}{Z_G(\lambda^*)}.   
\end{align}
To approximate $i_t(G)$, it suffices to approximate $\=P[Y=y]$ and $Z_G(\lambda^*)$. \cite{vishesh2022approximate} use the polynomial interpolation to approximate $Z_G(\lambda^*)$. For $\=P[Y=y]$, they apply the Fourier inversion formula, so it holds that
\[
\=P[Y=y] = \frac{1}{2\pi\sigma}\int_{-\pi\sigma}^{\pi\sigma} \=E\inbr{\mathrm{e}^{itY}} \mathrm{e}^{-ity} dt.
\]
Then they show that for high Fourier phases of the characteristic function, $\=E\inbr{\mathrm{e}^{itY}}$ with large $t$'s, are negligible through a lemma analogous to \Cref{lemma:characterized-func-bound}. For low Fourier phases, they use summation to approximate the integration and a polynomial interpolation. %to approximate $\=E\inbr{\mathrm{e}^{itY}}$ with small $t$'s.

% In order to give an FPTAS for counting exact $t$-size hypergraph independent sets satisfying $t\le n\alpha_H(1)$, we need to go through the proofs of  Lemma 4.3 and Theorem 1.1 in \cite{vishesh2022approximate}.

% In their proof, they first use a lemma (similar to \Cref{lemma:proper-t-finding}) to find the proper external field $\lambda'$ such that the expectation is near $t$.
% In their proof of Lemma 4.3, they show that they can compute the $Z(\lambda' \mathrm{e}^{it/\sigma})$ for small $t$'s by the interpolation method so that they can approximate the low Fourier phases. 
% In their proof of Theorem 1.1 and 1.2, they first use the Fourier inversion formula for lattice to present $\=P[X = t]$, then use lemmas that are similar to \Cref{lemma:characterized-func-bound} to control the high Fourier phases, then approximate the low Fourier phases by the aforementioned Lemma 4.3.

% \subsection{Deterministic algorithm}
Now, we show a stronger form of \Cref{theorem:fptas-exact-k}. We also give a lower bound of $\alpha_H(\lambda)$ ~(\Cref{corollary:occupancy-fraction-lower-bound}). \Cref{theorem:fptas-exact-k} follows directly from \Cref{theorem:stronger-fptas-exact-k,corollary:occupancy-fraction-lower-bound}.

\begin{theorem}[stronger form of \Cref{theorem:fptas-exact-k}]
\label{theorem:stronger-fptas-exact-k}
Fix $k\ge2$, $\Delta\ge 3$. Let $H=(V, \+E)$ be a $k$-uniform hypergraph with maximum degree $\Delta$ and $n = \abs{V}$. Let $\varepsilon$, $\lambda_{c, \varepsilon}$ be defined as in \Cref{theorem:zero-freeness-his-lee-yang}.
There exists a deterministic algorithm that, on input $H$, an integer $1\le t\le n\alpha_H(\lambda_{c, \varepsilon})$, and an error parameter $\eta \in (0, 1)$, outputs an $\eta$-relative approximation to $i_t(H)$ in time $(n/\eta)^{O_{k, \Delta, \varepsilon}(1)}$.
\end{theorem}

To establish a lower bound for the occupancy fraction $\alpha_H(\lambda)$, we note that by linearity of expectation, it suffices to use the marginal lower bounds under Lov\'{a}sz local lemma conditions~\cite{LocalLemma,haeupler2011new}.

\begin{lemma}[\text{\cite[Theorem 2.1]{haeupler2011new}, restated}]\label{HSS}
Let $H=(V,\+E)$ be a $k$-uniform hypergraph with maximum degree $\Delta$ and let $\lambda>0$ such that
\begin{equation}\label{eq:lll-condition}
\mathrm{e}\left(\frac{\lambda}{1+\lambda}\right)^{k}\cdot k\Delta<1,
\end{equation}
Let $I\sim \mu_{H,\lambda}$ be a random independent set in $H$ with external field $\lambda$. Then, for any $v\in V$,
\[
\=P[v\notin I]\leq \frac{1}{1+\lambda}\cdot \left(1-\mathrm{e}\left(\frac{\lambda}{1+\lambda}\right)^{k}\right)^{-k\Delta}.
\]
\end{lemma}

The following corollary is immediate by noting that the condition in \Cref{theorem:zero-freeness-his-lee-yang} satisfies \eqref{eq:lll-condition}. %and by recalling that $\alpha_H(\lambda)$ is the average occupancy fraction.
\begin{corollary}\label{corollary:occupancy-fraction-lower-bound}
 Fix $k\ge 2$, $\Delta\ge 3$. Let $H=(V,\+E)$ be a $k$-uniform hypergraph with maximum degree $\Delta$. Fix any $\varepsilon \in \tp{0, \frac{1}{9k^5\Delta^2}}$. Let $\lambda_{c, \varepsilon}$ be defined as in \Cref{theorem:zero-freeness-his-lee-yang}, then for any $\lambda\in [0,\lambda_{c,\varepsilon}]$,
 \[
 \alpha_H(\lambda)\geq 1-\frac{1}{1+\lambda}\cdot \left(1+\frac{1}{4\mathrm{e}\Delta k^3}\right).
 \]
\end{corollary}

% To apply their algorithm, we need 
% \begin{enumerate}
%     \item for a given $t$, find the proper $\lambda'$, \label{item:proper-lambda}
%     \item check that we can approximate the low Fourier phases by interpolation,\label{item:low-fourier-phases-interpolation}
%     \item control the high Fourier phases.\label{item:high-fourier-phases}
% \end{enumerate}

% For \Cref{item:high-fourier-phases}, we have \Cref{lemma:characterized-func-bound}. 

% For \Cref{item:proper-lambda}, we give the following lemma.
% For a given integer $t$, the next lemma is used to find a proper external field so that $n\alpha_H(\lambda')\approx t$.

Before giving the proof of \Cref{theorem:stronger-fptas-exact-k}, we first collect some useful lemmas following the high-level idea in \cite{vishesh2022approximate}. Given an integer $t$, we first describe how to find the appropriate $\lambda^*$ such that $\abs{\alpha_H(\lambda^*)-t}\le 1/2$ (\Cref{lemma:proper-t-finding}). We provide an algorithm to approximate cumulants which is used in \Cref{lemma:proper-t-finding}. The $s$-th cumulant of a random variable $Y$ is defined in terms of the coefficients of the cumulant generating function $K_Y(z) = \log \=E\inbr{\mathrm{e}^{z Y}}$ (when this expectation exists in a neighborhood of $0$). In particular, the $s$-th cumulant is 
\[
\kappa_s(Y) = K_Y^{(s)}(0).
\]
We show that we can provide an additive approximation to cumulants.

\begin{lemma}\label{lemma:deterministic-cumulant}
    Fix $k\ge 2$, $\Delta\ge 3$. Let $H=(V,\+E)$ be a $k$-uniform hypergraph with maximum degree $\Delta$. Let $n = \abs{V}$. Fix any $\varepsilon \in \tp{0, \frac{1}{9k^5\Delta^2}}$. Let $\lambda_{c, \varepsilon}$ be defined as in \Cref{theorem:zero-freeness-his-lee-yang}. For any $\lambda \in [n^{-2}, \lambda_{c, \varepsilon}]$, let $I \sim \mu_{H, \lambda}$, $X = \abs{I}$ and $\kappa_s(X)$ be the $s$-th cumulant of $X$. Let $\eta \in (0, 1)$ be an error parameter. 
    There exists an algorithm that outputs an $\eta$-additive approximation to $\kappa_s(X)$ in time $O_{k, \Delta, \varepsilon, s}((n/\eta)^{O_{k, \Delta, \varepsilon}(1)})$. In particular, this provides an FPTAS for the multiplicative approximation of~$\=E[X]$ and $\-{Var}[X]$.
\end{lemma}
\begin{remark}
    This is an analog of \cite[Theorem 1.7]{vishesh2022approximate}. However, they do not need a lower bound of $\lambda \ge n^{-2}$ because they are able to handle small $\lambda$'s differently by a cluster expansion. And for large $\lambda$'s, they use a polynomial interpolation. So, their algorithm does not require a lower bound of $\lambda$ and runs in linear time in $n$. We claim that the regime we considered is sufficient to prove \Cref{theorem:fptas-exact-k}.  We will see this fact in the proof of \Cref{lemma:proper-t-finding}.
\end{remark}
Now, we describe the high-level ideas of the proof. First, we can express cumulants by the derivatives of the log-partition function. By the standard tool in \cite{barvinok2016combinatorics} and the fundamental theorem of algebra, one can express the derivatives of the log-partition function by the combination of inverse power series. For example, let $Z(x) = \prod_{i=1}^n \tp{1 - r_i x}$ where $r_1, r_2, \ldots, r_n$ donates the inverse roots of $Z$. So we have that,
\[
\frac{d^s \log Z(x)}{d x^s} = -(s-1)!\sum\limits_{i=1}^n r_i^s \sum_{a=0}^{\infty} \binom{s-1+a}{s-1} (r_i x)^a.
\]
Then we can truncate the series and calculate the truncated series using the approximate inverse power series derived by \cite{patel2017deterministic,Liu2017TheIP}. This finishes the proof of the additive approximation of the cumulant. For the FPTAS for $\=E[X]$ and $\-{Var}[X]$, we need to convert the additive error to a multiplicative error. So, it suffices to show that $\=E[X]=\Omega_{k, \Delta, \varepsilon}(\lambda n)$ and $\-{Var}=\Omega_{k, \Delta, \varepsilon}(\lambda n)$. To see this, note that $\lambda \ge n^{-2}$, if we want a $\eta'$-relative approximation, we can set $\eta = \eta' n^{-1}$. And these two lower bounds follow from \Cref{corollary:occupancy-fraction-lower-bound} and  \Cref{corollary:univariate-variance-bound}.
% by the linearity of the expectation, we have that
% \[ \=E[X] = \sum\limits_{v\in V} \=P[v \text{ is occupied}]. \]
% Fix $v\in V$, for each edge $e_i$ that contains $v$, we pick exactly one distinct vertex $v_i\neq v$ from each edge. Denote these vertices by $v_1, v_2, \ldots, v_m$. Also note that for every vertex, the probability of being unoccupied is lower bounded by $\frac{1}{1 + \lambda}$. So $\=P[v \text{ is occupied}] \ge \frac{\lambda}{(1+\lambda)^{m + 1}}\ge \frac{\lambda}{(1+\lambda)^{\Delta + 1}}$. Noting that $\lambda\le \lambda_{c, \varepsilon} = O_{k,\Delta,\varepsilon}(1)$, we have $\=E[X] = \Omega_{k, \Delta, \varepsilon}(\lambda n)$. 

Next, we present a lemma for finding an appropriate external field.
\begin{lemma}\label{lemma:proper-t-finding}
Fix $k\ge 2$, $\Delta\ge 3$. Let $H = (V, \+E)$ be a $k$-uniform hypergraph with maximum degree $\Delta$. Let $n = \abs{V}$. 
Fix any $\varepsilon\in\left(0,\frac{1}{9k^5\Delta^2}\right)$. Let $\lambda_{c, \varepsilon}$ be defined as in \Cref{theorem:zero-freeness-his-lee-yang}. Then there exists a constant $\zeta = \zeta_{k, \Delta, \varepsilon}$. For $t\in \=Z_{\ge 1}$ and $t\le n \alpha_H(\lambda_{c, \varepsilon})$, there exists an integer $s \in \{1, 2, \ldots, 2\zeta n\}$ such that
\[
\abs{n \alpha_H(s/(2 \zeta n)) - t} \le 1/2.
\]
Furthermore, for a given $t$, there exists an algorithm to find such $s$ that runs in time $O_{k, \Delta, \varepsilon}(n^{O_{k, \Delta, \varepsilon}(1)})$.
\end{lemma}
\begin{proof}
For any $\lambda \in (0, \lambda_{c, \varepsilon}]$, let $I_{\lambda} \sim \mu_{H, \lambda}$, $X_{\lambda}=\abs{I_{\lambda}}$.
By a standard calculation, we have that $\alpha_H(\lambda) = \frac{1}{\abs{V}} \frac{\lambda Z_H'(\lambda)}{Z_H(\lambda)}$. By \Cref{corollary:occupancy-fraction-lower-bound}, it holds that
\[
\frac{\partial}{\partial \lambda}\alpha_H(\lambda) = \frac{1}{\lambda n} \-{Var}[X_{\lambda}] = \Theta_{k, \Delta, \varepsilon}(1). 
\] Therefore, we set $\zeta = \max_{\lambda \in (0, \lambda_{c, \varepsilon}]}\frac{\partial}{\partial \lambda}\alpha_H(\lambda) = \Theta_{k, \Delta, \varepsilon}(1)$. By the above calculation, we know that $\alpha_H(\lambda)$ is nondecreasing. And we see that $n\alpha_H(\lambda)$ increases at most $1/2$ over an internal length $1/(2\zeta n)$. So for any $t\le n \alpha$, there must be an integer $s \in \{1, 2, \ldots, 2 \zeta n\}$ such that $\abs{n \alpha_H(s/(2\zeta n)) - t}\le 1/2$. 

Note $n\alpha_H(\lambda) = \=E_\lambda(X_\lambda)$ and recall that for $\lambda \in (n^{-2}, \lambda_{c, \varepsilon}]$, we can use \Cref{lemma:deterministic-cumulant} to compute $\=E_\lambda(X_\lambda)$ with $1/4$ additive error by setting $\eta = 1/4$. So we can use this method to find such $s$.
\end{proof}

The next lemma is an analog of \cite[Lemma 4.3]{vishesh2022approximate}, which shows that for low Fourier phases, $\=E\inbr{\mathrm{e}^{itY}}$ with small $t$'s, we can approximate them by the interpolation method.
\begin{lemma}\label{lemma:low-Fourier-phases-approx}
    Fix $k\ge2$, $\Delta\ge 3$ and a parameter $C \ge 1$. Let $H=(V, \+E)$ be a $k$-uniform hypergraph with maximum degree $\Delta$. Let $n = \abs{V}$. Let $\varepsilon$, $\lambda_{c, \varepsilon}$ be defined as in \Cref{theorem:zero-freeness-his-lee-yang}. For any $\lambda \in [n^{-2}, \lambda_{c, \varepsilon}]$, let $I\sim \mu_{H, \lambda}$, $X=\abs{I}$, $\bar{\mu}=\=E[X]$, $\sigma^2 = \-{Var}[X]$ and $Y=(X-\bar{\mu})/\sigma$. There exists a deterministic algorithm that, on input $H$, $\lambda$, an error parameter $\eta \in (0, 1/\sqrt{n})$, and $t\in \inbr{-C\sqrt{\log 1/\eta}, C\sqrt{\log 1/\eta}}$ outputs an $\eta$-relative, $\eta$-additive approximation to $\=E\inbr{\mathrm{e}^{i t Y}}$ in time $(n/\eta)^{O_{k, \Delta, \varepsilon, C}(1)}$.
\end{lemma}
Now we describe the proof of \Cref{lemma:low-Fourier-phases-approx}. We first express the characteristic function by the partition function.
\[
\=E\inbr{\mathrm{e}^{itY}} = \mathrm{e}^{-it\bar{\mu}/\sigma} \cdot \=E\inbr{\mathrm{e}^{itX/\sigma}} = \mathrm{e}^{-it\bar{\mu}/\sigma}\cdot \frac{Z_H(\lambda \mathrm{e}^{it/\sigma})}{Z_H(\lambda)}.
\]
So, if we can approximate the partition function, $\bar{\mu}$ and $\sigma$, then we can approximate $\=E\inbr{\mathrm{e}^{itY}}$. We can use \Cref{lemma:deterministic-cumulant} to approximate $\bar{\mu}$ and $\sigma$. For the partition function, we use a polynomial interpolation.

For $\lambda\le \frac{1}{10\Delta}$, for any $t\in \=R$, $\lambda \mathrm{e}^{it/\sigma}$ is in the zero-free region due to \cite{galvin2024zeroes}.

For $\lambda > \frac{1}{10\Delta}$, by \Cref{corollary:univariate-variance-bound}, we have $\sigma^2 = \Omega_{k, \Delta, \varepsilon}(n)$. We may assume that for $\abs{t}\le C\sqrt{\log 1/\eta}$, $\lambda \mathrm{e}^{it/\sigma}$ is in $\+D_{\varepsilon}$, otherwise it follows that $\eta^{-1} = \exp(\Omega_{k, \Delta, \varepsilon, C}(n))$, so that the exhaustive enumeration runs in the claimed time.

Finally, we are ready to give a proof sketch of \Cref{theorem:stronger-fptas-exact-k}. Given all the ingredients of zero-freeness and local CLT, our proof closely follows that of~\cite{vishesh2022approximate} and we will not repeat formally.

% \todo{yyx: why this argument do not need consider $\sigma<1$? fixed, text is needed to be changed too.}
% We now describe the high-level ideas of its proof. 
First, we use \Cref{lemma:proper-t-finding} to find $\lambda^*$ such that $\abs{n\alpha_H(\lambda^*) - t} \le 1/2$. Let $\bar{\mu} = \=E[X]$, $\sigma^2 = \-{Var}[X]$, note that we only need to consider the case satisfying $\sigma > c > 1$ where $c$ is a constant. 
To see this, we argue that if $\sigma\le c$, then $t$ must also be an absolute constant, and we can simply enumerate all the independent sets of size $t$ using brute-force in time $O(n^t)$. By \Cref{corollary:univariate-variance-bound}, we have $\sigma^2 = \Theta_{k, \Delta, \varepsilon}(\lambda n)$, so there exists a constant $\zeta_{k, \Delta, \varepsilon}$ such that $\sigma \le c$ means $\lambda \le \zeta_{k, \Delta, \varepsilon}/n$. 
Also, by the proof of \Cref{lemma:proper-t-finding}, we know that by changing $\lambda$ from $0$ to $\zeta_{k, \Delta, \varepsilon}/n$, $n \alpha_H(\lambda)$ only changes by a constant. Therefore, the corresponding $t$'s satisfy $t \le \theta_{k, \Delta, \varepsilon}$ for some  constant $\theta_{k, \Delta, \varepsilon}$. As such, we can use brute-force enumeration to calculate $i_t(H)$ in time $O(n^t)$. 

Henceforth, we assume $\sigma > c > 1$.
By \Cref{theorem:clt-part2}, \Cref{corollary:univariate-variance-bound} and the fact that $\sigma > c$, we have $\=P[X = t] = \Omega_{k,\Delta,\varepsilon}(1/\sqrt{\lambda^* n}) = \Omega_{k,\Delta,\varepsilon}(1/\sigma)$. Recall \eqref{equation:exact-t-prob}, it suffices to approximate $\=P[X = t]$ and $Z_H(\lambda^*)$ in order to approximate $i_t(H)$. Let $Y = (X - \bar{\mu})/\sigma$ and $y = (t - \bar{\mu})/\sigma$. By the Fourier inversion formula, 
\[
\=P[X = t] = \=P[Y=y] = \frac{1}{2\pi\sigma}\int_{-\pi\sigma}^{\pi\sigma} \=E\inbr{\mathrm{e}^{itY}} \mathrm{e}^{-ity} dt.
\]
By \Cref{lemma:characterized-func-bound}, let $\gamma = \min(\pi \sigma, C_{k,\Delta,\varepsilon} \sqrt{\log 1/\eta})$ where $C_{k,\Delta,\varepsilon}$ is a sufficiently large constant depending on $k, \Delta, \varepsilon$, we know that the contribution of high Fourier phases is negligible compared to $\=P[X = t] = \Omega_{k,\Delta,\varepsilon}(1/\sqrt{\lambda^* n}) = \Omega_{k,\Delta,\varepsilon}(1/\sigma)$. So we know that,
\[
\=P[Y=y] \approx  \frac{1}{2\pi\sigma}\int_{-\gamma}^{\gamma} \=E\inbr{\mathrm{e}^{itY}} \mathrm{e}^{-ity} dt.
\]
Then we replace the integration with summation. For $\zeta = o(\eta)$, we have that
\[
\=P[Y=y] \approx  \frac{\zeta}{2\pi\sigma}\sum\limits_{t = -\gamma \zeta^{-1}}^{\gamma \zeta^{-1}} \=E\inbr{\mathrm{e}^{it\zeta^{-1} Y}} \mathrm{e}^{-i\zeta^{-1}ty}.
\]
For the characteristic function, we use \Cref{lemma:low-Fourier-phases-approx} for approximation. To approximate $\bar{\mu}$ and $\sigma$, we apply \Cref{lemma:deterministic-cumulant}. These allow us to approximate $\=P[Y=y]$. Recalling \eqref{equation:exact-t-prob}, the missing part to approximate $i_t(H)$ is $Z_H(\lambda^*)$, which can be derived by the polynomial interpolation.
\end{document}